\documentclass[11p,reqno]{amsart}
\usepackage[colorlinks=true, pdfstartview=FitV, linkcolor=blue, citecolor=blue, urlcolor=blue]{hyperref}
\usepackage[usenames,dvipsnames]{xcolor}
\usepackage{latexsym}
\usepackage{amssymb}
\usepackage{amscd}
\usepackage{bm}
\usepackage{amsmath,amsthm}
\usepackage{setspace}
\usepackage{amssymb,amsthm,amsmath,mathrsfs,bm}
\usepackage{graphicx}
\usepackage{tabularx}
\usepackage{subfigure}
\usepackage{appendix}
\usepackage{enumerate}
\usepackage{cite}

\usepackage{tikz}
\usetikzlibrary{positioning, arrows.meta, shapes.geometric}

\newcommand{\be}{\begin{eqnarray}}
\newcommand{\ee}{\end{eqnarray}}
\newcommand{\bez}{\begin{eqnarray*}}
\newcommand{\eez}{\end{eqnarray*}}

\numberwithin{equation}{section}

\numberwithin{equation}{section}
\newtheorem{theorem}{Theorem}[section]
\newtheorem{definition}{Definition}[section]
\newtheorem{lemma}[theorem]{Lemma}
\newtheorem{coro}[theorem]{Corollary}

\theoremstyle{definition}

\newtheorem{remark}[theorem]{Remark}

\DeclareMathOperator{\Pf}{Pf}

\allowdisplaybreaks[4]
\begin{document}

\title[Nonisospectral deformation of OPs and Painlev\'{e}]{Stationary reduction method based on nonisospectral deformation of orthogonal polynomials, and discrete Painlev\'{e}-type equations}

\date{}

\author{Xiao-Lu Yue}
\address{{Department of Mathematics, China University of Mining and Technology, Beijing, China;} 
Department of Mathematics, City University of Hong Kong, Tat Chee Avenue, Kowloon, Hong Kong; Academy of Mathematics and Systems Science, Chinese Academy of Sciences, Beijing 100190, P.R. China, and School of Mathematical Sciences, University of Chinese Academy of Sciences, Beijing 100049, P.R. China.
}
\email{{B202617@cumtb.edu.cn}; yuexiaolu@lsec.cc.ac.cn
}

\author{Xiang-Ke Chang}
\address{SKLMS \& ICMSEC, Academy of Mathematics and Systems Science, Chinese Academy of Sciences, Beijing 100190, P.R. China, and School of Mathematical Sciences, University of Chinese Academy of Sciences, Beijing 100049, P.R. China.}
\email{changxk@lsec.cc.ac.cn}

\author{Xing-Biao Hu}
\address{{School of Arts and Science, Fuyao University of Science and Technology, Fuzhou, China}; ICMSEC, Academy of Mathematics and Systems Science, Chinese Academy of Sciences, Beijing 100190, P.R. China, and School of Mathematical Sciences, University of Chinese Academy of Sciences, Beijing 100049, P.R. China.}
\email{hxb@lsec.cc.ac.cn}

\begin{abstract}
In this work, we propose a new approach called ``stationary reduction method based on nonisospectral deformation of orthogonal polynomials" for deriving discrete Painlev\'{e}-type (d-P-type) equations. We apply this approach to  (bi)orthogonal polynomials satisfying ordinary orthogonality, $(1,m)$-biorthogonality, generalized Laurent biorthogonality, Cauchy biorthogonality and partial-skew orthogonality.  As a result, several seemingly novel classes of high order d-P-type equations or integrable difference systems with potential relations with new d-P-type equations, along with their particular solutions and respective Lax pairs, are derived. Notably, the derived integrable difference system related to the Cauchy biorthogonality is a stationary reduction of a nonisospectral generalization involving the first two flows of the Toda hierarchy of CKP type. Additionally, the integrable difference system related to the partial-skew orthogonality is associated with the nonisospectral Toda hierarchy of BKP type, and it is found to admit a solution expressed in terms of Pfaffians.



\end{abstract}



\maketitle
\tableofcontents

\section{Introduction} 

\subsection{On OPs and their analogues}

An orthogonal polynomial (OP) sequence is a family of polynomials such that any two different polynomials in the sequence are orthogonal to each other under some inner product, among which, the classical OPs, such as Jacobi, Laguerre and Hermite  OPs, are the most widely used. Since the 19th century, the theory of OPs has been well developed; see e.g. \cite{askey1975orthogonal,askey1985some,beals2016special,brezinski1980pade,chihara1978introduction,ismail2009classical,szeg1939orthogonal,baik2007discrete,marcellan2021orthogonal,simon2009orthogonal}.

The theory of OPs encompasses two distinct but interconnected facets. The two facets have many things in common, and the
division line is quite blurred, it is more or less along algebra vs. analysis.
The first aspect pertains to the formal and algebraic elements of the theory, which establishes strong links with special functions, combinatorics, and algebra. The exploration of broader classes of OPs using mathematical analysis techniques constitutes the other facet of the theory. In this regard, the primary inquiries revolve around the asymptotic properties of the polynomials and their zeros, the reconstruction of the orthogonality measure, and related aspects. 


Nowadays, there have been some more general definitions of OPs \cite{brezinski1992biorthogonality,dunkl2014orthogonal,bultheel1999orthogonal,stahl1992general,marcellan2021orthogonal}, such as matrix OPs, multiple OPs, multivariable OPs , bi-OPs, rational orthogonal functions etc., and some analogues of orthogonality, such as skew-orthogonality, partial skew-orthogonality etc. (Unless otherwise specified, polynomials or functions with certain orthogonality are collectively referred to as OPs for simplicity.) In particular, new concepts called \textit{Cauchy bi-OPs} \cite{bertola2010cauchy}, \textit{partial-skew-OPs} \cite{chang2018partial} and \textit{generalized Laurent bi-OPs} \cite{wang2022generalization} were proposed very recently and the corresponding theories and applications have been increasingly developed (see e.g. \cite{bertola2009cubic,bertola2009cauchy,bertola2010cauchy,bertola2014cauchy,bertola2013strong,gharakhloo2023modulated,ito2023generalized,chang2021two,forrester2021fox,li2019cauchy,chang2022hermite,lago2019mixed,bertola2014universality,chang2018degasperis,chang2018application,fidalgo2020asymptotic,gonzalez2022strong}).

OPs have emerged with great importance in the fields of mathematical physics, quantum mechanics, numerical analysis, statistics, probability, and many other disciplines.
In recent decades, considerable attention has been paid to interdisciplinary studies between the theory of OPs and integrable systems. For example, OPs can serve as wave functions appearing in the Lax pairs of Toda-type lattices \cite{aptekarev1997toda,ismail2009classical,van2022orthogonal,adler1995matrix,adler1997string,adler1999generalized,alvarez2013orthogonal,alvarez2011multiple,ariznabarreta2016multivariate,chang2018partial,kharchev1997faces,mukaihira2002schur,nenciu2005lax,peherstorfer2007toda,vinet1998integrable,deift2000orthogonal}, while the compatibility conditions of discrete spectral transformations of OPs can give rise to discrete integrable systems \cite{aptekarev2016toda,aptekarev2016discrete,chang2015about,miki2011discrete,papageorgiou1995orthogonal,spiridonov2007integrable,spiridonov1995discrete,spiridonov1997discrete,tsujimoto2000discrete}. The theory of OPs can play important roles in the study of peaked soliton problems for a class of integrable partial differential equations \cite{beals2000multipeakons,beals2001peakons,lundmark2016inverse,chang2016multipeakons,chang2018degasperis,chang2019isospectral,lago2019mixed,chang2022hermite,lundmark2022view}. In addition, semi-classical OPs will lead to Painlev\'{e}-type equations \cite{vanass2018,van2022orthogonal}.

{It is noted that nonisospectral integrable systems have also been of considerable interest. Here a nonisospectral integrable system means an integrable system possessing a nonisospectral Lax pair, that is, the spectral parameter in the Lax pair is time-dependent, as opposed to being a constant. In fact, nonisospectral deformations of OPs can be used to derive Toda-type lattices;  see e.g.\cite{berezansky1994nonisospectral,chen2015three,Chen2018Non,levi1992non} etc. Recently, this idea has been employed to study the so-called peakon lattices, yielding some new nonisospectral Camassa--Holm type equations \cite{chang2014generalized,chang2018moment,chang2024peakon}.}
This work is devoted to an interdisciplinary study on nonisospectral deformation of OPs and integrable Painlev\'{e}-type equations.

\subsection{On Painlev\'{e} equations}
Around the early 20th century, Painlev\'{e}, Gambier and Fuchs investigated the problem, proposed by Picard, related to second-order ordinary differential equations of the following form
\begin{align}
\frac{d^2y}{dx^2}=F\left(\frac{dy}{dx},y,x\right),\label{intro0-3}
\end{align}
where $F$ is a rational function of $\frac{dy}{dx}$ and $y$, and is analytic in terms of $x$. The objective is to classify the differential equation \eqref{intro0-3} under the condition that all movable singularities of the solution are poles. They discovered that, up to a M\"{o}bius transformation, there were fifty equations of the form \eqref{intro0-3} that possess this property, now known as the \textit{Painlev\'{e} property}.

Painlev\'{e}, Gambier and Fuchs further demonstrated that, among these fifty equations, forty-four of them can be reduced to linear equations that can be solved using elliptic functions or previously known special functions such as Airy functions or Bessel functions. The remaining six equations, which cannot be reduced to linear form, give rise to new nonlinear ordinary differential equations that define new transcendental functions. These equations are known as the Painlev\'{e} equations \cite{clarkson2003painleve,conte2012painleve,conte2008painleve,fokas2006painleve,noumi2004painleve,iwasaki2013gauss,ince1956ordinary,levi2013painleve,vanass2018}:
\begin{align*}
\text{P}_{\text{I}}:\quad \frac{d^2 y}{dx^2}=&6y^2+x, \\
\text{P}_{\text{II}}: \quad\frac{d^2 y}{dx^2}=&2y^3+xy+\alpha, \\
\text{P}_{\text{III}}: \quad \frac{d^2 y}{dx^2}=&\frac{(\frac{dy}{dx})^2}{y}-\frac{\frac{dy}{dx}}{x}+\frac{\alpha y^2+\beta}{x}+\gamma y^3+\frac{\delta}{y}, \\
\text{P}_{\text{IV}}: \quad \frac{d^2 y}{dx^2}=&\frac{(\frac{dy}{dx})^2}{2y}-\frac{3y^3}{2}+4xy^2+2(x^2-\alpha)y+\frac{\beta}{y}, \\
\text{P}_{\text{V}}: \quad \frac{d^2 y}{dx^2}=&\left(\frac{1}{2y}+\frac{1}{y-1}\right)\left(\frac{dy}{dx}\right)^2-\frac{\frac{dy}{dx}}{x}+\frac{(y-1)^2}{x^2}\left(\alpha y+\frac{\beta}{y}\right)+\frac{\gamma y}{x}+\frac{\delta y(y+1)}{y-1},\\
\text{P}_{\text{VI}}: \quad \frac{d^2 y}{dx^2}=&\frac{1}{2}\left(\frac{1}{y}+\frac{1}{y-1}+\frac{1}{y-x}\right)\left(\frac{dy}{dx}\right)^2-\left(\frac{1}{x}+\frac{1}{x-1}+\frac{1}{y-x}\right)\frac{dy}{dx}\\
&+\frac{y(y-1)(y-x)}{x^2(x-1)^2}\left(\alpha+\frac{\beta x}{y^2}+\frac{\gamma(x-1)}{(y-1)^2}+\frac{\delta x(x-1)}{(y-x)^2}\right),
\end{align*}
whose solutions are called the Painlev\'e transcendents. 
Here, $\alpha$, $\beta$, $\gamma$, and $\delta$ are all constants. It is also noted that each of the Painlev\'e equations can be written as a (non-autonomous) Hamiltonian system for a suitable Hamiltonian function $H_J(q,p,z)$. The function $\sigma(z) \equiv H_J(q,p,z)$ satisfies a second-order, second-degree ordinary differential equation, known as the Jimbo–Miwa–Okamoto equation or Painlev\'e $\sigma$-equation, whose solution is expressible in terms of the solution of the associated Painlev\'e  equation \cite{jimbo1981monodromy}.

Although the Painlev\'{e} equations were initially regarded as purely mathematical objects, they also widely appear in various studies of physical systems. For instance, P$_{\text{II}}$ is related to the Korteweg-de Vries (KdV) equation \cite{segur1981asymptotic,claeys2010painleve,fokas1981linearization,ablowitz1978nonlinear}, which models shallow water waves, and it also appears in the long time asymptotics for the the mKdV equation \cite{deift1993steepest}. The Painlev\'{e} equations can also arise in the solutions of the nonlinear Schrödinger equation \cite{bertola2013universality,boscolo2002self,wang2023defocusing} and appear in the study of the Camassa-Holm equation \cite{barnes2022similarity,de2010painleve}. Additionally, the Painlev\'{e} equations have applications in many other fields \cite{forrester2002application,ablowitz1991soliton,jimbo1980density,bertola2015asymptotics,balogh2015strong,bertola2024exactly,baik1999distribution,fokas2006painleve,forrester2010log,tracy1994level,mehta2004random,levi2013painleve}, including OPs, statistical mechanics, random matrices, plasma physics, quantum gravity, general relativity, and etc.

In recent years, there has been increasing interest in discrete Painlev\'{e} (d-P) equations \cite{joshi2019discrete,conte2012painleve,conte2008painleve,kajiwara2017geometric,grammaticos1991integrable,ramani1991discrete,vanass2018,van2022orthogonal,noumi2004painleve,grammaticos2004discrete,cassatella2012riemann,gordoa2010matrix,dzhamay2018some,joshi2023monodromy,hietarinta2016discrete,joshi2021riemann,sakai2001rational}. d-P equations are nonlinear, non-autonomous, second-order ordinary difference equations that tend to continuous Painlev\'{e} equations in a certain limit.  The systematic study of d-P equations started
in the early 1990s in the work of Grammaticos, Ramani, who first, together with Papageorgiou \cite{grammaticos1991integrable}, introduced the notion of \textit{singularity confinement} as a discrete counterpart of the Painlevé property and proposed to use it as an integrability detector for discrete systems, then, together with Hietarinta \cite{ramani1991discrete}, applied this idea to obtain non-autonomous version of the two-dimensional integrable mappings
known as the Quispel–Roberts–Thompson (QRT) mappings  and succeeded in constructing d-P equations systematically.

Based on whether the coefficients are linear, exponential, or elliptic functions of $n$, d-P equations can be classified into three-types \cite{kajiwara2017geometric,joshi2019discrete,ramani1991discrete,conte2008painleve}, denoted by the prefixes $d-$, $q-$, or $ell-$ respectively, before the equation names. For example, there exist 
\begin{align*}
\text{d-P}_{\text{I}}: \quad &x_{n+1}+x_n+x_{n-1}=\frac{z_n+a(-1)^n}{x_n}+b,\\
\text{d-P}_{\text{II}}: \quad &x_{n+1}+x_{n-1}=\frac{x_nz_n+a}{1-x_n^2},\\
\text{d-P}_{\text{IV}}: \quad &(x_{n+1}+x_n)(x_n+x_{n-1})=\frac{(x_n^2-a^2)(x_n^2-b^2)}{(x_n+z_n)^2-c^2},\\
\text{d-P}_{\text{V}}: \quad &\frac{(x_{n+1}+x_n-z_{n+1}-z_n)(x_n+x_{n-1}-z_n-z_{n-1})}{(x_{n+1}+x_n)(x_n+x_{n-1})}\\=
&\frac{((x_n-z_n)^2-a^2)((x_n-z_n)^2-b^2)}{(x_n-c^2)(x_n-d^2)},
\end{align*}
where $z_n=\alpha n+\beta$, $a,\ b,\ c,\ d$ are all constants, and
\begin{align*}
\text{q-P}_{\text{III}}: \quad &x_{n+1}x_{n-1}=\frac{(x_n-aq_n)(x_n-bq_n)}{(1-cx_n)(1-x_n/c)},\\
\text{q-P}_{\text{V}}: \quad &(x_{n+1}x_n-1)(x_{n-1}x_n-1)=\frac{(x_n-a)(x_n-1/a)(x_n-b)(x_n-1/b)}{(1-cx_nq_n)(1-x_nq_n/c)},\\
\text{q-P}_{\text{VI}}: \quad &\frac{(x_{n+1}x_n-q_nq_{n+1})(x_{n-1}x_n-q_nq_{n-1})}{(x_{n+1}x_n-1)(x_{n-1}x_n-1)}=\frac{(x_n-aq_n)(x_n-q_n/a)(x_n-bq_n)(x_n-q_n/b)}{(x_n-c)(x_n-1/c)(x_n-d)(x_n-1/d)},
\end{align*}
where $q_n=q_0q^ n+\beta$, $a,\ b,\ c,\ d$ are all constants. An example of a scalar elliptic d-P equation is
\begin{align*}
&cn(\gamma_n)dn(\gamma_n)(1-k^2sn^4(z_n))x_n(x_{n+1}+x_{n-1})\\
&-cn(z_n)dn(z_n)(1-k^2sn^2(z_n)sn^2(\gamma_n))(x_{n+1}x_{n-1}+x_n^2)\\
&+(cn^2(z_n)-cn^2(\gamma_n))cn(z_n)dn(z_n)(1+k^2x^2_nx_{n+1}x_{n-1})=0,
\end{align*}
in which 
\begin{align*}
z_n=(\gamma_e+\gamma_0)n+\omega,\quad \gamma_n=\left\{
		\begin{aligned}
			&\gamma_e \quad n=2j,\\
			&\gamma_0 \quad n=2j+1,\nonumber
		\end{aligned}\right.
\end{align*}
and $cn$, $dn$, $sn$ are all Jacobi elliptic functions.

It is Sakai that gave the definite classification scheme of d-P equations based on algebro-geometric ideas \cite{sakai2001rational}. There is also certain connection between d-P-type equations and birational representation of affine Weyl groups \cite{noumi2004painleve,noumi1998affine}. For a comprehensive survey of the geometric aspects of d-P equations, see the article by Kajiwara, Noumi and Yamada \cite{kajiwara2017geometric} together with the references therein.

At the end of this subsection, we mention that d-P equations also appear in some applied problems \cite{borodin2003discrete,borodin2003distribution,fokas1991discrete,borodin2016lectures,hone2014discrete,hone2002lattice,fokas1992isomonodromy}, such as the computations of gap probabilities of various ensembles in the emerging field of integrable probability, quantum gravity, and reductions of lattice equations etc.

\subsection{Connections between Painlev\'{e} equations and OPs}
It is known that there exists close connections between the (discrete and continuous) Painlev\'{e} equations and OPs. The relationship between Painlev\'{e} equations and OPs can be dated back to the work of Shohat \cite{shohat1939differential} in 1939 and later Freud \cite{freud1976coefficients} in 1976. However the equations in their works were not identified as d-P equations until the work of Fokas, Its and Kitaev \cite{fokas1991discrete,fokas1992isomonodromy} in the early 1990s.  Later, Magnus  demonstrated certain relationship between semi-classical OPs and the (continuous) Painlev\'e  equations \cite{magnus1995painleve}. Nowadays, a multitude of connections between Painlev\'{e} equations and OPs have been discovered; see a recent monography by Van Assche \cite{vanass2018} and references therein or some others e.g. \cite{clarkson2006painleve,clarkson2013recurrence,van2007discrete,boelen2010discrete,chen1997ladder,chen2008ladder,clarkson2014relationship,dai2010determinants,forrester2004discrete,boelen2013generalized,chen2010painleve,yue2022laurent,forrester2006bi,xu2015painleve,cassatella2014singularity,cassatella2019riemann,dzhamay2020recurrence,van2022orthogonal,bertola2015asymptotics,balogh2015strong,bertola2024exactly}. The relationship can be summarized as follows:

\begin{enumerate}
    \item[(1)] d-P equations are satisfied by the recurrence coefficients of certain semi-classical OPs.

\item[(2)] The recurrence coefficients of OPs undergoing a Toda-type evolution satisfy Painlev\'{e} differential equations, and their special solutions are associated with special functions such as Airy functions, Bessel functions, (confluent) hypergeometric functions, and parabolic cylinder functions. 

\item[(3)] Expressions for rational solutions of some Painlev\'{e} equations can be formulated utilizing Wronskians of certain OPs.

\item[(4)] Special transcendental solutions of Painlev\'{e} equations are often used to establish the local asymptotics of OPs at critical points.

\end{enumerate}

 Motivated by the fact that the cross-research on OPs and Painlev\'{e} equations has promoted the mutual development of both fields, we are interested in the interdisciplinary studies of OPs and d-P equations. Specifically, we are curious about what types of Painlev\'{e} equations are related to the recently proposed Cauchy bi-OPs \cite{bertola2010cauchy}, partial-skew-OPs \cite{chang2018partial} and generalized Laurent bi-OPs \cite{wang2022generalization}.

 To this end, we propose a new method called \textit{stationary reduction method based on nonisospectral deformation of OPs} for deriving d-P-type equations in Section \ref{BM}. Subsequently, \textbf{ this approach are applied to various OPs, including ordinary OPs, $\mathbf{(1,m)}$-type bi-OPs, generalized Laurent bi-OPs, Cauchy bi-OPs and partial-skew OPs, in Section \ref{OOP}-\ref{Partial-skew}. As a result, some seemingly new  high order  d-P-type equations are obtained}.  It is noted that the obtained new integrable difference systems from partial-skew OPs and Cauchy bi-OP, which might be related to new  d-P-type equations of  high order, are associated to nonisospectral Toda hierarchies of BKP and CKP types, respectively. In particular, the obtained integrable difference system from partial-skew OPs enjoys solutions in terms of Pfaffians. Section \ref{cd} is devoted to conclusion and discussions.
Some technical details involving Pfaffians are included in the appendices.

It is also noted that the concept of Cauchy bi-OPs was first proposed in \cite{bertola2010cauchy} and they are motivated by the multipeakon solutions of a shallow water wave equation called the Degasperis--Procesi equation \cite{lundmark2003multi,lundmark2005degasperis}. So far, the Cauchy bi-OPs have been extensively investigated in the literature, including the inspired Cauchy two-matrix model,  Toda lattice of CKP-type (C-Toda), mixed Hermite--Pad\'e approximation problem, as well as their generalizations etc. (see e.g. \cite{bertola2009cauchy,bertola2009cubic,bertola2010cauchy,bertola2014cauchy,bertola2013strong,chang2025isospectral,chang2018degasperis,chang2021two,li2019cauchy,forrester2021fox,lago2019mixed,bertola2014universality,gonzalez2022strong}). In fact, isospectral deformation of the Cauchy bi-OPs \cite{chang2018degasperis} can lead to the C-Toda lattice, which is also shown to be related  with the discrete CKP equation \cite{bobenko2017discrete,fu17}. It is noted that  Krichever and Zabrodin introduced the so-called \textit{constrained Toda (C-Toda) hierarchy} as a certain subhierarchy of the 2D Toda lattice in the recent work  \cite{krichever2022constrained}. It remains unclear on the difference between these two ``C-Toda" lattices. In Section \ref{Cauchy}, we will clarify that they coincide with each other.

The concept of partial-skew-OPs was introduced in \cite{chang2018partial}, with the motivation by the study of a random matrix model called the Bures random ensemble \cite{forrester2016relating} as well as the hints from the formulation of the Novikov peakon solution in terms of Pfaffians \cite{chang2018application}. It is shown that isospectral deformation of the partial-skew OPs are closely related to integrable lattices including Toda lattice of BKP-type (B-Toda) \cite{chang2025isospectral,chang2018partial} and they also solve certain mixed Hermite--Pad\'e approximation problem \cite{chang2022hermite}. It is also noted that,  in a recent paper \cite{krichever2023toda},  Krichever and Zabrodin introduced the so-called \textit{Toda lattice with constraint of type B} as a certain subhierarchy of the 2D Toda lattice. A natural question is whether these two ``B-Toda'' lattices are related with each other, which will be clarified in Section \ref{Partial-skew}.

\section{Stationary reduction method based on nonisospectral deformation of OPs}\label{BM}
In the literature, there exist some methods for deriving d-P equations 
(see e.g. \cite{grammaticos1991integrable,fokas1993continuous,grammaticos2004discrete,grammaticos1999discrete,conte2012painleve,levi1992non,vanass2018}),
among which two effective methods are the \textit{compatibility method based on orthogonality} and the \textit{approach according to stationary reduction of nonisospectral flow}.

The compatibility method based on orthogonality \cite{vanass2018} typically involves the following steps. First one needs to construct a semi-classical weight function to derive a structure relation with the help of the Pearson equation. Then one can derive a system of difference equations by using the compatibility condition of the recurrence relation and the structure relation. Finally, a d-P equation can be obtained by eliminating a number of variables. This method has been successfully applied to some well-known orthogonalities so that different d-P equations have been derived. \textbf{However, this method fails for some novel orthogonalities, such as partial-skew orthogonality, Cauchy biorthogonality etc. One of the difficulties lies in introducing appropriate
semi-classical weight functions to derive the corresponding structure relations.}

The art of the approach according to stationary reduction of nonisospectral flow \cite{levi1992non} is as follows.
Consider a nonisospectral flow 
\begin{align}
\frac{d}{dt}A_{n}+A_{n}B_{n}-B_{n+1}A_{n}=0\label{lu2}
\end{align}
associated with the Lax pair
\begin{subequations}\label{Laxn}
\begin{align}
\psi_{n+1}=A_{n}({\bf q},z)\psi_n,\\
\frac{d}{dt}\psi_{n}=B_{n}({\bf q},z)\psi_n.
\end{align}
\end{subequations}
Here ${\bf q}=(q_1,q_2,\ldots,q_N)$ is a vector function in $t$ and the spectral parameter $z$ depends on the time variable $t$ satisfying 
\begin{align*}
z_t+f=0,
\end{align*}
where $f$ is a scalar function in $t$ and $z$ satisfying $f\neq 0$.
By setting
\begin{align*}
A_n=F_n,\quad B_n=-fG_n,
\end{align*}
one will obtain a d-P equation 
\begin{align*}
F_{n,z}+F_{n}G_{n}-G_{n+1}F_{n}=0,
\end{align*}
by considering the stationary equation associated with~\eqref{lu2}. In addition, when the partial derivative of $\psi_{n}$ with respect to $t$ is zero, the Lax pair \eqref{Laxn} degenerates into 
\begin{subequations}
\begin{align*}
\psi_{n+1}=F_{n}({\bf q},z)\psi_n,\\
\psi_{n,z}=G_{n}({\bf q},z)\psi_n.
\end{align*}
\end{subequations}
which gives the Lax pair of the d-P equation.
\textbf{It is evident that this method involves a direct performance of the stationary reduction on the nonisospectral equation. Unfortunately, it cannot provide the solution to the d-P equations, nor can it show the realization of the the stationary reduction from the perspective of the solution.}

Due to the limitations of these two methods, a natural question arises: Can one derive d-P equations and their solutions without constructing the semi-classical weight functions? Furthermore, it is important to understand how the stationary reduction mechanism operates in this context. To this end, we propose a new method for deriving d-P equations---\textit{the stationary reduction method based on nonisospectral deformation of OPs}. The research route is illustrated in the following figure (Fig. \ref{fig:enter-label2}).

\begin{figure}[htbp]
    \centering
    \resizebox{0.9\linewidth}{!}{ 
        \begin{tikzpicture}[
            box/.style={
                rectangle, 
                draw, 
                thick, 
                rounded corners=5pt, 
                text width=4.5cm,       
                align=center, 
                minimum height=1.6cm, 
                font=\normalsize\bfseries, 
                inner sep=8pt,
                execute at begin node={\hyphenpenalty=10000\relax} 
            },
            arrow/.style={-{Stealth[length=3mm, width=2mm]}, thick},
            node distance=1.4cm
        ]

        \node[box] (L1) {OPs with certain orthogonality};                                      
        \node[box, below=1.4cm of L1] (L2) {Time evolution of the spectral parameter and moments}; 
        \node[box, below=1.2cm of L2] (L3) {Explicit weight function};                          
        \node[box, below=1.2cm of L3] (L4) {Recurrence relation of the moments};                

        \node[box, right=3.5cm of L1, yshift=1.2cm] (M1) {Recurrence relation of OPs};          
        \node[box, below=1.2cm of M1] (M2) {Nonisospectral deformation of OPs};                 
        \node[box, below=2.5cm of M2, xshift=1cm] (M3) {Realization of stationary reduction}; 

        \node[box, right=2.6cm of M2, yshift=1.4cm] (R1) {Nonisospectral equation};             
        \node[box, below=1.8cm of R1] (R2) {Discrete \mbox{Painlevé-type} equation};             

        
        \draw[arrow, rounded corners=8pt] (L1.north) |- (M1.west);

        \draw[arrow, rounded corners=6pt] (L1.east) -- ++(1.5,0) |- (M2.west);
        \draw[arrow, rounded corners=6pt] (L2.east) -- ++(1.5,0) |- (M2.west);

        \draw[arrow, rounded corners=6pt] (L2.east) -- ++(2.2,0) |- (M3.west);
        \draw[arrow, rounded corners=6pt] (L4.east) -- ++(2.2,0) |- (M3.west);
        
        \draw[arrow, rounded corners=6pt] (M1.east) -- ++(1.5,0) |- (R1.west);
        \draw[arrow, rounded corners=6pt] (M2.east) -- ++(1.5,0) |- (R1.west);

        \draw[arrow] (L2.south) -- (L3.north);
        \draw[arrow] (L3.south) -- (L4.north);

        \draw[arrow] (R1.south) -- node[left, font=\normalsize\bfseries, align=center, xshift=-4pt] {Stationary\\reduction} (R2.north);

        \end{tikzpicture}
    }
    \caption{\small Research route }
    \label{fig:enter-label2}
\end{figure}
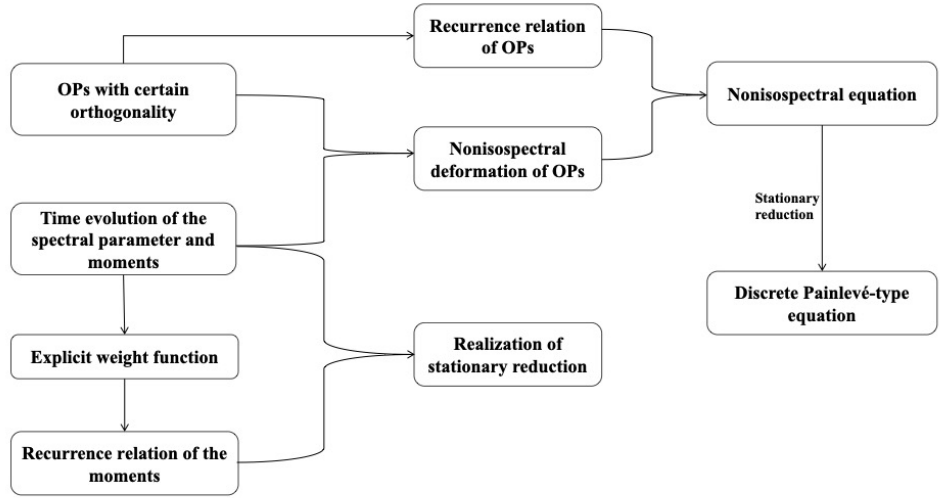

More precisely, we start from a family of OPs without giving a specific weight function and first perform nonisospectral deformation on the OPs by introducing the measure $\mu(z;t)$, where the spectral parameter $z(t)$ depends on the time variable $t$ and satisfies
\begin{align}
 \frac{d}{dt}z(t)=\alpha z(t).\label{lu3}
\end{align}
Subsequently, we utilize the compatibility condition for the deformed OPs to derive the nonisospectral integrable equations. By implementing stationary reduction, we are able to obtain d-P equations or integrable difference systems along with explicit expressions of their particular solutions. This approach enables us to rigorously demonstrate the achievability of stationary reduction. As we will see, this new method can be applied to various recently proposed new orthogonalities, including Cauchy biorthogonality, partial-skew orthogonality, etc. As a result, different classes of d-P-type equations or integrable difference systems with potential relations with new d-P-type equations are derived.

\section{Nonisospectral deformation of ordinary OPs and d-P}\label{OOP}
It is well known that, upon introducing a suitable time deformation of the measure of the OPs, one can derive the time evolution equation for the OPs by using the evolution relation of the measure $\mu(z;t)$ with respect to the time variable $t$ along with the recurrence relation and orthogonality of the OPs. By considering the compatibility condition of the recurrence relation and the time evolution relation satisfied by the OPs, one will obtain integrable equations. For example, the Toda lattice can be obtained  by considering ordinary OPs with $d\mu(z;t)=e^{tz}d\mu(z)$; the Lotka-Volterra lattice can be derived by considering ordinary OPs with symmetric measure and $d\mu(z;t)=e^{tz^2}d\mu(z)$ (see e.g. \cite{aptekarev1997toda,peherstorfer2007toda,van2022orthogonal}). Similarly, one can also derive nonisospectral integrable equations by using the compatibility condition of the recurrence relation and the time evolution relation satisfied by OPs \cite{berezansky1994nonisospectral,chang2018moment}. In this case, it is required that not only the measures $\mu(z;t)$ depends on the time variable $t$, but the spectral parameter $z$ is also related to $t$.  

In this section we shall take the nonisospectral deformation of the ordinary OPs (and those with symmetric measure), together with nonisospectral Toda and Lotka-Volterra lattices and the corresponding d-P-type equations, as examples to illustrate the process of our approach.

\subsection{Ordinary OPs}
Here the ordinary OPs denote a family of polynomials $\{P_n(z)\}_{n\in \mathbb{N}}$ satisfying the ``ordinary'' orthogonality condition
\begin{align}
\int P_n(z)P_m(z)d\mu(z)=h_n\delta_{nm},\label{3-1-1}
\end{align}
where each $P_n(z)$ is a monic polynomial of degree $n$ in $z$ and $\mu$ is a positive measure for which all the moments
\begin{align*}
c_i=\int z^i d\mu(z), \quad i=0,1,2,\ldots
\end{align*}
exist. In many cases, \eqref{3-1-1} can be written as
\begin{align}
\int P_n(z)P_m(z)w(z)dz=h_n\delta_{nm},\nonumber
\end{align}
where $w(z)$ is the corresponding weight function. From the orthogonality condition \eqref{3-1-1}, one can get the three-term recurrence relation of the form
\begin{align}
z P_n(z)=P_{n+1}(z)+b_nP_{n}(z)+a_n^2P_{n-1}(z),\label{3-1-2}
\end{align}
together with the initial values $P_{-1}(z)=0, P_0(z)=1$. The determinant expressions of the polynomials $\{P_n(z)\}_{n\in \mathbb{N}}$ and recurrence coefficients $\{a_n^2\}_{n\geq0}$ and $\{b_n\}_{n\geq0}$ can be obtained by using the orthogonality condition, that is,
\begin{align}
&P_n(z)=\frac{1}{\Delta_n}
\det\left(
\begin{array}{cccc}
c_{i+j}\\
z^j \\
\end{array}
\right)_{\substack{i=0,\ldots,n-1\\ j=0,\ldots,n\,\,\,\,\,\,\,\,}},\label{3-1-3}\\
&a_n^2=\frac{\Delta_{n+1}\Delta_{n-1}}{\Delta_{n}^2},\quad b_n=\frac{\Delta_{n+1}^{\ast}}{\Delta_{n+1}}-\frac{\Delta_{n}^{\ast}}{\Delta_{n}},
\end{align}
where the Hankel determinant denotes $\Delta_n=\det(c_{i+j})_{i,j=0,1,\ldots,n-1}$
and the determinant $\Delta_{n}^{\ast}$ is obtained from $\Delta_{n}$ by replacing the last column $(c_{n-1} , c_{n},\ldots, c_{2n-2})^\top$ by $(c_{n} , c_{n+1},\ldots  , c_{2n-1})^\top$. In addition, $\{h_n\}_{n\in \mathbb{N}}$ admit the expressions
$$\quad\quad h_n=\frac{\Delta_{n+1}}{\Delta_n}.$$

\subsection{Nonisospectral Toda lattice}
Consider the monic OPs $\{P_n(z;t)\}_{n\in \mathbb{N}}$ in \eqref{3-1-3} satisfying the three-term recurrence relation~\eqref{3-1-2}.
 Suppose that the spectral parameter $z$ evloves according to the time evolution \eqref{lu3}.  We now choose a weight function such that the moments satisfy the time evolution
\begin{align}
 \frac{d}{dt}c_j(t)=\alpha jc_j(t)+\alpha_1 c_{j+1}(t)+\alpha_2 c_{j+2}(t).\label{3-5-1}
\end{align}
Since
\begin{align*}
 \frac{d}{dt}c_j(t)=&\int  z^j\left(\alpha j w(z;t)+\frac{d}{dt}w(z;t)+\alpha w(z;t)\right)dz \nonumber \\
=&\alpha jc_j(t)+\int  z^j\left(\frac{d}{dt}w(z;t)+\alpha w(z;t)\right)dz,
\end{align*}
we have
\begin{align*}
\int  z^j\left( \frac{d}{dt}w(z;t)+\alpha w(z;t)\right)dz=\alpha_1 c_{j+1}(t)+\alpha_2 c_{j+2}(t).
\end{align*}
Therefore, it is required that
\begin{align*}
 \frac{d}{dt}w(z;t)+\alpha w(z;t)=(\alpha_1 z+\alpha_2 z^2)w(z;t),
\end{align*}
based on which, concrete moments will be given in \eqref{3-7-2}.

\begin{lemma}\label{th4}
Under the assumptions  \eqref{lu3} and \eqref{3-5-1}, the monic OPs \eqref{3-1-3} satisfy the time evolution
\begin{align}
\frac{d}{dt}P_n(z;t)=&n\alpha P_n(z;t)-a^2_n(\alpha_1+\alpha_2(b_{n-1}+b_n))P_{n-1}(z;t)-\alpha_2a^2_{n-1}a^2_nP_{n-2}(z;t).\label{3-5-2}
\end{align}
\end{lemma}

\begin{proof}
We rewrite the orthogonality condition \eqref{3-1-1} as
\begin{align}
\int P_n(z;t)P_m(z;t)w(z;t)dz=0,\quad\quad m=0,1,\ldots,n-1.\label{3-5-3}
\end{align}
By taking derivation of both sides of \eqref{3-5-3} with respect to $t$ for $m=0,1,\ldots,n-1$, we are led to 
\begin{align}
0=&\int \left(\frac{d}{dt}P_n(z;t)P_m(z;t)+P_n(z;t)\frac{d}{dt}P_m(z;t)\right)w(z;t)dz\nonumber\\
&+\int P_n(z;t)P_m(z;t)\left(\frac{d}{dt}w(z;t)+\alpha w(z;t)\right)dz\nonumber\\
=&\int \left(\frac{d}{dt}P_n(z;t)+(\alpha_1 z+\alpha_2 z^2)P_n(z;t)\right)P_m(z;t)w(z;t)dz.\nonumber
\end{align}
Then, by using the recurrence relation~\eqref{3-1-2}, we obtain
\begin{align}
0=&\int \left(\frac{d}{dt}P_n(z;t)+\alpha_1a^2_nP_{n-1}\right)P_m(z;t)w(z;t)dz\nonumber\\
+&\int \alpha_2a^2_n\left((b_{n-1}+b_n)P_{n-1}(z;t)+a^2_{n-1}P_{n-2}(z;t)\right)P_m(z;t)w(z;t)dz.\nonumber
\end{align}
Since the polynomial $\frac{d}{dt}P_n(z;t)$ is of degree $n$, we can express it as a linear combination of the OPs $P_i$ with $0\leq i\leq n$, that is,
$$\frac{d}{dt}P_n(z;t)=\sum_{i=0}^n\gamma_iP_i(z;t),$$
using which, we get
\begin{align*}
&\frac{d}{dt}P_n(z;t)+\alpha_1a^2_nP_{n-1}(z;t)+\alpha_2a^2_n[(b_{n-1}+b_n)P_{n-1}(z;t)+a^2_{n-1}P_{n-2}(z;t)]\nonumber\\
=&\gamma_nP_n(z;t)+(\gamma_{n-1}+a^2_n(\alpha_1+\alpha_2(b_{n-1}+b_n)))P_{n-1}(z;t)\nonumber\\
&+(\gamma_{n-2}+\alpha_2a^2_{n-1}a^2_n)P_{n-2}(z;t)+\sum_{i=0}^{n-3}\gamma_iP_i(z;t).
\end{align*}
By comparing the highest degree of $z$ on both sides, it is evident from \eqref{lu3} that
$$\gamma_n= n\alpha .$$
According to the orthogonality condition, it is not hard to see that $\gamma_0=\gamma_1= \cdots=\gamma_{n-3}=0$ and 
\begin{align}
&\gamma_{n-2}=-\alpha_2a^2_{n-1}a^2_n,\qquad \gamma_{n-1}=-a^2_n(\alpha_1+\alpha_2(b_{n-1}+b_n)).\nonumber
\end{align}
Therefore, the conclusion \eqref{3-5-2} follows.
\end{proof}

The compatibility condition between the recurrence relation~\eqref{3-1-2} and the time evolution~\eqref{3-5-2} can be employed to produce a nonisospectral Toda lattice. In fact, we have the following theorem.
\begin{theorem}
Under the assumptions \eqref{lu3} and \eqref{3-5-1}, the recurrence coefficients $\{a_n^2\}$ and $\{b_n\}$ for the monic OPs \eqref{3-1-3} satisfy the nonisospectral Toda lattice
\begin{subequations}\label{3-5-9}
\begin{align}
&\frac{da^2_n}{dt}
=2\alpha a^2_n+\alpha_1a^2_n(b_n-b_{n-1})+\alpha_2a^2_n(a^2_{n+1}-a^2_{n-1}+b^2_n-b^2_{n-1}),\label{3-5-10}\\
&\frac{db_n}{dt}
=\alpha b_n+\alpha_1(a^2_{n+1}-a^2_n)+\alpha_2(a^2_{n+1}(b_{n+1}+b_n)-a^2_n(b_{n-1}+b_n)).\label{3-5-11}
\end{align}
\end{subequations}
\end{theorem}
\begin{proof}
Differentiating ~\eqref{3-1-2}, we first have
\begin{align}
&\alpha z P_n(z;t)+z \frac{d}{dt}P_n(z;t)\nonumber\\
=&\frac{d}{dt}P_{n+1}(z;t)+\frac{db_n}{dt}P_{n-1}(z;t)+b_n\frac{d}{dt}P_{n-1}(z;t)
+\frac{da^2_n}{dt}P_{n-1}(z;t)+a^2_n\frac{d}{dt}P_{n-1}(z;t).\label{3-5-6}
\end{align}
By substituting the corresponding expressions ~\eqref{3-5-2} for $\frac{d}{dt}P_{n+1}(z;t),\ \frac{d}{dt}P_n(z;t),\ \frac{d}{dt}P_{n-1}(z;t)$ into both sides of~\eqref{3-5-6}, we then obtain 
\begin{align}
&\text{LHS of} ~\eqref{3-5-6}\nonumber\\
=&(n+1)\alpha (P_{n+1}(z;t)+b_nP_n(z;t)+a^2_nP_{n-1}(z;t))\nonumber\\
&-a^2_n(\alpha_1+\alpha_2(b_{n-1}+b_n))(P_n(z;t)+b_{n-1}P_{n-1}(z;t)+a^2_{n-1}P_{n-2}(z;t))\nonumber\\
&-\alpha_2a^2_{n-1}a^2_n(P_{n-1}(z;t)+b_{n-2}P_{n-2}(z;t)+a^2_{n-2}P_{n-3}(z;t)),\label{3-5-7}
\end{align}
\begin{align}
&\text{RHS of} ~\eqref{3-5-6}\nonumber\\
=&(n+1)\alpha P_{n+1}(z;t)-a^2_{n+1}(\alpha_1+\alpha_2(b_{n+1}+b_n))P_n(z;t)-\alpha_2a^2_{n+1}a^2_nP_{n-1}(z;t)\nonumber\\
&+\frac{db_n}{dt}P_n(z;t)+b_n(n\alpha P_n(z;t)-a^2_n(\alpha_1+\alpha_2(b_{n-1}+b_n))P_{n-1}(z;t)\nonumber\\
&-\alpha_2a^2_{n-1}a^2_nP_{n-2}(z;t))+\frac{da^2_n}{dt}P_{n-1}(z;t)+a^2_n((n-1)\alpha P_{n-}(z;t)\nonumber\\&-a^2_{n-1}(\alpha_1+\alpha_2(b_{n-2}+b_{n-1}))P_{n-2}(z;t)-\alpha_2a^2_{n-2}a^2_{n-1}P_{n-3}(z;t)),\label{3-5-8}
\end{align}
where \eqref{3-1-2} is used to replace $z P_n(z;t),\ z P_{n-1}(z;t),\ z P_{n-2}(z;t)$ in~\eqref{3-5-7}.

Finally, we obtain the desired equation \eqref{3-5-9} by comparing the coefficients of $P_{n+1}(z;t),$ $ P_n(z;t),$ $ P_{n-1}(z;t),\ P_{n-2}(z;t)$ and $P_{n-3}(z;t)$ in~\eqref{3-5-7} and ~\eqref{3-5-8}. 
\end{proof}

\begin{remark}
It is worth noting that, when $\alpha\neq0$,~\eqref{3-5-9} represents a generalized nonisospectral Toda lattice that encompasses both the first and second flows of the isospectral Toda hierarchy. 
When $\alpha=\alpha_2=0$, \eqref{3-5-9} reduces to the first isospectral  Toda flow. Similarly, it yields the second isospectral Toda flow when $\alpha=\alpha_1=0$.
\end{remark}

\subsection{Asymmetric d-$\text{P}_\text{I}$ related to nonisospectral Toda}
The three-term recurrence relation~\eqref{3-1-2} and the time evolution~\eqref{3-5-2} constitute the Lax pair of the nonisospectral Toda lattice \eqref{3-5-9}, which we can write in matrix form as
\begin{align}\label{lax_ntoda}
\psi_{n+1}=A_n\psi_n,\quad \frac{d\psi_n}{dt}=B_n\psi_n,
\end{align}
where $\psi_n=(P_{n-1}(z;t),P_n(z;t))^\top$ and
\begin{align*}
A_n=\left(
\begin{array}{cc}
 0 & 1 \\
 -a_n^2 & z-b_n  \\
\end{array}
\right),
\end{align*}
\begin{align*}
B_n=\left(
\begin{array}{cc}
 \alpha _2 a_{n-1}^2-\left(z -b_{n-1}\right) \left(\alpha _2 z +\alpha _1+\alpha _2 b_{n-1}\right)+\alpha  (n-1) & \alpha _2 z +\alpha _1+\alpha _2 b_{n-1} \\
 -a_n^2 \left(\alpha _2 z +\alpha _1+\alpha _2 b_n\right) & \alpha _2 a_n^2+\alpha  n \\
\end{array}
\right).
\end{align*}
We claim that the compatibility condition of the linear system
\begin{align}
\psi_{n+1}=F_n\psi_n,\quad \frac{\partial\psi_n}{\partial z}=G_n\psi_n,\label{3-6-1}
\end{align}
with
$$F_n=A_n,\quad G_n=B_n/\left(\frac{dz}{dt}\right)=\frac{1}{\alpha z }B_n$$
yields a d-P-type equation.

\begin{theorem}
Under the assumptions  \eqref{lu3} and \eqref{3-5-1} as well as the stationary reduction, the recurrence coefficients $\{a_n^2\}_{n\in \mathbb{N}}$ and $\{b_n\}_{n\in \mathbb{N}}$  for the monic OPs \eqref{3-1-3} satisfy the asymmetric d-P$_{\text{I}}$ 
\begin{subequations}\label{3-6-8}
\begin{align}
&\alpha_2u_n(u_n+v_{n+1}+v_n)+\alpha_1u_n+\alpha n+e_1=0,\\
&\alpha_2v_n(u_n+u_{n-1}+v_n)+\alpha_1v_n+\alpha (n-1)+f_1=0,
\end{align}
\end{subequations}
with the Lax pair \eqref{3-6-1}, where
\begin{subequations}
\begin{align*}
&u_n=\frac{\alpha_2a^2_{n+1}+\alpha n+e_1}{-\alpha_2b_n-\alpha_1},\qquad v_n=\frac{\alpha_2a^2_n+\alpha (n-1)+f_1}{-\alpha_2b_n-\alpha_1},
\end{align*}
\end{subequations}
with some arbitrary constants $e_1$ and $f_1$.
\end{theorem}

\begin{proof}
The compatibility condition in \eqref{3-6-1} gives a difference system in matrix form
\begin{align*}
F_{n,z}+F_nG_n-G_{n+1}F_n=0,
\end{align*}
 the explicit form of which reads
\begin{subequations}\label{3-6-2}
\begin{align}
&2\alpha +\alpha_1(b_n-b_{n-1})+\alpha_2(a^2_{n+1}-a^2_{n-1}+b^2_n-b^2_{n-1})=0,\label{3-6-3}\\
&\alpha b_n+\alpha_1(a^2_{n+1}-a^2_n)+\alpha_2(a^2_{n+1}(b_{n+1}+b_n)-a^2_n(b_{n-1}+b_n))=0.\label{3-6-4}
\end{align}
\end{subequations}
Now we are ready to simplify \eqref{3-6-2} to obtain the asymmetric d-P$_{\text{I}}$ (see e.g. \cite{gordoa2007new}).
It is noted that \eqref{3-6-2} can be equivalently written as
\begin{align}
\left(
\begin{array}{cc}
 Ea^2_n-a^2_nE^{-1} & b_n(E-1) \\
(1-E^{-1})b_n & (E-E^{-1})  \\
\end{array}
\right)
\left(
\begin{array}{cc}
 \alpha_2b_n+\alpha_1 \\
 \alpha_2a^2_n+(n-1)\alpha  \\
\end{array}
\right)=
\left(
\begin{array}{cc}
 0 \\
 0 \\
\end{array}
\right),\label{3-6-5}
\end{align}
where $E$ is the shift operator on $n$, $E^kf_n=f_{n+k}$. Multiplying~\eqref{3-6-5} from the left by 
\begin{align}
\left(
\begin{array}{cc}
0 & E \\
 \alpha_2b_n+\alpha_1 & (\alpha_2a_{n+1}^2+n\alpha)E \\
\end{array}
\right), \nonumber
\end{align}
we can get
\begin{align}
(E-1)\left(
\begin{array}{cc}
 b_nR_{1,n}+(E+1)S_{1,n} \\
a^2_nR_{1,n-1}R_{1,n}+b_nR_{1,n}S_{1,n}+S_{1,n+1}S_{1,n}  \\
\end{array}
\right)=
\left(
\begin{array}{cc}
 0 \\
 0 \\
\end{array}
\right),\label{3-6-6}
\end{align}
where
\begin{align}
&R_{1,n}=\alpha_2b_n+\alpha_1,\nonumber\\
&S_{1,n}=\alpha_2a^2_n+(n-1).\nonumber
\end{align}
Expanding the expression \eqref{3-6-6} explicitly, we have 
\begin{subequations}\label{3-6-7}
\begin{align}
&\alpha_2(a^2_{n+1}+a^2_n+b_n^2)+\alpha_1b_n+(2n-1)\alpha+c_1=0,\\
&a^2_n(\alpha_2b_n+\alpha_1)(\alpha_2b_{n-1}+\alpha_1)-(\alpha_2a^2_n+(n-1)\alpha)^2-c_1(\alpha_2a^2_n+(n-1)\alpha)+d_1=0,
\end{align}
\end{subequations}
where $c_1, d_1$ are two arbitrary constants.
Finally, we can derive the asymmetric d-P$_{\text{I}}$ \eqref{3-6-8}
from \eqref{3-6-7} via the B\"{a}cklund transformation 
\begin{align*}
&b_n=u_n+v_n,&&a^2_n=u_{n-1}v_n,\\
&u_n=\frac{\alpha_2a^2_{n+1}+\alpha n+e_1}{-\alpha_2b_n-\alpha_1},&&v_n=\frac{\alpha_2a^2_n+\alpha (n-1)+f_1}{-\alpha_2b_n-\alpha_1},\\
&c_1=e_1+f_1,\quad d_1=-e_1f_1.
\end{align*}
\end{proof}

\begin{remark}
It is noted that the asymmetric d-P$_{\text{I}}$ can be obtained from the d-P$_{\text{I}}$
\begin{align*}
&\alpha_2w_n(w_{n+1}+w_n+w_{n-1})+\alpha_1w_n+\frac{1}{2}\alpha n-\frac{1}{4}(3\alpha-2e_1-2f_1)\nonumber\\&-(-1)^n\frac{1}{4}(\alpha-2e_1-2f_1)=0,
\end{align*}
by considering the cases of odd and even values of $n$ as a coupled system and then setting $w_{2n}=v_n, w_{2n+1}=u_n$ (see. e.g.  \cite{grammaticos1998continuous}). 
\end{remark}
\begin{remark}
Obviously, \eqref{3-6-2} is the stationary reduction (i.e. setting $\frac{da^2_n}{dt}=\frac{db_n}{dt}=0$) of the nonisospectral Toda lattice~\eqref{3-5-9}. This means that 
the asymmetric d-P$_{\text{I}}$ \eqref{3-6-8} together with its Lax pair \eqref{3-6-1} can be deduced from stationary reduction of the Lax pair \eqref{lax_ntoda} of the nonisospectral Toda lattice \eqref{3-5-9}.
\end{remark}

\subsection{Realization of stationary reduction}
In this subsection, we would like to demonstrate the
feasibility of the stationary reduction from the perspective of solution.
To this end, we will introduce a concrete weight function to define the corresponding moments.

\begin{lemma}\label{th5}
Under the assumption  \eqref{lu3} together with $\frac{\alpha_2}{\alpha}<0$, define the moments as
\begin{align}
c_j(t)=\int^{+\infty}_0z^j(0)e^{j\alpha t}e^{\frac{\alpha_1}{\alpha}z(0)e^{\alpha t}+\frac{\alpha_2}{2\alpha}z^2(0)e^{2\alpha t}}dz(0).\label{3-7-2}
\end{align}
Then the moments simultaneously satisfy the time evolution \eqref{3-5-1} and
\begin{align}
\frac{d}{dt}c_j(t)=-\alpha c_j(t).\label{3-7-3}
\end{align}
\end{lemma}

\begin{proof}
First, let's seek for an appropriate weight function so that the moments satisfy the evolution relation \eqref{3-5-1}. It follows from \eqref{lu3} that
\begin{align*}
z(t)=z(0)e^{\alpha t},
\end{align*}
where $z(0)$ is the spectral parameter at initial time. In the case of the integral interval $[0,+\infty)$, the moments can be written as
\begin{align*}
c_j(t)=&\int^{+\infty}_0z^j(t)d\mu(z;t)=\int^{+\infty}_0z^j(0)e^{j\alpha t}f(z(0);t)dz(0),
\end{align*}
where $e^{j\alpha t}f(z(0);t)$ is the deformed weight function that needs to be determined. Differentiating the above expressions for the moments  produces
\begin{align*}
\frac{d}{dt}c_j(t)=\alpha jc_j(t)+\int^{+\infty}_0z^j(0)e^{j\alpha t}\frac{df(z(0);t)}{dt}dz(0),
\end{align*}
from which we can obtain
\begin{align*}
\frac{d}{dt}f(z(0);t)=(\alpha_1 z(0)e^{\alpha t}+\alpha_2 z^2(0)e^{2\alpha t})f(z(0);t),
\end{align*}
with the help of the evolution relation ~\eqref{3-5-1}.
This implies that it is reasonable to take
$$f(z(0);t)=e^{\frac{\alpha_1}{\alpha}z(0)e^{\alpha t}+\frac{\alpha_2}{2\alpha}z^2(0)e^{2\alpha t}},$$
which yields the expressions of the moments \eqref{3-7-2}.

Next we aim to verify the evolution relation \eqref{3-7-3} satisfied by \eqref{3-7-2}.
Integration by parts gives
\begin{align*}
c_j(t)=&-\int^{+\infty}_0z(0)jz^{j-1}(0)e^{j\alpha t+\frac{\alpha_1}{\alpha}z(0)e^{\alpha t}+\frac{\alpha_2}{2\alpha}z^2(0)e^{2\alpha t}}dz(0)\notag\\
&-\int^{+\infty}_0\frac{\alpha_1}{\alpha}z^{j+1}(0)e^{(j+1)\alpha t}e^{\frac{\alpha_1}{\alpha}z(0)e^{\alpha t} +\frac{\alpha_2}{2\alpha}z^2(0)e^{2\alpha t}}dz(0)\notag\\
&-\int^{+\infty}_0\frac{\alpha_2}{\alpha}z^{j+2}(0)e^{(j+2)\alpha t}e^{\frac{\alpha_1}{\alpha}z(0)e^{\alpha t} +\frac{\alpha_2}{2\alpha}z^2(0)e^{2\alpha t}}dz(0)\notag\\
=&-jc_j-\frac{\alpha_1}{\alpha}c_{j+1}-\frac{\alpha_2}{\alpha}c_{j+2},
\end{align*}
where we used the fact at the boundary 
\begin{align*}
\lim_{z(0)\to 0}z^{j+1}(0)f(z(0);t)=\lim_{z(0)\to +\infty}z^{j+1}(0)f(z(0);t)=0.
\end{align*}
Therefore we have
\begin{align*}
c_{j+2}=-(j+1)\frac{\alpha}{\alpha_2}c_j-\frac{\alpha_1}{\alpha_2}c_{j+1},
\end{align*}
inserting which into Eq.~\eqref{3-5-1}, we obtain Eq.~\eqref{3-7-3}. The proof is completed
\end{proof}
Finally, we show that under time evolution~\eqref{3-7-3}, the nonisospectral  Toda lattice and the corresponding Lax pair can indeed be stationary. In fact, the following result immediately follows from the above lemma.

\begin{theorem}\label{th:stationary_toda}
  Under the definition of the moments \eqref{3-7-2}, we have
  \begin{align*}
    \frac{da^2_n(t)}{dt}=\frac{db_n}{dt}=0, \qquad \text{and}\qquad \frac{d\gamma_{n,j}(t)}{dt}=0,
\end{align*}
where $\{a_n^2\}$ and $\{b_n\}$ are the recurrence coefficients for the monic OPs \eqref{3-1-3} and $\{\gamma_{n,j}\}$ are the coefficients in the expansion $P_n(z;t)=\sum_{j=0}^n\gamma_{n,j}(t)z(t)^j$.
\end{theorem}
\begin{proof}
The main argument involves the utilization of \eqref{3-7-3}, which holds based on the definition \eqref{3-7-2} of the moments. It obviously follows from~\eqref{3-7-3} that  
$$\frac{d}{dt}\Delta_n=-n\alpha\Delta_n,\ \frac{d}{dt}\Delta_n^{*}=-n\alpha\Delta_n^{*}.$$
Then we easily get
\begin{align*}
\frac{da^2_n}{dt}=\frac{d}{dt}\left(\frac{\Delta_{n+1}\Delta_{n-1}}{\Delta_n^2}\right)
=0,\qquad \frac{db_n}{dt}=\frac{d}{dt}\left(\frac{\Delta^{*}_{n+1}}{\Delta_{n+1}}-\frac{\Delta^{*}_n}{\Delta_n}\right)
=0.
\end{align*}

Furthermore, recall that the coefficients of the expansion for $P_n(z)$ are written as
$$\gamma_{n,j}=\frac{T_{n,j}}{\Delta_n},\quad  j=0,\ldots,n,$$
where $T_{n,j}$ denotes the $(n+1,j+1)$ cofactor of the numerator of $P_n(z)$, i.e.
\begin{align*}
T_{n,j}=(-1)^{n+j}\det\begin{pmatrix}
c_{p+q}
\end{pmatrix}_{q=0,1,\ldots,n,\,q\neq j}^{p=0,1,\ldots,n-1}.
\end{align*}
By employing the evolution relation~\eqref{3-7-3}, one can obtain 
$$\frac{d}{dt}T_{n,j}=-n\alpha T_{n,j},$$
from which it is not hard to  conclude that $\gamma_{n,j}$ is also independent of $t$.
\end{proof}

In summary, under the definition~\eqref{3-7-2} of moments, the nonisospectral Toda lattice \eqref{3-5-9} does indeed exhibit the stationary reduction. Besides, we also have demonstrated that the coefficients $\gamma_{n,j}$ in the expansion of the monic OPs don't depend on time $t$, which means that $\frac{\partial\psi_n}{\partial t}=0$ so that the Lax pair \eqref{lax_ntoda} of the nonisospectral Toda lattice does indeed exhibit the stationary reduction under the definition of the moments \eqref{3-7-2}, from which the asymmetric d-P$_{\text{I}}$ \eqref{3-6-8} together with its Lax pair \eqref{3-6-1} arises.

\subsection{Nonisospectral Lotka-Volterra lattice}\label{nlv}
Now we consider the monic OPs $\{P_n(z;t)\}_{n\in \mathbb{N}}$ \eqref{3-1-3} with symmetric measure.
Let the weight function $w(z;t)$ be an even function with respect to $z$ and the integral interval be symmetric. In fact, in this case, it is obvious that the moments satisfy
$$c_{2n}\neq0, \qquad c_{2n+1}=0,$$
from which it is not hard to see $b_n=0$ and the recurrence relation \eqref{3-1-2} becomes
\begin{align}
z P_n(z;t)=P_{n+1}(z;t)+a_n^2P_{n-1}(z;t).\label{3-2-1}
\end{align}

Suppose that the spectral parameter $z$ satisfies the time evolution \eqref{lu3}. We now seek for a weight function such that the moments satisfy the time evolution
\begin{align}
\frac{d}{dt}c_j(t)=\alpha jc_j(t)+\alpha_1 c_{j+2}(t)+\alpha_2 c_{j+4}(t).\label{3-2-2}
\end{align}
Since 
\begin{align*}
\frac{d}{dt}c_j(t)=&\int  z^j\left(\alpha j w(z;t)+\frac{d}{dt}w(z;t)+\alpha w(z;t)\right)dz \nonumber\\
=&\alpha jc_j(t)+\int  z^j\left(\frac{d}{dt}w(z;t)+\alpha w(z;t)\right)dz,
\end{align*}
we have
\begin{align*}
\int  z^j\left(\frac{d}{dt}w(z;t)+\alpha w(z;t)\right)dz=\alpha_1 c_{j+2}(t)+\alpha_2 c_{j+4}(t),
\end{align*}
from which it suffices to set
\begin{align}
\frac{d}{dt}w(z;t)+\alpha w(z;t)=(\alpha_1 z^2+\alpha_2 z^4)w(z;t).\label{3-2-3}
\end{align}
This implies that such a weight function exists and one can choose the expressions \eqref{3-4-2} for the moments.
\begin{lemma}\label{th1}
Under the assumptions  \eqref{lu3} and \eqref{3-2-2}, the monic OPs \eqref{3-1-3} satisfy the time evolution
\begin{align}
\frac{d}{dt}P_n(z;t)=&n\alpha P_n(z;t)-a^2_{n-1}a^2_n(\alpha_1+\alpha_2(a^2_{n-2}+a^2_{n-1}+a^2_n\nonumber\\&+a^2_{n+1}))P_{n-2}(z;t)
-\alpha_2a^2_{n-3}a^2_{n-2}a^2_{n-1}a^2_nP_{n-4}(z;t).\label{3-2-4}
\end{align}
\end{lemma}

\begin{proof}
For $m=0,1,\ldots,n-1$, by taking the derivative to the orthogonality condition \eqref{3-1-1} with respect to time $t$, we have
\begin{align*}
0=&\int \left(\frac{d}{dt}P_n(z;t)P_m(z;t)+P_n(z;t)\frac{d}{dt}P_m(z;t)\right)w(z;t)dz \nonumber\\
&+\int P_n(z;t)P_m(z;t)(\frac{d}{dt}w(z;t)+\alpha w(z;t))dz,
\end{align*}
which immediately results in
\begin{align*}
\int \left(\frac{d}{dt}P_n(z;t)+(\alpha_1 z^2+\alpha_2 z^4)P_n(z;t)\right)P_m(z;t)w(z;t)dz=0
\end{align*}
by virtue of \eqref{3-2-3}. With the help of the recurrence relation \eqref{3-2-1}, we obtain
\begin{align}
0=&\int \left(\frac{d}{dt}P_n(z;t)+\alpha_1a^2_{n-1}a^2_n P_{n-2}\right)P_m(z;t)w(z;t)dz\nonumber\\
&+\int \alpha_2a^2_{n-1}a^2_n(a^2_{n-2}+a^2_{n-1}+a^2_n+a^2_{n+1})P_{n-2}(z;t)P_m(z;t)w(z;t)dz\nonumber\\
&+\int \alpha_2a^2_{n-1}a^2_na^2_{n-3}a^2_{n-2}P_{n-4}(z;t)P_m(z;t)w(z;t)dz\nonumber.
\end{align}
Upon setting $\frac{d}{dt}P_n(z;t)=\sum_{i=0}^n\gamma_iP_i(z;t)$, we have
\begin{align}
&\frac{d}{dt}P_n(z;t)+\alpha_1a^2_{n-1}a^2_nP_{n-2}(z;t)\nonumber\\
&+\alpha_2a^2_{n-1}a^2_n((a^2_{n-2}+a^2_{n-1}+a^2_n+a^2_{n+1})P_{n-2}(z;t)+a^2_{n-3}a^2_{n-2}P_{n-4}(z;t))\nonumber\\
=&\gamma_nP_n(z;t)+\gamma_{n-1}P_{n-1}(z;t)\nonumber\\
&+(\gamma_{n-2}+\alpha_1a^2_{n-1}a^2_n+\alpha_2a^2_{n-1}a^2_n(a^2_{n-2}+a^2_{n-1}+a^2_n
+a^2_{n+1}))P_{n-2}(z;t)\nonumber\\
&+\gamma_{n-3}P_{n-3}(z;t)+(\gamma_{n-4}+\alpha_2a^2_{n-3}a^2_{n-2}a^2_{n-1}a^2_n)P_{n-4}(z;t)+\sum_{i=0}^{n-5}\gamma_iP_i(z;t).\label{3-2-5}
\end{align}
Comparing the highest power of $z$ on both sides of the equation ~\eqref{3-2-5} immediately gives $\gamma_n= n\alpha $. By making use of the orthogonality condition, it is not hard to get
$$\gamma_0=\gamma_1= \cdots=\gamma_{n-5}=\gamma_{n-3}=\gamma_{n-1}=0$$ and
\begin{align}
&\gamma_{n-4}=-\alpha_2a^2_{n-3}a^2_{n-2}a^2_{n-1}a^2_n,\qquad \gamma_{n-2}=-a^2_{n-1}a^2_n(\alpha_1+\alpha_2(a^2_{n-2}+a^2_{n-1}+a^2_n+a^2_{n+1})).\nonumber
\end{align}
Therefore,  the time evolution \eqref{3-2-4} of $\{P_n(z;t)\}_{n\in \mathbb{N}}$ is verified.
\end{proof}
The compatibility condition of the three-term recurrence relation ~\eqref{3-2-1} and the time evolution \eqref{3-2-4} enables us to derive a differential equation. In fact, we have the following theorem.
\begin{theorem} Under the assumptions  \eqref{lu3} and \eqref{3-2-2}, the recurrence coefficients $\{a_n^2\}_{n\in\mathbb{N}}$ for the monic OPs \eqref{3-1-3} satisfy the nonisospectral Lotka-Volterra lattice
\begin{align}
\frac{da^2_n}{dt}=&2\alpha a^2_n+\alpha_1a^2_n(a^2_{n+1}-a^2_{n-1})\nonumber\\&+\alpha_2 a^2_n(a^2_{n+1}(a^2_n+a^2_{n+1}+a^2_{n+2})-a^2_{n-1}(a^2_{n-2}+a^2_{n-1}+a^2_n)).\label{3-2-9}
\end{align}
\end{theorem}
\begin{proof}
Taking the derivative of ~\eqref{3-2-1} with respect to time $t$, we have
\begin{align}
&\alpha z P_n(z;t)+z \frac{d}{dt}P_n(z;t)=\frac{d}{dt}P_{n+1}(z;t)+\frac{da^2_n}{dt}P_{n-1}(z;t)+a^2_n\frac{d}{dt}P_{n-1}(z;t).\label{3-2-6}
\end{align}
By inserting the three-term recurrence relation ~\eqref{3-2-1} and the time evolution ~\eqref{3-2-4} into both sides of the equation~\eqref{3-2-6}, we obtain
\begin{align}
&\text{LHS of} ~\eqref{3-2-6}\nonumber\\
=&\alpha (P_{n+1}(z;t)+a^2_nP_{n-1}(z;t))+n\alpha(P_{n+1}(z;t)+a^2_nP_{n-1}(z;t))\nonumber\\
&-a^2_{n-1}a^2_n(\alpha_1+\alpha_2(a^2_{n-2}+a^2_{n-1}+a^2_n+a^2_{n+1}))(P_{n-1}(z;t)+a^2_{n-2}P_{n-3}(z;t))\nonumber\\
&-\alpha_2a^2_{n-3}a^2_{n-2}+a^2_{n-1}a^2_n(P_{n-3}(z;t)+a^2_{n-4}P_{n-5}(z;t)),\label{3-2-7}
\end{align}
\begin{align}
&\text{RHS of} ~\eqref{3-2-6}\nonumber\\
=&(n+1)\alpha P_{n+1}(z;t)-a^2_{n+1}a^2_n(\alpha_1+\alpha_2(a^2_{n-1}+a^2_n+a^2_{n+1}+a^2_{n+2}))P_{n-1}(z;t)\nonumber\\
&-\alpha_2a^2_{n-2}a^2_{n-1}a^2_na^2_{n+1}P_{n-3}(z;t)+\frac{da^2_n}{dt}P_{n-1}(z;t)+(n-1)\alpha a^2_n P_{n-1}(z;t)\nonumber\\
&-\alpha_2a^2_{n-4}a^2_{n-3}a^2_{n-2}a^2_{n-1}a^2_nP_{n-5}(z;t)\nonumber\\
&-a^2_{n-2}a^2_{n-1}a^2_n(\alpha_1+\alpha_2(a^2_{n-3}+a^2_{n-2}+a^2_{n-1}+a^2_n))P_{n-3}(z;t).\label{3-2-8}
\end{align}
Comparing the right hand sides of~\eqref{3-2-7} and ~\eqref{3-2-8}, one will see that the differential equation \eqref{3-2-9} follows from the coefficient of $P_{n-1}(z;t)$. 
\end{proof}
\begin{remark}
 It is noted that \eqref{3-2-9} can be identified as a nonisosepctral generalization that incorporates both the first and second flows of the isospectral Lotka-Volterra hierarchy when $\alpha\neq0$. In fact, when $\alpha=\alpha_2=0$, it reduces to the first isospectral Lotka-Volterra  flow, and  when $\alpha=\alpha_1=0$, it gives the second isospectral Lotka-Volterra flow.
\end{remark}

\subsection{d-P$_{\text{I}}$ related to nonisospectral Lotka-Volterra}
In the previous subsection, we utilized the orthogonality condition \eqref{3-1-1} to derive the time evolution \eqref{3-2-4} of the OPs, which together with the recurrence relation ~\eqref{3-2-1} constitute the Lax pair of the nonisospectral Lotka-Volterra \eqref{3-2-9}.
As we will see, we can obtain a d-P-type equation from a stationary Lax pair by employing the method introduced in Section \ref{BM}.

The recurrence relation~\eqref{3-2-1} and time evolution~\eqref{3-2-4} can be rewritten in matrix form as
\begin{align}\label{lax_nlv}
\psi_{n+1}=A_n\psi_n,\quad \frac{d\psi_n}{dt}=B_n\psi_n,
\end{align}
where $\psi_n=(P_{n-1}(z;t),P_{n}(z;t))^\top$, 
\begin{align*}
A_n=\left(
\begin{array}{cc}
 0 & 1 \\
 -a_n^2 & z  \\
\end{array}
\right),
\end{align*}
\begin{align*}
B_n=\left(
\begin{array}{cc}
 -\alpha _2 z ^4-z ^2 \left(\alpha _2 a_n^2+\alpha _1\right)+\omega_n+\alpha  (n-1) & \alpha _2 z ^3+z  \left(\alpha _2 \left(a_{n-1}^2+a_n^2\right)+\alpha _1\right) \\
 -\alpha _2 z ^3 a_n^2-z  a_n^2 \left(\alpha _2 \left(a_{n+1}^2+a_n^2\right)+\alpha _1\right) & \alpha _2 z ^2 a_n^2+\omega_{n+1}+\alpha  n \\
\end{array}
\right),
\end{align*}
with $\omega_n=a_{n-1}^2 \left(\alpha _2 \left(a_{n-2}^2+a_{n-1}^2+a_n^2\right)+\alpha _1\right)$. Now we claim that the compatibility condition of the linear system
\begin{align}
\psi_{n+1}=F_n\psi_n,\quad \frac{\partial\psi_n}{\partial z}=G_n\psi_n,\label{oddp}
\end{align}
with
$$F_n=A_n,\quad G_n=B_n/\left(\frac{dz}{dt}\right)=\frac{1}{\alpha z }B_n$$
will yield a d-P-type equation. In fact, we have the following theorem.

\begin{theorem}
Under the assumptions  \eqref{lu3} and \eqref{3-2-2} as well as the stationary reduction, the recurrence coefficients $\{a_n^2\}_{n\in \mathbb{N}}$ for the monic OPs \eqref{3-1-3}  satisfy the following d-P$_{\text{I}}$
\begin{align}
\alpha_1 a^2_{n+1}(a^2_n+a^2_{n+1}+a^2_{n+2}+\frac{\alpha_2}{\alpha_1})+(n+\frac{1}{2})\alpha-(-1)^{n}\beta_0-\tilde{c}=0,\label{3-3-2}
\end{align}
with the Lax pair \eqref{oddp}, where
\begin{align}
&\beta_0=\frac{\alpha}{2}+\alpha_1a^2_1+\alpha_2a^2_1(a^2_1+a^2_2), \qquad c=-\alpha_1a^2_1-\alpha_2 a^2_1(a^2_1+a^2_2),\nonumber
\end{align}
and $\tilde{c}=-c$ when $n$ is odd; $\tilde{c}=0$ when $n$ is even. 
\end{theorem}

\begin{proof}
By considering the compatibility condition of \eqref{oddp}, i.e.
$$F_{n,z}+F_nG_n-G_{n+1}F_n=0,$$
we can obtain
\begin{align}
2\alpha +\alpha_1(a^2_{n+1}-a^2_{n-1})+\alpha_2 (a^2_{n+1}(a^2_n+a^2_{n+1}+a^2_{n+2})-a^2_{n-1}(a^2_{n-2}+a^2_{n-1}+a^2_n))=0, \label{3-3-1}
\end{align}
which can be further simplified to get the d-P$_{\text{I}}$. In fact, summing the above equation results in
\begin{align*}
2n\alpha +\alpha_1(a^2_{n+1}+a^2_n)+\alpha_2 (a^2_{n+1}(a^2_n+a^2_{n+1}+a^2_{n+2})+a^2_n(a^2_{n-1}+a^2_{n+1}+a^2_n))+c=0,
\end{align*}
from which, \eqref{3-3-2}  follows.
\end{proof}

\begin{remark}
It is worth noting that the equation ~\eqref{3-3-1} is the stationary reduction, i.e. setting $\frac{da^2_n}{dt}=0$, of the nonisospectral Lotka-Volterra equation ~\eqref{3-2-9}. This means that 
the d-P$_{\text{I}}$ \eqref{3-3-2} together with its Lax pair \eqref{oddp} can be deduced from the stationary reduction of the Lax pair \eqref{lax_nlv} for the nonisospectral Lotka-Volterra equation.
\end{remark}

\subsection{Realization of stationary reduction}\label{lvasr}
In the previous section, we derive the d-P$_{\text{I}}$ by formally applying the stationary reduction to the nonisospectral Lotka-Volterra equation. Now we will construct an explicit weight function and rigorously demonstrate the feasibility of the stationary reduction from the perspective of solution.

\begin{lemma}\label{th2}
Under the assumption  \eqref{lu3} together with $\frac{\alpha_2}{\alpha}<0$, define the moments as
\begin{align}
c_j(t)=\int^{+\infty}_{-\infty}z^j(0)e^{j\alpha t}e^{\frac{\alpha_1}{2\alpha}z^2(0)e^{2\alpha t}+\frac{\alpha_2}{4\alpha}z^4(0)e^{4\alpha t}}dz(0).\label{3-4-2}
\end{align}
Then the moments simultaneously satisfy the time evolution \eqref{3-2-2} and
\begin{align}
\frac{d}{dt}c_j(t)=-\alpha c_j(t).\label{3-4-1}
\end{align}
\end{lemma}

\begin{proof}
First, let's seek for an appropriate weight function so that the moments satisfy the evolution relation~\eqref{3-2-2}. Since it follows from  \eqref{lu3} that
\begin{align*}
z(t)=z(0)e^{\alpha t},
\end{align*}
where $z(0)$ is the spectral parameter at the initial time, in the case of the integral interval $(-\infty,+\infty)$,
the moments can be written as 
\begin{align*}
c_j(t)=&\int^{+\infty}_{-\infty}z^j(t)d\mu(z;t)=\int^{+\infty}_{-\infty}z^j(0)e^{j\alpha t}f(z(0);t)dz(0),
\end{align*}
where $e^{j\alpha t}f(z(0);t)$ is the deformed weight function that needs to be determined. By taking the derivative of the moments, we have
\begin{align*}
\frac{d}{dt}c_j(t)=\alpha jc_j(t)+\int^{+\infty}_{-\infty}z^j(0)e^{j\alpha t}\frac{df(z(0);t)}{dt}dz(0).
\end{align*}
It can be seen from the evolution relation~\eqref{3-2-2} that $f(z(0);t)$ needs to satisfy the time evolution
\begin{align*}
\frac{d}{dt}f(z(0);t)=(\alpha_1 z^2(0)e^{2\alpha t}+\alpha_2 z^4(0)e^{4\alpha t})f(z(0);t),
\end{align*}
so it is reasonable to set
$$f(z(0);t)=e^{\frac{\alpha_1}{2\alpha}z^2(0)e^{2\alpha t}+\frac{\alpha_2}{4\alpha}z^4(0)e^{4\alpha t}},$$
from which we have the explicit expression \eqref{3-4-2} for the moments.
Now it remains to be confirmed that the evolution relation \eqref{3-4-1} is satisfied by the moments in \eqref{3-4-2}.

Using integration by parts and the fact at the boundary
\begin{align*}
\lim_{z(0)\to -\infty}z^{j+1}(0)f(z(0);t)=\lim_{z\to +\infty}z^{j+1}(0)f(z(0);t)=0,
\end{align*}
we have
\begin{align*}
c_j(t)=&-\int^{+\infty}_{-\infty}z(0)jz^{j-1}(0)e^{j\alpha t+\frac{\alpha_1}{2\alpha}z^2(0)e^{2\alpha t}+\frac{\alpha_2}{4\alpha}z^4(0)e^{4\alpha t}}dz(0)\notag\\
&-\int^{+\infty}_{-\infty}\frac{\alpha_1}{\alpha}z^{j+2}(0)e^{(j+2)\alpha t}e^{\frac{\alpha_1}{2\alpha}z^2(0)e^{2\alpha t} +\frac{\alpha_2}{4\alpha}z^4(0)e^{4\alpha t}}dz(0)\notag\\
&-\int^{+\infty}_{-\infty}\frac{\alpha_2}{\alpha}z^{j+4}(0)e^{(j+4)\alpha t}e^{\frac{\alpha_1}{2\alpha}z^2(0)e^{2\alpha t} +\frac{\alpha_2}{4\alpha}z^4(0)e^{4\alpha t}}dz(0)\notag\\
=&-jc_j-\frac{\alpha_1}{\alpha}c_{j+2}-\frac{\alpha_2}{\alpha}c_{j+4},
\end{align*}
which results in
$$c_{j+4}=-(j+1)\frac{\alpha}{\alpha_2}c_j-\frac{\alpha_1}{\alpha_2}c_{j+2}.$$
It is evident that ~\eqref{3-4-1} immediately follows from substituting the above relation into the evolution relation~\eqref{3-2-2}.
Therefore, we complete the proof.
\end{proof}

The above lemma implies the following result.
\begin{theorem}
Under the definition of the moments \eqref{3-4-2}, we have 
\begin{align*}
\frac{da^2_n(t)}{dt}=0 \qquad \text{and}\qquad \frac{d\gamma_{n,j}(t)}{dt}=0,
\end{align*}
where $\{a_n^2\}$ are the recurrence coefficients for the monic OPs \eqref{3-1-3} and $\{\gamma_{n,j}\}$ are the coefficients of OPs in the expansions $P_n(z;t)=\sum_{j=0}^n\gamma_{n,j}(t)z(t)^j$.
\end{theorem}
\begin{proof}
The argument is mainly based on the time evolution~\eqref{3-4-1}, which holds for the moments definied in \eqref{3-4-2}.
It follows from the time evolution~\eqref{3-4-1} of the moments that the Hankel determinant $\Delta_n$ satisfies the evolution $\frac{d}{dt}\Delta_n=-n\alpha\Delta_n$. Therefore we have
\begin{align*}
\frac{da^2_n(t)}{dt}=\frac{d}{dt}\left(\frac{\Delta_{n+1}\Delta_{n-1}}{\Delta_n^2}\right)=0.
\end{align*}
Furthermore, using the same argument in Theorem \ref{th:stationary_toda}, one can also conclude that $\gamma_{n,j}$ is also independent of $t$.
\end{proof}

It follows from the above theorem that, for the moments defined in ~\eqref{3-4-2}, the nonisospectral Lotka-Volterra equation \eqref{3-2-9} does indeed exhibit the stationary reduction. Besides, we also have demonstrated that the coefficients $\gamma_{n,j}$ of OPs don't depend on time $t$, implying $\frac{\partial\psi_n}{\partial t}=0$, which means that
 the Lax pair \eqref{lax_nlv} of the nonisospectral Lotka-Volterra equation  does indeed allow the stationary reduction under the definition of the moments \eqref{3-4-2}, from which the d-P$_{\text{I}}$ \eqref{3-3-2} together with its Lax pair  \eqref{oddp} arises.

\section{Nonisospectral deformation of $(1,m)$-type bi-OPs and d-P}\label{1moP}
In this section, we consider the corresponding problems related to the $(1,m)$-type bi-OPs \cite{chang2015about,maeda2013orthogonal}, which are generalizations of the ordinary OPs discussed in the previous section and associated with the Muttalib-Borodin random ensemble \cite{borodin1998biorthogonal,muttalib1995random}. Firstly, the time evolution of the $(1,m)$-type bi-OPs is deduced by utilizing the biorthogonality condition. Subsequently, the nonisospectral integrable equations are derived based on the compatibility condition between the recurrence relation and the obtained time evolution. Then, by considering stationary reduction of the Lax pair of the nonisospectral integrable equations, we obtain a family of d-P-type equations together with the corresponding Lax pair. Lastly, we construct an explicit weight function that ensures the validity of the aforementioned process.

\subsection{$(1,m)$-type bi-OPs}
For a fixed $m\in\mathbb{N}^+$, define the inner product $\langle\cdot|\cdot\rangle_{m}$ according to
\begin{align*}
	\langle f(z)|w(z)dz|g(z) \rangle_{m}
	=\int f(z)g(z^m)w(z)dz,
\end{align*}
where $w(z)$ is a formal weight function so that all the moments
$$c_i=\int z^iw(z)dz, \ i\in \mathbb{N}$$ exist.
Two families of monic polynomials  $\{P_n(z)\}_{n\in \mathbb{N}}$ and $\{Q_n(z)\}_{n\in \mathbb{N}}$ with an exact degree $n$ for each $P_n(z)$ and $Q_n(z)$  are called $(1,m)$-type bi-OPs with respect to the bilinear $2$-form $\langle\cdot|\cdot\rangle_{m}$ if they satisfy biorthogonality condition
\begin{align}\label{4-1-1}
	\langle P_n(z)|w(z)dz|Q_s(z)\rangle_{m}=h_n \delta_{ns},\quad h_n\neq 0.
\end{align}
Obviously, when $m=1$, the biorthogonality condition reduces to the case for the ordinary OPs in Section \ref{OOP}. 

It is noted that the biorthogonality condition~\eqref{4-1-1} can be equivalently written as 
\begin{align*}
	\langle P_n(z)|w(z)dz|z^s\rangle_{m}=0,\quad s=0,\,1,\,\ldots,\,n-1,
\end{align*}
for the sequence of polynomials $\{P_n(z)\}_{n\in\mathbb{N}}$ and
\begin{align*}
	\langle z^s|w(z)dz|Q_n(z)\rangle_{m}=0,\quad s=0,\,1,\,\ldots,\,n-1
\end{align*}
for the sequence of polynomials $\{Q_n(z)\}_{n\in\mathbb{N}}$.
Based on the biorthogonality, we have the following determinant expressions for the $(1,m)$-type bi-OPs
\begin{align}
P_n(z)=\frac{1}{\tau_n}
\det\left(
\begin{array}{cccc}
c_{mi+j}\\
z^j \\
\end{array}
\right)_{\substack{i=0,\ldots,n-1\\ j=0,\ldots,n\,\,\,\,\,\,\,\,}},\label{exp:mbop1}
\qquad
Q_n(z)=\frac{1}{\tau_n}
\det\left(
\begin{array}{cccc}
c_{mi+j}&z^j \\
\end{array}
\right)_{\substack{i=0,\ldots,n\,\,\,\,\,\,\,\\ j=0,\ldots,n-1}},
\end{align}
under the assumption $\tau_n=\det(c_{mi+j})_{i,j=0}^{n-1}\neq 0$. In addition, $\{h_n\}_{n\in \mathbb{N}}$ admit the determinant representations
$$\quad\quad h_n=\frac{\tau_{n+1}}{\tau_n}.$$
It can also be proved that $P_n(z)$ and $Q_n(z)$ respectively satisfy the following $(m+2)$-term recurrence relations
\begin{align}		
&z^mP_n(z)=\sum_{j=-1}^mb_{n,j}P_{n+j}(z)
,\label{4-1-2}
&z Q_n(z)=\sum_{j=-m}^1a_{n,j}Q_{n+j}(z),
\end{align}
where $b_{n,m}=a_{n,n+1}=1$ and $b_{n,j}=0,j>m$ or $j<-1$. Since $P_n(z)$ and $Q_n(z)$ have similar structures, we shall maintain our focus on $\{P_n(z)\}_{n\in \mathbb{N}}$.

In order to facilitate our discussion, we  rewrite its recurrence relation ~\eqref{4-1-2} in matrix form as
\begin{equation}
z^mP=JP,\label{4-1-4}
\end{equation}
where
\begin{equation}
		P=\begin{pmatrix}
			P_0\\
			P_1\\
			P_2\\
			\vdots
\end{pmatrix},\quad
J=\begin{pmatrix}
			b_{0,0}&b_{0,1}&b_{0,2}&\cdots&b_{0,m}&0&0&0&\cdots\\
			b_{1,-1}&b_{1,0}&b_{1,1}&\cdots&b_{1,m-1}&b_{1,m}&0&0&\cdots\\0&b_{2,-1}&b_{2,0}&\cdots&b_{2,m-2}&b_{2,m-1}&b_{2,m}&0&\cdots\\
			\vdots&\vdots&\vdots&\ddots&\vdots&\vdots&\vdots&\vdots&\ddots&
\end{pmatrix}.\nonumber
\end{equation}

\subsection{Nonisospectral  Blaszak-Marciniak lattice}\label{nbm}
In this subsection, we consider nonisospectral deformation of the $(1,m)$-type bi-OPs, and use the compatibility condition to derive nonisospectral integrable equations.

Under the condition that the spectral parameter $z$ satisfies the time evolution~\eqref{lu3}, we seek for an appropriate weight function to ensure that the moments satisfy the time evolution
\begin{align}
\frac{d}{dt}c_j(t)=\alpha jc_j(t)+\alpha_1 c_{j+m}(t)+\alpha_2 c_{j+2m}(t).\label{4-2-1}
\end{align}
Based on 
\begin{align*}
\frac{d}{dt}c_j(t)=&\int  z^j(t)\left(\alpha j w(z;t)+\frac{d}{dt}w(z;t)+\alpha w(z;t)\right)dz\\
=&\alpha jc_j(t)+\int  z^j(t)\left(\frac{d}{dt}w(z;t)+\alpha w(z;t)\right)dz,
\end{align*}
we thus obtain
\begin{align*}
\int z^j(t)\left(\frac{d}{dt}w(z;t)+\alpha w(z;t)\right)dz=\alpha_1 c_{j+m}(t)+\alpha_2 c_{j+2m}(t),
\end{align*}
which implies that it is reasonable to set
\begin{align}
\frac{d}{dt}w(z;t)+\alpha w(z;t)=(\alpha_1 z^m+\alpha_2 z^{2m})w(z;t).\label{4-2-2}
\end{align}
Such a weight function is given in \eqref{4-4-3}.
\begin{lemma}\label{th6}
Under the assumptions  \eqref{lu3} and \eqref{4-2-1}, the $(1,m)$-type bi-OPs $\{P_n(z)\}_{n\in \mathbb{N}}$
in \eqref{exp:mbop1} satisfy the time evolution
\begin{align}
\frac{d}{dt}P_n(z;t)=&n\alpha P_n(z;t)-\left(\alpha_1b_{n,-1}+\alpha_2b_{n,-1}(b_{n-1,0}+b_{n,0})\right)P_{n-1}(z;t)\nonumber\\&-\alpha_2b_{n,-1}b_{n-1,-1}P_{n-2}(z;t).\label{4-2-3}
\end{align}
\end{lemma}
\begin{proof}
Recall that the biorthogonality condition ~\eqref{4-1-1} gives
\begin{align*}
\int P_n(z;t)Q_s(z^m;t)w(z;t)dz=0,\quad\quad 
\end{align*}
for $s=0,1,\ldots,n-1$. By taking derivation of the above relation with respect to $t$, we obtain
\begin{align}
0=&\int \left(\frac{d}{dt}P_n(z;t)Q_s(z^m;t)+P_n(z;t)\frac{d}{dt}Q_s(z^m;t)\right)w(z;t)dz\nonumber\\
&+\int P_n(z;t)Q_s(z^m;t)\left(\frac{d}{dt}w(z;t)+\alpha w(z;t)\right)dz.\label{4-2-4}
\end{align}
Then, inserting Eq.~\eqref{4-2-2} into Eq.~\eqref{4-2-4} gives
\begin{align}
0=&\int \left(\frac{d}{dt}P_n(z;t)+(\alpha_1 z^m+\alpha_2 z^{2m})P_n(z;t)\right)Q_s(z^m;t)w(z;t)dz\nonumber\\
=&\int \left(\frac{d}{dt}P_n(z;t)+\alpha_1b_{n,-1}P_{n-1}(z;t)\right)Q_s(z^m;t)w(z;t)dz\nonumber\\
&+\alpha_2b_{n,-1}\left(\int (b_{n-1,0}+b_{n,0})P_{n-1}(z;t)Q_s(z^m;t)w(z;t)dz\nonumber\right.\\
&\left.+\int b_{n-1,-1}P_{n-2}(z;t)Q_s(z^m;t)w(z;t)dz\right),\nonumber
\end{align}
where we used the recurrence relation ~\eqref{4-1-2} to replace the expressions $z^mP_n(z;t)$ and $z^{2m}P_n(z;t)$.

Upon the reasonable setting
$$\frac{d}{dt}P_n(z;t)=\sum_{i=0}^n\gamma_iP_i(z;t),$$
we have
\begin{align}
&\frac{d}{dt}P_n(z;t)+\alpha_1b_{n,-1}P_{n-1}(z;t)+\alpha_2b_{n,-1}[(b_{n-1,0}+b_{n,0})P_{n-1}(z;t)+b_{n-1,-1}P_{n-2}(z;t)]\nonumber\\
=&\gamma_nP_n(z)+(\gamma_{n-1}+\alpha_1b_{n,-1}+\alpha_2b_{n,-1}(b_{n-1,0}+b_{n,0}))P_{n-1}(z;t)\nonumber\\
&+(\gamma_{n-2}+\alpha_2b_{n,-1}b_{n-1,-1})P_{n-2}(z;t)+\sum_{i=0}^{n-3}\gamma_iP_i(z;t).\label{4-2-5}
\end{align}
Comparing the highest power of $z$ on both sides of the equation ~\eqref{4-2-5} shows that $\gamma_n= n\alpha $.  According to the biorthogonality condition for $ s=0,1,\ldots,n-1 $, we conclude from \eqref{4-2-5} that $\gamma_0=\gamma_1= \cdots=\gamma_{n-3}=0$, and
\begin{align}
&\gamma_{n-2}=-\alpha_2b_{n,-1}b_{n-1,-1},\qquad \gamma_{n-1}=-\alpha_1b_{n,-1}-\alpha_2b_{n,-1}(b_{n-1,0}+b_{n,0}).\nonumber
\end{align}
Therefore ~\eqref{4-2-3} follows.
\end{proof}

In the following theorem, we will derive a nonisospectral Blaszak-Marciniak (BM) lattice by means of the compatibility condition of the recurrence relation ~\eqref{4-1-2} and the time evolution ~\eqref{4-2-3}.
\begin{theorem}
    Under the assumptions  \eqref{lu3} and \eqref{4-2-1}, the recurrence coefficients $\{b_{n,l}\}$ in \eqref{4-1-2}  for the $(1,m)$-type bi-OPs satisfy the following differential system
    \begin{align}
\frac{db_{n,l}}{dt}=&(m-l)\alpha b_{n,l}+\alpha_1(b_{n,l+1}b_{n+l+1,-1}-b_{n-1,l+1}b_{n,-1})\nonumber\\
&+\alpha_2(b_{n,l+2}b_{n+l+2,-1}b_{n+l+1,-1}+b_{n,l+1}b_{n+l+1,-1}(b_{n+l+1,0}+b_{n+l,0})\nonumber\\
&\qquad-b_{n-1,l+1}b_{n,-1}(b_{n-1,0}+b_{n,0})-b_{n,-1}b_{n-1,-1}b_{n-2,l+2}),\nonumber\\
&l=-1,0,\ldots,m-1
\label{4-2-9}
\end{align}
 with $b_{n,m}=1$ and $b_{n,j}=0$ for $j>m$ or $j<-1$.
\end{theorem}
\begin{proof}
First of all, take the derivative of ~\eqref{4-1-2} with respect to time $t$ to get
\begin{align}
&\alpha mz^m P_n(z;t)+z^m \frac{d}{dt}P_n(z;t)\nonumber\\=&\frac{d}{dt}P_{n+m}(z;t)+\frac{db_{n,m-1}(t)}{dt}P_{n+m-1}(z;t)+b_{n,m-1}(t)\frac{d}{dt}P_{n+m-1}(z;t)\nonumber\\&+\frac{db_{n,m-2}(t)}{dt}P_{n+m-2}(z;t)+b_{n,m-2}(t)\frac{d}{dt}P_{n+m-2}(z)
+\cdots+\frac{db_{n,-1}(t)}{dt}P_{n-1}(z;t)\nonumber\\&+b_{n,-1}(t)\frac{d}{dt}P_{n-1}(z;t)\quad n\geq0.\label{4-2-6}
\end{align}
Substituting the recurrence relation ~\eqref{4-1-2} and the time evolution ~\eqref{4-2-3} into the both sides of~\eqref{4-2-6}, and then comparing the coefficients of polynomials $P_{n+m}(z;t),$ $P_{n+m-1}(z;t)$, $\ldots,$ $P_{n-3}(z;t)$ of the both sides, we finally obtain \eqref{4-2-9} after a lengthy calculation.

\end{proof}
\begin{remark}
 When $\alpha\neq0$, \eqref{4-2-9} is a generalized nonisospectral BM lattice that incorparates both the first and second flows of the isospectral BM hierarchy \cite{blaszak1994r}. 
In the case of $\alpha=\alpha_2=0$, it is reduced to the first isospectral BM flow, while it becomes the second isospectral BM flow when $\alpha=\alpha_1=0$.

\end{remark}

\subsection{d-P related to nonisospectral BM}
The recurrence relation ~\eqref{4-1-2} and the time evolution ~\eqref{4-2-3} form the Lax pair of the nonisospectral BM lattice, which can be expressed in matrix form as
\begin{align}
z^mP=BP, \qquad \frac{d}{dt}P=MP , \label{4-3-1}
\end{align}
where $P(z;t)=(P_0(z;t),P_1(z;t),P_2(z;t),\ldots)^\top$,
\begin{align*}
&B=\left(
\begin{array}{ccccccccc}
 b_{0,0}  &\cdots  &b_{0,m-1}  &1 &  \\
b_{1,-1}  & b_{1,0}&\cdots  &b_{1,m-1}  &1  &   \\
   &b_{2,-1} &b_{2,0} &\cdots  &b_{2,m-1}  &1  &   \\
  && \ddots & \ddots  & & \ddots& \ddots \\
  \end{array}
\right),\\
&M=\left(
\begin{array}{ccccc}
 0 &  &  &     \\
 f_1 & \alpha &  &     \\
 -\alpha_2b_{2,-1}b_{1,-1} &f_2  & 2\alpha &     \\
  & -\alpha_2b_{3,-1}b_{2,-1} & f_3  & 3\alpha   \\
   &  &  \ddots & \ddots & \ddots \ \\
\end{array}
\right),\quad 
\end{align*}
with $f_n=-b_{n,-1}(\alpha_1+\alpha_2(b_{n,0}+b_{n-1,0}))$.

By following the process in Section \ref{BM}, we can consider the stationary reduction of the Lax pair
\begin{align}
z^mP=\widehat{B}P, \qquad \frac{\partial}{\partial z}P=\widehat{M}P, \label{4-3-2}
\end{align}
where
$$\widehat{B}=B,\ \widehat{M}=M/\left(\frac{dz}{dt}\right)=\frac{1}{\alpha z }M.$$
Its compatibility condition gives
$$mz^{m-1}I+z^m\widehat{M}=\frac{\partial}{\partial z}B+\widehat{B}\widehat{M},$$
which can be equivalently written as the following explicit form
\begin{align}
&(m-l)\alpha b_{n,l}+\alpha_1(b_{n,l+1}b_{n+l+1,-1}-b_{n-1,l+1}b_{n,-1})\nonumber\\
&+\alpha_2[b_{n,l+2}b_{n+l+2,-1}b_{n+l+1,-1}+b_{n,l+1}b_{n+l+1,-1}(b_{n+l+1,0}+b_{n+l,0})\nonumber\\
&-b_{n-1,l+1}b_{n,-1}(b_{n-1,0}+b_{n,0})-b_{n,-1}b_{n-1,-1}b_{n-2,l+2}]=0,\nonumber\\
&\quad\quad\quad\quad\quad\quad\quad\quad\quad\quad\quad\quad\quad\quad l=-1,0,\ldots,m-1.\label{4-3-3}
\end{align}
Here $I$ denotes the identity matrix. The above equation corresponds to the stationary reduction of the nonisospectral BM lattice and it can be associated with a class of d-P-type equations.

\begin{theorem}
    Under the assumptions  \eqref{lu3} and \eqref{4-2-1} as well as the stationary reduction, the recurrence coefficients $\{b_{n,0},b_{n,-1}\}$ in \eqref{4-1-2} for the $(1,m)$-type bi-OPs  satisfy the following difference system
    \begin{subequations}\label{4-3-20}
\begin{align}
		&my_{n+1}b_{n,1}+my_nb_{n-1,1}\nonumber\\
  =&\frac{2m^2\alpha}{\alpha_2x_n}(-\frac{\alpha_1}{\sqrt{2}m\alpha}+\frac{1}{x_n})+(m+1-m(b_{n,1}+b_{n-1,1}))n+\frac{\beta_0}{\alpha}-mb_{n,1},\label{4-3-21}\\
		-&\frac{\alpha_2 x_nx_{n-1}}{2m^2\alpha}
  =\frac{(y_n+n)b_{n-1,1}}{H_n},\label{4-3-22}
\end{align}
\end{subequations}
 where
\begin{align}
H_n=&mb^2_{n-1,1}y^2_n+m(1-(n+1)b_{n,1}+nb_{n-1,1})b_{n-1,1}y_n+\frac{1}{2}(m+1\nonumber\\ &-2mb_{n,1}b_{n-1,1})n^2+\frac{1}{2}(-m+1+2mb_{n-1,1}-2mb_{n,1}b_{n-1,1})n\nonumber\\
&-((m+1-m(b_{n,1}+b_{n-1,1}))n+1-mb_{n,1})(y_n+n)b_{n-1,1}\nonumber\\
&-\frac{\sqrt{2}m^2\alpha}{\alpha_2x_n}(y_n+n)\times(y_{n-1}+n-1)b_{n-2,2}-\frac{\sqrt{2}m^2\alpha}{\alpha_2}(y_n+n)(y_{n+1}+n+1)\nonumber\\
&\times\left(\frac{1}{x_{n-1}}+\frac{1}{x_n}+\frac{1}{x_{n+1}}-\frac{\alpha_1}{\sqrt{2}m\alpha}\right)b_{n-1,2}+\frac{m^2\alpha}{\alpha_2}(y_n+n)(y_{n+1}+n+1)\nonumber\\
&\times((y_{n+2}+n+2)b_{n-1,3}+(y_{n-1}+n-1)b_{n-2,3}),\nonumber
\end{align}
{{\begin{equation}
	\begin{split}
		b_{n,m-l}=&\frac{1}{l!}\sum_{p=0}^{[\frac{l}{2}]}\sum_{i=0}^{l}(-1)^iC^i_l\prod_{j=n+m-l+1}^{n+m}\left(\frac{\sqrt{2}m^2\alpha}{\alpha_2}(y_j+j)\left(\frac{1}{x_j}+\frac{1}{x_{j-1}}-\frac{\alpha_1}{\sqrt{2}m\alpha}\right)\right)\\
		&\times\prod_{k=0}^{i-1}\frac{(y_{n-k}+n-k)\left(\frac{1}{x_{n-k}}+\frac{1}{x_{n-k-1}}-\frac{\alpha_1}{\sqrt{2}m\alpha}\right)}{(y_{n+m-k}+n+m-k)\left(\frac{1}{x_{n+m-k}}+\frac{1}{x_{n+m-k-1}}-\frac{\alpha_1}{\sqrt{2}m\alpha}\right)}\\
		&\times \left(\sum_{q=0}^p\prod_{r\in B_{q}}-\frac{\alpha_2(m-i-r+1)}{2m^2\alpha\left(\frac{1}{x_{n+r}}+\frac{1}{n+x_{r-1}}-\frac{\alpha_1}{2\sqrt{2}m}\right)\left(\frac{1}{x_{n+r-1}}+\frac{1}{x_{n+r-2}}-\frac{\alpha_1}{\sqrt{2}m\alpha}\right)}\right.\\
		&\left.\times\prod_{s\in C_{p-q}}-\frac{\alpha_2(-i+1-s)}{2m^2\alpha\left(\frac{1}{x_{n+s}}+\frac{1}{x_{n+s-1}}-\frac{\alpha_1}{\sqrt{2}m\alpha}\right)\left(\frac{1}{x_{n+s-1}}+\frac{1}{n+x_{s-2}}-\frac{\alpha_1}{\sqrt{2}m\alpha}\right)}\right),\\
		&\quad l=m-3,m-2,m-1.\nonumber
	\end{split}
\end{equation}}}
and
\begin{align*}
&x_n=\frac{\sqrt{2}m\alpha}{\alpha_1+\alpha_2b_{n,0}},\quad y_n=-\frac{\alpha_2}{\alpha m}b_{n,-1}-n,\\
&\beta_0= -\alpha_2(b_{1,-1}b_{0,1}+b_{0,-1}b_{-1,1}+b^2_{0,0})-\alpha_1b_{0,0}. 
\end{align*}
{{Here $B_{q}$ denotes a set of $q$ integers 
in the interval $\left[m-l+2, m-i\right]$, where any two distinct elements are at least $2$ units apart, while $C_{q}$  denotes a set of $q$ integers in $\left[-i+2, 0\right]$ with the same spacing constraint. Also note that we use the convention $\prod_{\varnothing}=1$.}}
\end{theorem}

\begin{proof}
Write the equations in \eqref{4-3-3} for $l=-1\ 0,\ 1$ as follows
\begin{subequations}
\begin{align}
&(m+1)\alpha b_{n,-1}+\alpha_1b_{n,-1}(b_{n,0}-b_{n-1,0})
+\alpha_2b_{n,-1}[b_{n,1}b_{n+1,-1}+b_{n,0}(b_{n,0}+b_{n-1,0})\nonumber\\
&-b_{n-1,0}(b_{n-1,0}+b_{n,0})-b_{n-1,-1}b_{n-2,1}]=0,
\label{4-3-4}\\
&m\alpha b_{n,0}+\alpha_1(b_{n,1}b_{n+1,-1}-b_{n-1,1}b_{n,-1})+\alpha_2[b_{n,2}b_{n+2,-1}b_{n+1,-1}+b_{n,1}b_{n+1,-1}(b_{n+1,0}+b_{n,0})\nonumber\\
&-b_{n-1,1}b_{n,-1}(b_{n-1,0}+b_{n,0})-b_{n,-1}b_{n-1,-1}b_{n-2,2}]=0,\label{4-3-5}\\
&(m-1)\alpha b_{n,1}+\alpha_1(b_{n,2}b_{n+2,-1}-b_{n-1,2}b_{n,-1})+\alpha_2[b_{n,3}b_{n+3,-1}b_{n+2,-1}-b_{n,-1}b_{n-1,-1}b_{n-2,3}\nonumber\\
&+b_{n,2}b_{n+2,-1}(b_{n+2,0}+b_{n+1,0})-b_{n-1,2}b_{n,-1}(b_{n-1,0}+b_{n,0})]=0.\label{4-3-6}
\end{align}
\end{subequations}
Dividing both sides of \eqref{4-3-4} by $b_{n,-1}$ 
and summing it for $n$ gives
\begin{align*}
&-\alpha_2(b_{n+1,-1}b_{n,1}+b_{n,-1}b_{n-1,1}+b^2_{n,0})-\alpha_1 b_{n,0}\nonumber\\
=&-\alpha_2(b_{1,-1}b_{0,1}+b_{0,-1}b_{-1,1}+b^2_{0,0})-\alpha_1 b_{0,0}+(m+1)n\alpha.
\end{align*}
Upon setting
$$ -\alpha_2(b_{1,-1}b_{0,1}+b_{0,-1}b_{-1,1}+b^2_{0,0})-\alpha_1b_{0,0} \triangleq \beta_0,$$
we obtain
\begin{equation}
		-\alpha_2(b_{n+1,-1}b_{n,1}+b_{n,-1}b_{n-1,1}+b^2_{n,0})-\alpha_1b_{n,0}=(m+1)n\alpha+\beta_0.\label{4-3-9}
	\end{equation}
After implementing the summation on ~\eqref{4-3-5}, we can get
\begin{align}
&-\alpha_2b_{n,-1}(b_{n+1,-1}b_{n-1,2}+b_{n-1,-1}b_{n-2,2}+b_{n-1,1}(b_{n,0}+b_{n-1,0}))\nonumber\\
&\qquad\qquad\qquad\qquad\qquad\qquad\qquad\qquad\qquad\qquad\quad-\alpha_1b_{n,-1}b_{n-1,1}=m\alpha \sum_{j=0}^{n-1}b_{j,0}.\label{4-3-10}
\end{align}
Multiplying both sides of ~\eqref{4-3-4} by $b_{n-1,1}$ and summing it for $n$, we have
\begin{equation}
		\begin{split}
			&X_{n+1}b_{n,1}b_{n-1,1}\\
   =&(m+1)\alpha\sum_{j=0}^{n}b_{j,-1}b_{j-1,1}+\sum_{j=0}^{n}W_jb_{j-1,1}(b_{j-1,0}-b_{j,0})\\
			=&(m+1)\alpha\sum_{j=0}^{n}b_{j,-1}b_{j-1,1}+\sum_{j=0}^{n-1}(W_{j+1}b_{j,1}-W_jb_{j-1,1})b_{j,0}-W_nb_{n-1,1}b_{n,0},\label{4-3-11}
   \end{split}
	\end{equation}
  where 
 \begin{align}
 &X_{n}=-\alpha_2b_{n,-1}b_{n-1,-1},\label{4-3-12}\\
 &W_n=-b_{n,-1}(\alpha_2(b_{n,0}+b_{n-1,0})+\alpha_1).\label{4-3-13}
 \end{align}
Subsequently, by virtue of ~\eqref{4-3-5}, we obtain
\begin{align}
&X_{n+1}b_{n,1}b_{n-1,1}\nonumber\\
=&(m+1)\alpha\sum_{j=0}^{n}b_{j,-1}b_{j-1,1}+\sum_{j=0}^{n-1}(m\alpha b_{j,0}+X_jb_{j-2,2}-X_{j+2}b_{j,2})b_{j,0}-W_nb_{n-1,1}b_{n,0}\nonumber\\
			=&m\alpha\sum_{j=0}^{n-1}(b_{j+1,-1}b_{j,1}+b_{j,-1}b_{j-1,1}+b^2_{j,0})+m\alpha b_{n,-1}b_{n-1,1}-W_nb_{n-1,1}b_{n,0}\nonumber\\
			&-(m-1)\alpha\sum_{j=0}^{n}b_{j,-1}b_{j-1,1}+\sum_{j=0}^{n-1}(X_jb_{j-2,2}-X_{j+2}b_{j,2})b_{j,0}.\label{4-3-14}
	\end{align}
Based on ~\eqref{4-3-9} and ~\eqref{4-3-10}, we can conclude
\begin{align}
&m\alpha \sum_{j=0}^{n-1}(b_{j+1,-1}b_{j,1}+b_{j,-1}b_{j-1,1}+b^2_{j,0})\nonumber\\
=&-\frac{m\alpha}{\alpha_2}\sum_{j=0}^{n-1}(mj\alpha+j\alpha+\beta_0+\alpha_1b_{j,0})\nonumber\\
=&-\frac{n(n-1)m\alpha}{2\alpha_2}((m+1)\alpha+\beta_0))-\frac{\alpha_1}{\alpha_2}(X_{n+1}b_{n-1,2}+X_nb_{n-2,2}+W_nb_{n-1,1}).\label{4-3-15}
\end{align}
Furthermore, we can deduce from ~\eqref{4-3-6}, ~\eqref{4-3-12} and ~\eqref{4-3-13}
\begin{align}
&-(m-1)\alpha\sum_{j=0}^{n}b_{j,-1}b_{j-1,1}+\sum_{j=0}^{n-1}(X_jb_{j-2,2}-X_{j+2}b_{j,2})b_{j,0}\nonumber\\
=&-\sum_{j=0}^{n}b_{j,-1}(X_{j+2}b_{j-1,3}-X_{j-1}b_{j-3,3}-W_{j-1}b_{j-2,2}+W_{j+1}b_{j-1,2})\nonumber\\
&+\sum_{j=0}^{n-1}(X_jb_{j-2,2}-X_{j+2}b_{j,2})b_{j,0}\nonumber\\	
=&-\alpha_1(\sum_{j=0}^{n}b_{j,-1}b_{j-1,-1}b_{j-2,2}-\sum_{j=1}^{n}b_{j+1,-1}b_{j,-1}b_{j-1,2})\nonumber\\
&+\sum_{j=0}^{n-1}X_jb_{j,0}b_{j-2,2}-\sum_{j=0}^{n}X_{j+1}b_{j+1,0}b_{j-1,2}+\sum_{j=0}^{n}X_jb_{j-1,0}b_{j-2,2}\nonumber\\
&-\sum_{j=0}^{n}X_{j+1}b_{j,0}b_{j-1,2}+\sum_{j=0}^{n}X_jb_{j-2,0}b_{j-2,2}-\sum_{j=0}^{n-1}X_{j+2}b_{j,0}b_{j,2}\nonumber\\
&-\alpha_2\sum_{j=0}^{n}(b_{j-2,-1}b_{j-1,-1}b_{j,-1}b_{j-3,3}-b_{j,-1}b_{j+1,-1}b_{j+2,-1}b_{j-1,3})\nonumber\\		
=&\alpha_1b_{n,-1}b_{n+1,-1}b_{n-1,2}+\alpha_2b_{n,-1}b_{n+1,-1}(b_{n+2,-1}b_{n-1,3}\nonumber\\
&+b_{n-1,-1}b_{n-2,3})-X_nb_{n,0}b_{n-2,2}-X_{n+1}b_{n-1,2}(b_{n-1,0}+b_{n,0}+b_{n+1,0}). \label{4-3-16}
\end{align}
By inserting \eqref{4-3-12}, \eqref{4-3-13}, \eqref{4-3-15} and \eqref{4-3-16} into  \eqref{4-3-14}, the resulting expression becomes
\begin{align}
&\alpha^2_2b_{n,-1}b_{n+1,-1}b_{n,1}b_{n-1,1}\nonumber\\
=&\frac{n(n-1)m\alpha^2}{2}((m+1)\alpha+\beta_0)-\alpha_2\alpha mb_{n,-1}b_{n-1,1}\nonumber\\		&-\alpha_2(b_{n,-1}b_{n-1,-1}b_{n-2,2}(\alpha_2b_{n,0}+\alpha_1)+b_{n+1,-1}b_{n,-1}b_{n-1,2}\nonumber\\
&\times(\alpha_2(b_{n+1,0}+b_{n,0}+b_{n-1,0})+2\alpha_1))-\alpha^2_2b_{n,-1}b_{n+1,-1}(b_{n+2,-1}b_{n-1,3}\nonumber\\
&+b_{n-1,-1}b_{n-2,3})-b_{n,-1}b_{n-1,1}(\alpha_2b_{n,0}+\alpha_1)(\alpha_2(b_{n,0}+b_{n-1,0})+\alpha_1).\label{4-3-17}
\end{align}
In addition, by taking $l=1,2,\ldots, m-1$ from \eqref{4-3-3}, we have
\begin{align}
b_{n,l}=&-\frac{\alpha_2}{\alpha (m-l)}(b_{n+l+2,-1}b_{n+l+1,-1}b_{n,l+2}-b_{n,-1}b_{n-1,-1}b_{n-2,l+2}\nonumber\\& -b_{n,-1}b_{n-1,l+1}(b_{n,0}+b_{n-1,0}+\frac{\alpha_1}{\alpha_2})\nonumber\\&+b_{n+l+1,-1}b_{n,l+1}(b_{n+l+1,0}+b_{n+l,0}+\frac{\alpha_1}{\alpha_2})), \quad l=1,2\cdots m-1,\label{4-3-18}
\end{align}
from which it follows that all $b_{n,j}$ can be represented by $b_{n,0}$ and $b_{n,-1}$, by noting that $b_{n,m}=1$ and $b_{n,j}=0,j>m$ or $j<-1$. 

The desired equation \eqref{4-3-20} is obtained by \eqref{4-3-9}, \eqref{4-3-17} with the help of the transformation
\begin{align*}
x_n=\frac{\sqrt{2}m\alpha}{\alpha_1+\alpha_2b_{n,0}},\quad y_n=-\frac{\alpha_2}{\alpha m}b_{n,-1}-n 
\end{align*}
in conjunction with \eqref{4-3-18}.
    
\end{proof}

\begin{remark}
    It's worth noting that, when $m=1$, we get from \eqref{4-3-20}
\begin{equation*}
		\begin{cases}
			y_{n+1}+y_{n}=\frac{2\alpha}{\alpha_2x_n}(-\frac{\alpha_1}{\sqrt{2}\alpha}+\frac{1}{x_n})+\frac{\beta_0}{\alpha}-1,\\
			\frac{-\alpha_2}{2\alpha}x_nx_{n-1}=\frac{y_n+n}{y^2_n},
		\end{cases}
\end{equation*}
{{which can be identified as the d-P$_\text{II}$ on the $E_6^{(1)}$ Sakai surface family \cite{kajiwara2017geometric,sakai2001rational, yue20261} and related to the asymmetric d-P$_{\text{I}}$ via a Miura transform \cite{ramani2000quadratic,yue20261}.}} 
In a more general sense, \eqref{4-3-20} can be regarded as a general class of d-P equations including the asymmetric d-P$_{\text{I}}$ and and d-P$_{\text{II}}$. When $m=2$, one can derive a discrete system with a relatively complicated form from \eqref{4-3-20}. We mention that, in a recent paper \cite{wang2025hard}, Wang and Xu obtained new limiting correlation kernels for the case of $m=2$, which results in the Chazy I equation that has a P$_\text{IV}$ reduction.
\end{remark}

\subsection{Realization of stationary reduction}\label{1masr}
In the previous two subsections, we performed a nonisospectral deformation on the $(1,m)$-type bi-OPs, resulting in the nonisospectral BM lattice. Then, we obtained  a family of d-P-type equations together with the Lax pair by applying the stationary reduction to the Lax pair for the nonisospectral BM lattice.

In the following, we will construct an explicit weight function under the assumption of the time evolution \eqref{4-2-1} for the moments, with the purpose of demonstrating the validity of the above mentioned stationary reduction process.

\begin{lemma}
Under the assumption  \eqref{lu3} together with $\frac{\alpha_2}{\alpha}<0$, define the moments as
\begin{align}
c_j(t)=\int^{+\infty}_0z^j(0)e^{j\alpha t}e^{\frac{\alpha_1}{m\alpha}z^m(0)e^{m\alpha t}+\frac{\alpha_2}{2m\alpha}z^{2m}(0)e^{2m\alpha t}}dz(0).\label{4-4-3}
\end{align}
Then the moments simultaneously satisfy the time evolution \eqref{4-2-1} and
\begin{align}
\frac{d}{dt}c_j(t)=-\alpha c_j(t).\label{4-4-4}
\end{align}
\end{lemma}
\begin{proof}
    Firstly, we attempt to find a weight function so that the moments satisfy the
evolution relation \eqref{4-2-1}.
Recall that \eqref{lu3} gives
\begin{align*}
z=z(0)e^{\alpha t},
\end{align*}
where $z(0)$ represents the spectral parameter at the initial time. In the case of the integral interval $(0,+\infty)$, the moments can be written as
\begin{align*}
c_j(t)=&\int^{+\infty}_0z^j(t)d\mu(z;t)=\int^{+\infty}_0z^j(0)e^{j\alpha t}f(z(0);t)dz(0),
\end{align*}
from which, we have
\begin{align*}
\frac{d}{dt}c_j(t)=\alpha jc_j(t)+\int^{+\infty}_0z^j(0)e^{j\alpha t}\frac{df(z(0);t)}{dt}dz(0).
\end{align*}
It is not hard to see that the moments satisfy the time evolution relation \eqref{4-2-1} as long as $f(z(0);t)$ satisfies the time evolution
\begin{align}
\frac{d}{dt}f(z(0);t)=(\alpha_1 z^m(0)e^{m\alpha t}+\alpha_2 z^{2m}(0)e^{2m\alpha t})f(z(0);t).
\end{align}
Consequently, it is reasonable to set
$$f(z(0);t)=e^{\frac{\alpha_1}{m\alpha}z^m(0)e^{m\alpha t}+\frac{\alpha_2}{2m\alpha}z^{2m}(0)e^{2m\alpha t}}$$
so that the concrete moments are given as \eqref{4-4-3}. The remaining part is to confirm that the concrete moments given by \eqref{4-4-3} also satisfy the evolution relation \eqref{4-4-4}.

By observing the fact at the boundary
\begin{align*}
\lim_{z(0)\to 0}z(0)f(z(0);t)=\lim_{z(0)\to +\infty}z(0)f(z(0);t)=0,
\end{align*}
we perform integration by parts to get
\begin{align*}
c_j(t)
=&-\int^{+\infty}_0z(0)jz^{j-1}(0)e^{j\alpha t}f(z(0);t)dz(0)\notag\\
&-\int^{+\infty}_0z^{j+1}(0)\frac{\alpha_1}{\alpha}z^{m-1}(0)e^{m\alpha t}
e^{j\alpha t}f(z(0);t)dz(0)\notag\\
&-\int^{+\infty}_0z^{j+1}(0)\frac{\alpha_2}{\alpha}z^{2m-1}(0)e^{2m\alpha t}
e^{j\alpha t}f(z(0);t)dz(0)\notag\\
=&-jc_j-\frac{\alpha_1}{\alpha}c_{j+m}-\frac{\alpha_2}{\alpha}c_{j+2m},
\end{align*}
which yields
$$c_{j+2m}=\frac{1}{\alpha_2}(-\alpha (j+1)c_j-\alpha_1c_{j+m}).$$
By inserting the above equation into \eqref{4-2-1}, we can arrive at the evolution relation \eqref{4-4-4}.
\end{proof}

Using the above lemma, we immediately have the following conclusion.
\begin{theorem}\label{th7}
Under the definition of the moments \eqref{4-4-3}, we have
\begin{align*}
\frac{d}{dt}b_{n,j}(t)=0,\qquad \text{and} \qquad
\frac{d}{dt}\gamma_{n,j}(t)=0,
\end{align*}
where $\{b_{n,j}\}$ are the recurrence coefficients in \eqref{4-1-2} for the monic $(1,m)$-type bi-OPs  and $\{\gamma_{n,j}\}$ are the coefficients of $(1,m)$-type bi-OPs with the expansions $P_n(z)=\sum_{j=0}^n\gamma_{n,j}z^j$.
\end{theorem}

\begin{proof}
Using the evolution relation \eqref{4-4-4} satisfied by the moments \eqref{4-4-3}, we will now prove that the coefficients and the recurrence coefficients of $(1,m)$-type bi-OPs are not dependent of time $t$. Firstly, it can be easily shown that 
$$\frac{d}{dt}\tau_n=-n\alpha\tau_n,$$
and 
$$\frac{d}{dt}h_n=\frac{d}{dt}\left(\frac{\tau_{n+1}}{\tau_n}\right)=-\alpha h_n.$$
Assume that 
\begin{align}
&P_n(z)=z^n+\gamma_{n,n-1}z^{n-1}+\gamma_{n,n-2}z^{n-2}+\cdots+\gamma_{n,0},\nonumber\\
&Q_n(z^{m})=(z^{m})^n+\beta_{n,n-1}(z^{m})^{n-1}+\beta_{n,n-2}(z^{m})^{n-2}+\cdots+\beta_{n,0}.\nonumber
\end{align}
Then the determinant representation of $P_n(z)$ provided in the previous context tells us
\begin{align*}
\gamma_{n,j}=\frac{T_{n,j}}{\tau_n},\quad j=0,\ldots,n-1,
\end{align*}
where
\begin{align*}
T_{n,j}=(-1)^{n+j}\det\begin{pmatrix}
c_{mp+q}
\end{pmatrix}_{q=0,1,\ldots,n,\,q\neq j}^{p=0,1,\ldots,n-1}.
\end{align*}
By making use of \eqref{4-4-4}, we can deduce that
$$\frac{d}{dt}T_{n,j}=-n \alpha T_{n,j}.$$
from which, it follows that $\gamma_{n,j}$ is independent of $t$.

Similarly, it can also be inferred that $\frac{d}{dt}\beta_{n,j}=0$.

Furthermore, by taking the inner product of both sides of the recurrence relation \eqref{4-1-2} with $Q_{n+j}(z^m)$, we have
\begin{align}
b_{n,j}h_{n+j}=&\int^{+\infty}_0z^mP_n(z)Q_{n+j}(z^m)w(z(t);t)dz,\nonumber\\
=&\frac{1}{\tau_n}\sum_{l=0}^{n+j}\beta_{n+j,l}M_{n+1,l},\nonumber
\end{align}
for $j=0,1,\ldots,m-1$,
where
\begin{align*}
M_{n+1,l}=\det\left(
\begin{array}{c}
c_{mp+q}\\
c_{(l+1)m+q}
\end{array}
\right)_{q=0,1,\ldots,n}^{p=0,1,\ldots,n-1}.
\end{align*}
It is not hard to obtain from \eqref{4-4-4}  $$\frac{d}{dt}M_{n+1,l}=-(n+1) \alpha M_{n+1,l},$$ from which we have
\begin{align*}
\frac{d}{dt}b_{n,j}=\frac{d}{dt}\left(\frac{1}{h_{n+j}\tau_n}\sum_{l=0}^{n+j}\beta_{n+j,l}M_{n+1,l}\right)=0.
\end{align*}
\end{proof}

In summary, since the coefficients are independent of time $t$, we have $\frac{\partial }{\partial t}P_n(z)=\frac{\partial }{\partial t}Q_n(z^m)=0$, the stationary reduction of the Lax pair of the nonisospectral BM equation is valid so that the Lax pair of the family of d-P-type equations is derived. And, we also conclude that the nonisospectral equation \eqref{4-2-9} does indeed admit the stationary reduction to \eqref{4-3-3} related to the  family of d-P-type equations ~\eqref{4-3-20}.

\section{Nonisospectral deformation of generalized Laurent bi-OPs and d-P}\label{Laurent}

Laurent bi-OPs (LBOPs) appear in problems related to the two-point Pad\'{e} approximations \cite{jones2006survey}, which can be regarded as a generalization of the OPs on the unit circle \cite{kharchev1997faces} whose moments satisfy certain symmetry conditions. There has been much literatures on the connection between LBOPs and integrable systems \cite{kharchev1997faces,tsujimoto2009elliptic,vinet2001spectral,zhedanov1998classical}. A link between LBOPs and d-P equations was established in \cite{yue2022laurent}. Recently, generalizations of LBOPs were proposed and the corresponding integrable equations were also investigated in  \cite{wang2022generalization}, which was partially motivated by the study of the so-called coupled pentagram map in \cite{wang2025pentagram}.
See also more recent results \cite{gharakhloo2023modulated,ito2023generalized} related to the generalized LBOPs.
In this section, we study the related problem of nonisospectral deformation of the generalized LBOPs and d-P-type equations.

\subsection{Generalized Laurent bi-OPs}\label{introductionLaurent}
Let's first give a brief review on the generalized LBOPs proposed in \cite{wang2022generalization}. For a fixed $k\in\mathbb{N}^+$, we define the inner product
\begin{align}\label{5-1-1-1}
	\langle f(z)|d\mu|g(z) \rangle_{k}
	=\int f(z)g\left(\frac{1}{z^k}\right)d \mu(z),
\end{align}
where $\mu$ is a formal measure such that all moments
$$c_i=\int z^id \mu(z), \ i\in \mathbb{Z}$$
exist.
In many cases, \eqref{5-1-1-1} can be written in the following form
\begin{align}\label{5-1-1}
	\langle f(z)|d\mu|g(z) \rangle_{k}
	=\int f(z)g\left(\frac{1}{z^k}\right)w(z)dz,
\end{align}
and the moments can be rewritten as
$$c_i=\int z^iw(z)d z, \ i\in \mathbb{Z},$$
where $w(z)$ is the weight function.

Two families of monic polynomials $\{P_n(z)\}_{n\in\mathbb{N}}$ and $\{Q_m(z)\}_{m\in\mathbb{N}}$ with $\deg(P_n(z))=\deg(Q_n(z))$ $=n$ are called generalized LBOPs with respect to the bilinear 2-form $\langle \cdot| \cdot\rangle_{k}$, if they satisfy biorthogonality condition
\begin{align}\label{5-1-2}
	\langle P_n(z)|d\mu|Q_m(z)\rangle_{k}=h_n \delta_{nm},\quad h_n\neq 0.
\end{align}	
It is not difficult to find that the biorthogonality condition \eqref{5-1-2} can be equivalently written as
\begin{align*}
	\langle P_n(z)|d\mu|z^m\rangle_{k}=0,\quad m=0,\,1,\,\ldots,\,n-1,
\end{align*}
for $\{P_n(z)\}_{n\in\mathbb{N}}$ and
\begin{align*}
	\langle z^m|d\mu|Q_n(z)\rangle_{k}=0.\quad m=0,\,1,\,\ldots,\,n-1,
\end{align*}
for $\{Q_n(z)\}_{n\in\mathbb{N}}$.

From the biorthogonality condition, it can be deduced that the generalized LBOPs admit the following determinant expressions
    \begin{align}\label{lbop}
	&P_n(z)=\frac{1}{\tau_n^{(0)}}
	\det\begin{pmatrix}
c_{i-kj}\quad  z^i
\end{pmatrix}_{j=0,1,\ldots,n-1}^{i=0,1,\ldots,n},\quad
	&Q_n(z)=\frac{1}{\tau_n^{(0)}}
	\det\left(
	\begin{array}{c}
c_{i-kj}\\
  z^j
\end{array}
\right)_{j=0,1,\ldots,n}^{i=0,1,\ldots,n-1},
 \end{align}
where $\tau_n^{(l)}=\det(c_{l+i-kj})_{i,j=0}^{n-1}\neq 0$.
Moreover, the generalized LBOPs satisfy the recurrence relation of $(k+2)$-term as follows
\begin{align}\label{5-1-3}
z^k(P_n(z)+a_nP_{n-1}(z))
=P_{n+k}(z)+b_{n,k-1}P_{n+k-1}(z)+\cdots+b_{n,0}P_n(z),
\end{align}
with 
$$
a_n=-\frac{\langle P_{n}(z)|d\mu|z^{-1}\rangle_k}{\langle P_{n-1}(z)|d\mu|z^{-1}\rangle_k}=\frac{\tau_{n-1}^{(0)}\tau_{n+1}^{(k)}}{\tau_{n}^{(k)}\tau_n^{(0)}}.
$$
In the following we mainly focus on the family of polynomials $\{P_n(z)\}_{n\in\mathbb{N}}$.

\subsection{Nonisospectral generalized mixed relativistic Toda lattice }\label{nrt}
In \cite{wang2022generalization}, it is shown that isospectral deformation of the generalized LBOPs lead to two different generalizations of the relativistic Toda (rToda) lattice. In this section we consider nonisospectral deformation to the generalized LBOPs, which actually lead to a nonisospectral generalized mixed rToda lattice.

Suppose that the spectral parameter $z$ satisfies the time evolution \eqref{lu3}, and we select a weight function $w(z;t)$ such that the moments satisfy the evolution relation
\begin{align}
\frac{d}{dt}c_j(t)=\alpha jc_j(t)+\alpha_1 c_{j+k}(t)+\alpha_2 c_{j-k}(t).\label{5-2-1}
\end{align}
Since we have 
\begin{align*}
\frac{d}{dt}c_j(t)=&\int  z^j\left(\alpha j w(z;t)+\frac{d}{dt}w(z;t)+\alpha w(z;t)\right)dz\nonumber\\
=&\alpha jc_j(t)+\int  z^j\left(\frac{d}{dt}w(z;t)+\alpha w(z;t)\right)dz,
\end{align*}
by taking the derivative of the moments with respect to $t$, it is reasonable to set
\begin{align}
\frac{d}{dt}w(z;t)+\alpha w(z;t)=(\alpha_1 z^k+\alpha_2 z^{-k})w(z;t).\label{5-2-2}
\end{align}
It is noted that a weight function satisfying the above evolution relation is constructed in \eqref{lbop-mom}.

\begin{lemma}\label{buthLBOPs}
Under the assumptions  \eqref{lu3} and \eqref{5-2-2}, the monic generalized LBOPs $\{P_n(z)\}_{n\in\mathbb{N}}$ in \eqref{lbop} satisfy the time evolution
\begin{align}
&\frac{d}{dt}P_n(z;t)+a_n\frac{d}{dt}P_{n-1}(z;t)\nonumber\\
=&n\alpha P_n(z;t)-\alpha_2\frac{a_{n-1}a_n}{b_{n-1,0}}P_{n-2}(z;t)\nonumber\\
&-\left(-(n-1)\alpha a_n+\alpha_1\left(\prod^{n}_{j=n-k}(-a_j)+\sum_{i=1}^k\prod^n_{j=n-i+1}(-a_j)b_{n-i,i-1}\right)+\alpha_2\frac{a_n}{b_{n,0}}\right)P_{n-1}(z;t).\label{5-2-4}
\end{align}
\end{lemma}

\begin{proof}
For $m=0,1,\ldots,n-1$, the biorthogonality condition \eqref{5-1-2} gives
\begin{align*}
\int P_n(z;t)Q_m(z^{-k};t)w(z;t)dz=0.
\end{align*}
By taking the derivative of the above equation with respect to time $t$, we have
\begin{align}
0=&\int \left(\frac{d}{dt}(P_n(z;t))Q_m(z^{-k};t)+P_n(z;t)\frac{d}{dt}Q_m(z^{-k};t)\right)w(z;t)dz\nonumber\\
&+\int P_n(z;t)Q_m(z^{-k};t)\left(\frac{d}{dt}w(z;t)+\alpha w(z;t)\right)dz\nonumber\\
=&\int \left(\frac{d}{dt}(P_n(z;t))+(\alpha_1 z^k+\alpha_2 z^{-k})P_n(z;t)\right)Q_m(z^{-k};t)w(z;t)dz\nonumber\\
=&\int \left(\frac{d}{dt}P_n(z;t)-\alpha_1a_nz^k P_{n-1}\right)Q_m(z^{-k};t)w(z;t)dz\nonumber\\
&+\frac{\alpha_2}{b_{n,0}}\int_0^{+\infty}a_nP_{n-1}(z;t)Q_m(z^{-k};t)w(z;t)dz\nonumber\\
=&\int \left(\frac{d}{dt}P_n(z;t)+\alpha_1a_na_{n-1}z^k P_{n-2}(z;t))\right)Q_m(z^{-k};t)w(z;t)dz\nonumber\\
&+\alpha_1(-a_n)b_{n-1,0}\int P_{n-1}(z;t)Q_m(z^{-k};t)w(z;t)dz\nonumber\\
&+\frac{\alpha_2}{b_{n,0}}\int a_nP_{n-1}(z;t)Q_m(z^{-k};t)w(z;t)dz\nonumber\\
=&\quad \cdots\nonumber\\
=&\int \left(\frac{d}{dt}P_n(z;t)-\alpha_1\prod^n_{j=n-k+1}(-a_j)(-z^k P_{n-k}(z;t))\right)Q_m(z^{-k};t)w(z;t)dz\nonumber\\
&+\alpha_1\sum_{i=1}^{k-1}\sum_{l=1}^i\prod^n_{j=n-i+1}(-a_j)b_{n-i,i-l}\int P_{n-l}Q_m(z^{-k};t)w(z;t)dz\nonumber\\
&+\frac{\alpha_2}{b_{n,0}}\int a_n P_{n-1}(z;t)Q_m(z^{-k};t)w(z;t)dz\nonumber\\
=&\int \left(\frac{d}{dt}P_n(z;t)-\alpha_1\prod^n_{j=n-k+1}(-a_j)(P_n(z;t)-z^k P_{n-k}(z;t))\right)Q_m(z^{-k};t)w(z;t)dz\nonumber\\
&+\alpha_1\sum_{i=1}^{k-1}\sum_{l=1}^i\prod^n_{j=n-i+1}(-a_j)b_{n-i,i-l}\int P_{n-l}(z;t)Q_m(z^{-k};t)w(z;t)dz\nonumber\\
&+\frac{\alpha_2}{b_{n,0}}\int a_n P_{n-1}(z;t)Q_m(z^{-k};t)w(z;t)dz,\label{5-2-3}
\end{align}
where we used the recurrence relation \eqref{5-1-3} and \eqref{5-2-2}.

Upon setting
$$\frac{d}{dt}P_n(z;t)=\sum_{i=0}^n\gamma_iP_i(z;t),$$
where $\gamma_i$ are some coefficients to be determined,
we obtain
\begin{align}
&\frac{d}{dt}P_n(z;t)-\alpha_1\prod^n_{j=n-k+1}(-a_j)(P_n(z;t)-z^k P_{n-k}(z;t))\nonumber\\
&+\alpha_1\sum_{i=1}^{k-1}\sum_{l=1}^i\prod^n_{j=n-i+1}(-a_j)b_{n-i,i-l}P_{n-l}(z;t)+\frac{\alpha_2}{b_{n,0}}a_nP_{n-1}(z;t)\nonumber\\
=&\gamma_nP_n(z;t)+\sum_{i=0}^{n-1}\gamma_iP_i(z;t)-\alpha_1\prod^n_{j=n-k+1}(-a_j)(P_n(z;t)-z^k P_{n-k}(z;t))\nonumber\\
&+\alpha_1\sum_{i=1}^{k-1}\sum_{l=1}^i\prod^n_{j=n-i+1}(-a_j)b_{n-i,i-l}P_{n-l}(z;t)+\frac{\alpha_2}{b_{n,0}}a_nP_{n-1}(z;t)\nonumber.
\end{align}

Comparing the coefficients of the highest powers of $z$ on both sides of the above equation shows that $\gamma_n = n\alpha $. Since the formula \eqref{5-2-3} hold for $ m=0,1,\ldots,n-1 $, according to the orthogonality condition, it is not hard to see
\begin{align}
&\frac{d}{dt}P_n(z;t)\nonumber\\
=&n\alpha P_n(z;t)-\frac{\alpha_2}{b_{n,0}}a_n P_{n-1}(z;t)\nonumber\\
&+\alpha_1\left(\prod^n_{j=n-k+1}(-a_j)(P_n(z;t)-z^k P_{n-k}(z;t))-\sum_{i=1}^{k-1}\sum_{l=1}^i\prod^n_{j=n-i+1}(-a_j)b_{n-i,i-l}P_{n-l}(z;t)\right),\nonumber
\end{align}
from which in conjunction with the corresponding expression for $a_n\frac{d}{dt}P_{n-1}(z;t)$,
we can eventually arrive at \eqref{5-2-4}.
\end{proof}

Now we are ready to derive an integrable lattice from the compatibility condition between \eqref{5-1-3} and \eqref{5-2-4}. 
\begin{theorem}
    Under the assumptions  \eqref{lu3} and \eqref{5-2-2}, the recurrence coefficients $\{a_n\}$ and $\{b_{n,j}\}$ in \eqref{5-1-3} for the monic generalized LBOPs satisfy the following ODE system
    \begin{subequations}\label{5-2-6}
\begin{align}
\frac{d a_n}{dt}=&\alpha a_n+\alpha_1(a_ns_n-\widetilde{f}_n)+\alpha_2 a_n\left(\frac{1}{b_{n,0}}-\frac{1}{b_{n-1,0}}\right),\\
\frac{d b_{n,0}}{dt}=&k \alpha b_{n,0}+\alpha_1(b_{n,0}s_n-g_{n,0}\widetilde{f}_{n+1})+\alpha_2\left(\frac{a_{n+1}}{b_{n+1,0}}b_{n,1}-\frac{a_n}{b_{n-1,0}}b_{n-1,1}\right),\\
\frac{d b_{n,i}}{dt}=&(k-i) \alpha b_{n,i}+\alpha_1\left(b_{n,i}s_n-g_{n,i}\widetilde{f}_{n+i+1}+(1-g_{n,k-1})\sum_{k=0}^{i-1}\prod_{j=i-l+1}^{k}(-a_{n+j})b_{n+i-l,l}\right)\nonumber\\
&+\alpha_2\left(\frac{a_{n+i+1}}{b_{n+i+1,0}}b_{n,i+1}-\frac{a_n}{b_{n-1,0}}b_{n-1,i+1}\right),\quad
i=1,\ldots,k-2,\\
\frac{d b_{n,k-1}}{dt}=&\alpha b_{n,k-1}+\alpha_1\left(b_{n,k-1}s_n-g_{n,k-1}\widetilde{f}_{n+k}+(1-g_{n,k-1})\sum_{k=0}^{k-2}\prod_{j=k-l}^{k}(-a_{n+j})b_{n+k-l-1,l}\right)\nonumber\\
&+\alpha_2\left(\frac{a_{n+k}}{b_{n+k,0}}-\frac{a_n}{b_{n-1,0}}\right)
\end{align}
\end{subequations}
with
\begin{align*}
&\widetilde{f}_n=-\left(\prod^{k}_{j=0}(-a_{n-j})+\sum_{i=0}^{k-1}\prod^i_{j=0}(-a_{n-j})b_{n-i-1,i}\right),\\
&g_{n,i}=-\sum_{l=0}^{i}\prod_{j=l+1}^{i+1}\left(-\frac{1}{a_{n+j}}\right)b_{n,l},\\
&s_n=\prod_{j=1}^k(-a_{n+j})+\sum_{l=0}^{k-1}\prod_{j=1}^l(-a_{n+j})b_{n,l}.
\end{align*}
\end{theorem}
\begin{proof}
 \eqref{5-1-3} and \eqref{5-2-4} can be written in matrix form as
\begin{align}
z^k AP=BP,\quad A\frac{dP}{dt}=LP,\label{5-2-5}
\end{align}
where $P(z;t)=(P_0(z;t),P_1(z;t),P_2(z;t),\ldots)^T$ and 
\begin{align*}
&
 A=\left(
\begin{array}{ccccc}
 1 &  &  &     \\
 a_1 & 1 &  &     \\
  & a_2 & 1 &     \\
  &  &  \ddots & \ddots  \ \\
\end{array}
\right), \qquad
L=\left(
\begin{array}{ccccc}
 0 &  &  &     \\
 f_1-\frac{\alpha_2a_1}{b_{1,0}} & \alpha &  &     \\
 -\frac{\alpha_2a_1a_2}{b_{1,0}} & f_2-\frac{\alpha_2a_2}{b_{2,0}} & 2\alpha &     \\
 & -\frac{\alpha_2a_2a_3}{b_{2,0}} & f_3-\frac{\alpha_2a_3}{b_{3,0}} & 3\alpha &     \\
  &   &  \ddots &  \ddots & \ddots  \ \\
\end{array}
\right),\\
&B=\left(
\begin{array}{ccccccccc}
 b_{0,0}  &\cdots  &b_{0,k-1}  &1 &  \\
  & b_{1,0}&\cdots  &b_{1,k-1}  &1  &   \\
   & & b_{2,0}&\cdots  &b_{2,k-1}  &1  &   \\
  &&  & \ddots  & & \ddots& \ddots \\
  \end{array}
\right)
\end{align*}
with $f_n=(n-1)\alpha a_n-\alpha_1\left(\prod^{k}_{j=0}(-a_{n-j})+\sum_{i=0}^{k-1}\prod^i_{j=0}(-a_{n-j})b_{n-i-1,i}\right)$. The compatibility condition of \eqref{5-2-5} yields
\begin{align*}
\frac{d}{d t}(A^{-1}B)=(k\alpha I+ A^{-1}L)A^{-1}B-A^{-1}BA^{-1}L,
\end{align*}
where $I$ represents the identity matrix. After a lengthy calculation, we get the desired explicit expression of the nonisospectral integrable lattice \eqref{5-2-6}.
   
\end{proof}

The integrable lattice \eqref{5-2-6} is a nonisospectral generalized mixed rToda lattice incorporating positive and negative flows. 
To get a better understanding on this generalized lattice, we present some special cases below.

When $k=1$, we obtain a nonisospectral generalization of the mixed rToda lattice 
\begin{align*}
\begin{split}
\frac{d a_n}{dt}&=\alpha a_n+\alpha_1 a_n(a_{n-1}-a_{n+1}+b_{n,0}-b_{n-1,0})+\alpha_2 a_n\left(\frac{1}{b_{n,0}}-\frac{1}{b_{n-1,0}}\right),\\
\frac{d b_{n,0}}{dt}&=\alpha b_{n,0}+\alpha_1 b_{n,0}(a_n-a_{n+1})+\alpha_2 \left(\frac{a_{n+1}}{b_{n+1,0}}-\frac{a_n}{b_{n-1,0}}\right),
\end{split}
\end{align*}
in the case of $\alpha\neq0$.
It is easy to see that when $\alpha=\alpha_2=0$, the above equation reduces to the first positive flow of the rToda lattice and it gives the first negative flow of the rToda lattice when $\alpha=\alpha_1=0$.

When $k=2$, we get the following integrable lattice
\begin{align*}
\begin{split}
\frac{d a_n}{dt}
&=\alpha a_n+\alpha_1 a_n((T^3-1)a_{n-1}a_{n-2}+(1-T^2)a_{n-1}b_{n-2,1}\\
&\quad +(T-1)b_{n-1,0})+\alpha_2 a_n\left(\frac{1}{b_{n,0}}-\frac{1}{b_{n-1,0}}\right),\\
\frac{d b_{n,0}}{dt}
&=2 \alpha b_{n,0}+\alpha_1b_{n,0}((T^2-1)a_{n-1}a_n+(1-T)a_nb_{n-1,1})\\
&\quad+\alpha_2\left(\frac{a_{n+1}}{b_{n+1,0}}b_{n,1}-\frac{a_n}{b_{n-1,0}}b_{n-1,1}\right),\\
\frac{d b_{n,1}}{dt}
&=\alpha b_{n,1}+\alpha_1\left(b_{n,1}a_{n+1}(a_{n+2}-a_n)+a_nb_{n,0}-a_{n+2}b_{n+1,0}\right)+\alpha_2\left(\frac{a_{n+2}}{b_{n+2,0}}-\frac{a_n}{b_{n-1,0}}\right),
\end{split}
\end{align*}
where $T$ is the shift operator on $n$.

When $k=3$, we have
\begin{align*}
\begin{split}
\frac{d a_n}{dt}
&=\alpha a_n-\alpha_1a_n((T^4-1)a_{n-1}a_{n-2}a_{n-3}+(1-T^3)a_{n-1}a_{n-2}b_{n-3,2}\\
&\quad+(T^2-1)a_{n-1}b_{n-2,1}+(1-T)b_{n-1,0})+\alpha_2 a_n\left(\frac{1}{b_{n,0}}-\frac{1}{b_{n-1,0}}\right),\\
\frac{d b_{n,0}}{dt}
&=3 \alpha b_{n,0}-\alpha_1 b_{n,0}((T^3-1)a_na_{n-1}a_{n-2}+(1-T^2)a_{n-1}a_{n}b_{n-2,2}\\
&\quad+(T-1)a_nb_{n-1,1})+\alpha_2\bigg(\frac{a_{n+1}}{b_{n+1,0}}b_{n,1}-\frac{a_n}{b_{n-1,0}}b_{n-1,1}\bigg),\\
\frac{d b_{n,1}}{dt}
&=2 \alpha b_{n,1}+\alpha_1(-b_{n,1}a_{n+1}(a_{n+2}a_{n+3}-a_{n-1}a_{n}-a_{n+2}b_{n,2}+a_nb_{n-1,2})\\
&\quad-a_{n+2}b_{n+1,0}(b_{n,2}-a_{n+3})+a_nb_{n,0}(b_{n-1,2}-a_{n-1}))\\
&\quad+\alpha_2\bigg(\frac{a_{n+2}}{b_{n+2,0}}b_{n,2}-\frac{a_n}{b_{n-1,0}}b_{n-1,2}\bigg),\\
\frac{d b_{n,2}}{dt}
&=\alpha b_{n,2}+\alpha_1(-b_{n,2}a_{n+1}a_{n+2}(a_{n+3}-a_n)+a_{n+2}a_{n+3}b_{n+1,1}-a_na_{n+1}b_{n,1}\\
&\quad+a_nb_{n,0}-a_{n+3}b_{n+2,0})+\alpha_2\bigg(\frac{a_{n+3}}{b_{n+3,0}}-\frac{a_n}{b_{n-1,0}}\bigg).
\end{split}
\end{align*}

\subsection{d-P related to nonisospectral generalized mixed rToda}\label{rt to p}
Based on the Lax pair of the nonisospectral generalized mixed rToda lattice \eqref{5-2-5}, we implement stationary reduction to get 
\begin{align}\label{nglopdp}
z^kAP=\widehat{B}P, \qquad \frac{\partial}{\partial z}P=\widehat{M}P,
\end{align}
where 
$$\widehat{B}=B, \ \widehat{M}=M/\left(\frac{dz}{dt}\right)=\frac{1}{\alpha z }M,\ M=A^{-1}L.$$
Its compatibility condition will yield a family of d-P-type equations.

\begin{theorem}
 Under the assumptions  \eqref{lu3} and \eqref{5-2-2}, the recurrence coefficients $\{a_n\}$ and $\{b_{n,j}\}$ in \eqref{5-1-3} for the monic generalized LBOPs satisfy the following integrable difference system
 \begin{subequations}\label{5-3-1}
\begin{align}
&\alpha a_n+\alpha_1(a_ns_n-\widetilde{f}_n)+\alpha_2 a_n\left(\frac{1}{b_{n,0}}-\frac{1}{b_{n-1,0}}\right)=0,\\
&k \alpha b_{n,0}+\alpha_1(b_{n,0}s_n-g_{n,0}\widetilde{f}_{n+1})+\alpha_2\bigg(\frac{a_{n+1}}{b_{n+1,0}}b_{n,1}-\frac{a_n}{b_{n-1,0}}b_{n-1,1}\bigg)=0,\\
&(k-i) \alpha b_{n,i}+\alpha_1\big(b_{n,i}s_n-g_{n,i}\widetilde{f}_{n+i+1}+(1-g_{n,k-1})\sum_{k=0}^{i-1}\prod_{j=i-l+1}^{k}(-a_{n+j})b_{n+i-l,l}\big)\nonumber\\
&+\alpha_2\bigg(\frac{a_{n+i+1}}{b_{n+i+1,0}}b_{n,i+1}-\frac{a_n}{b_{n-1,0}}b_{n-1,i+1}\bigg)=0,\quad
i=1,\ldots,k-2,\\
&\alpha b_{n,k-1}+\alpha_1\big(b_{n,k-1}s_n-g_{n,k-1}\widetilde{f}_{n+k}+(1-g_{n,k-1})\sum_{k=0}^{k-2}\prod_{j=k-l}^{k}(-a_{n+j})b_{n+k-l-1,l}\big)\nonumber\\
&+\alpha_2\bigg(\frac{a_{n+k}}{b_{n+k,0}}-\frac{a_n}{b_{n-1,0}}\bigg)=0,
\end{align}
\end{subequations}
with the Lax pair \eqref{nglopdp}.
\end{theorem}

\begin{proof}
From the compatibility condition of \eqref{nglopdp}, we obtain
$$kz^{k-1}I+z^k\widehat{M}=\frac{\partial}{\partial z}(A^{-1}\widehat{B})+A^{-1}\widehat{B}\widehat{M},$$
from which the explicit expressions in \eqref{5-3-1} follows. 
\end{proof}
When $k=1$, it follows from \eqref{5-3-1} that
\begin{subequations}
\begin{align}
&\alpha a_n+\alpha_1 a_n(a_{n-1}-a_{n+1}+b_{n,0}-b_{n-1,0})+\alpha_2 a_n\left(\frac{1}{b_{n,0}}-\frac{1}{b_{n-1,0}}\right)=0,\label{hxdlt1}\\
&\alpha b_{n,0}+\alpha_1 b_{n,0}(a_n-a_{n+1})+\alpha_2 \left(\frac{a_{n+1}}{b_{n+1,0}}-\frac{a_n}{b_{n-1,0}}\right)=0.\label{hxdlt2}
\end{align}
\end{subequations}
Dividing the both sides of \eqref{hxdlt1} by $a_n$ and making a summation, we  can get
\begin{align}
\alpha_1(b_{n,0}-a_{n+1}-a_{n})+\alpha_2 \frac{1}{b_{n,0}}+n\alpha+\beta_0=0,\label{hxdlt3}
\end{align}
where $\beta_0=\alpha_1(a_1-b_{0,0})-\alpha_2 \frac{1}{b_{0,0}}$. Similarly, dividing the both sides of \eqref{hxdlt2} and making a summation gives
\begin{align}
\alpha_1a_{n+1}-\alpha_2 \frac{a_{n+1}}{b_{n,0}b_{n+1,0}}-n\alpha-\gamma_0=0,\label{hxdlt4}
\end{align}
where $\gamma_0=\alpha_1a_1-\alpha_2 \frac{a_1}{b_{0,0}b_{1,0}}$. Observe that \eqref{hxdlt4} produces the following equations
\begin{align}
&a_{n+1}=\frac{(n\alpha+\gamma_0)b_{n,0}b_{n+1,0}}{\alpha_1 b_{n,0}b_{n+1,0}-\alpha_2},\nonumber\\
&a_{n}=\frac{((n-1)\alpha+\gamma_0)b_{n-1,0}b_{n,0}}{\alpha_1 b_{n-1,0}b_{n,0}-\alpha_2},\nonumber
\end{align}
plugging which into \eqref{hxdlt3}, we get
\begin{align*}
\alpha_1\left(b_{n,0}-\frac{(n\alpha+\gamma_0)b_{n,0}b_{n+1,0}}{\alpha_1 b_{n,0}b_{n+1,0}-\alpha_2}-\frac{((n-1)\alpha+\gamma_0)b_{n-1,0}b_{n,0}}{\alpha_1 b_{n-1,0}b_{n,0}-\alpha_2}\right)+\alpha_2 \frac{1}{b_{n,0}}+n\alpha+\beta_0=0.
\end{align*}
By further simplifying it, we finally obtain
\begin{align*}
\frac{\sqrt{-\frac{\alpha_1}{\alpha_2}}\frac{n\alpha+\gamma_0}{\alpha_1}}{1+\frac{\sqrt{-\frac{\alpha_1}{\alpha_2}}}{b_{n+1,0}}\frac{\sqrt{-\frac{\alpha_1}{\alpha_2}}}{b_{n,0}}}+\frac{\sqrt{-\frac{\alpha_1}{\alpha_2}}\frac{(n-1)\alpha+\gamma_0}{\alpha_1}}{1+\frac{\sqrt{-\frac{\alpha_1}{\alpha_2}}}{b_{n,0}}\frac{\sqrt{-\frac{\alpha_1}{\alpha_2}}}{b_{n-1,0}}}=-\frac{\sqrt{-\frac{\alpha_2}{\alpha_1}}}{b_{n,0}}+\sqrt{-\frac{\alpha_1}{\alpha_2}}b_{n,0}+\sqrt{-\frac{\alpha_1}{\alpha_2}}\frac{n\alpha+\beta_0}{\alpha_1},
\end{align*}
which is nothing but the alternate d-P$_{\text{II}}$ (alt d-P$_{\text{II}}$) appearing in \cite{nijhoff1996study}
\begin{align*}
\frac{z_n}{x_{n+1}x_n+1}+\frac{z_{n-1}}{x_nx_{n-1}+1}=-x_n+\frac{1}{x_n}+z_n+\mu
\end{align*}
with
\begin{align}
z_n=\sqrt{-\frac{\alpha_1}{\alpha_2}}\frac{n\alpha+\gamma_0}{\alpha_1},\quad x_n=\frac{\sqrt{-\frac{\alpha_1}{\alpha_2}}}{b_{n,0}},\quad \mu=-\sqrt{-\frac{\alpha_1}{\alpha_2}}\frac{\gamma_0-\beta_0}{\alpha_1}.\nonumber
\end{align}

\begin{remark}
We obtain a generalized family of d-P equations including  the alt d-P$_{\text{II}}$ from a nonisospectral deformation of the generalized LBOPs. In fact, \eqref{5-3-1} together with its Lax pair is the stationary form of nonisospectral generalized mixed rToda lattice ~\eqref{5-2-6} and its Lax pair, which is associated with a nonisospectral deformation of the generalized LBOPs. 
\end{remark}

\subsection{Realization of stationary reduction}\label{rtasr}
In the previous subsections, we have obtained the nonisospectral generalized mixed rToda lattice, based on which, the d-P-type equations are obtained as stationary reduction of the nonisospectral equations.
In this subsection we will show the feasibility of stationary reduction from the perspective of solutions.

\begin{lemma}\label{th8}
Under the assumption  \eqref{lu3} together with $\frac{\alpha_2}{\alpha}>0$ and $\frac{\alpha_1}{\alpha}<0$, define the moments as
\begin{align}
c_j(t)=\int^{+\infty}_0z^j(0)e^{j\alpha t}e^{\frac{\alpha_1}{k\alpha}z^k(0)e^{k\alpha t}-\frac{\alpha_2}{k\alpha}z^{-k}(0)e^{-k\alpha t}}dz(0).\label{lbop-mom}
\end{align}
Then the moments simultaneously satisfy the time evolution \eqref{5-2-1} and
\begin{align}
\frac{d}{dt}c_j(t)=-\alpha c_j(t).\label{5-4-1}
\end{align}
\end{lemma}
\begin{proof}
    Based on \eqref{lu3}, we have
    \begin{align*}
z=z(0)e^{\alpha t},
\end{align*}
where $z(0)$ is the spectral parameter at initial time. In the case of the integral interval $(0,+\infty)$, the moments can be written in terms of $z(0)$ as
\begin{align*}
c_j(t)=&\int^{+\infty}_0z^jd\mu(z;t)=\int^{+\infty}_0z^j(0)e^{j\alpha t}f(z(0);t)dz(0),
\end{align*}
from which we get
\begin{align*}
\frac{d}{dt}c_j(t)=\alpha jc_j(t)+\int^{+\infty}_0z^j(0)e^{j\alpha t}\frac{df(z(0);t)}{dt}dz(0).
\end{align*}
In order to ensure that the moments satisfy the time evolution relation \eqref{5-2-1}, it is sufficient to set for $f(z(0); t)$ 
\begin{align*}
\frac{d}{dt}f(z(0);t)=(\alpha_1 z^k(0)e^{k\alpha t}+\alpha_2 z^{-k}(0)e^{-k\alpha t})f(z(0);t).
\end{align*}
Consequently, it is reasonable to define
$$f(z(0);t)=e^{\frac{\alpha_1}{k\alpha}z^k(0)e^{k\alpha t}-\frac{\alpha_2}{k\alpha}z^{-k}(0)e^{-k\alpha t}},$$
from which we get the specific expressions of the moments as \eqref{lbop-mom}.

The remaining part is to confirm the evolution relation \eqref{5-4-1}. Using the relations at the boundary
$$\lim_{z(0)\to0}z(0)f(z(0);t)=\lim_{z(0)\to +\infty}z(0)f(z(0);t)=0,$$
we integrate the moments by parts to obtain 
\begin{align*}
c_j(t)
=&-\int^{+\infty}_0z(0)jz^{j-1}(0)e^{j\alpha t}f(z(0);t)dz(0)\notag\\
&-\int^{+\infty}_0z^{j+1}(0)\frac{\alpha_1}{\alpha}z^{k-1}(0)e^{k\alpha t}
e^{j\alpha t}f(z(0);t)dz(0)\notag\\
&-\int^{+\infty}_0z^{j+1}(0)\frac{\alpha_2}{\alpha}z^{-k-1}(0)e^{-k\alpha t}e^{j\alpha t}f(z(0);t)dz(0)\notag\\
=&-jc_j-\frac{\alpha_1}{\alpha}c_{j+k}-\frac{\alpha_2}{\alpha}c_{j-k},
\end{align*}
combining which and the evolution equation ~\eqref{5-2-1}, we immediately get \eqref{5-4-1}. \end{proof}

\begin{theorem}
Under the definition of the moments \eqref{lbop-mom}, we have 
\begin{align*}
\frac{d}{dt}b_{n,j}(t)=\frac{d}{dt}a_n(t)=0\qquad \text{and} \qquad \frac{d}{dt}\gamma_{n,j}(t)=0,
\end{align*}
where $\{a_n\}$ and $\{b_{n,j}\}$ are the recurrence coefficients in \eqref{5-1-3} for the monic generalized LBOPs and  $\{\gamma_{n,j}\}$ are the coefficients of the generalized LBOPs with the expansions $P_n(z)=\sum_{j=0}^n\gamma_{n,j}z^j$.
\end{theorem}

\begin{proof}
Let
\begin{align}
P_n(z;t)=&\gamma_{n,n}z^n+\gamma_{n,n-1}z^{n-1}+\gamma_{n,n-2}z^{n-2}+\cdots+\gamma_{n,0},\nonumber\\
Q_n(z;t)=&\beta_{n,n}z^n+\beta_{n,n-1}z^{n-1}+\beta_{n,n-2}z^{n-2}+\cdots+\beta_{n,0},\nonumber
\end{align}
where $\gamma_{n,n}=\beta_{n,n}=1$. From the determinant expression of $P_n(z;t)$ given  in \eqref{lbop}, we have
$$\gamma_{n,j}=\frac{T_{n,j}}{\tau_n^{(0)}},\,\,\, j=0,\ldots,n-1,$$ where
\begin{align*}
T_{n,j}=(-1)^{n+j}
\det
\begin{pmatrix}
c_{p-kq}
\end{pmatrix}_{p=0,1,\ldots,n,\,p\neq j}^{q=0,1,\ldots,n-1}.
\end{align*}
By the aid of \eqref{5-4-1}, it is easy to get 
\begin{align}
\frac{d}{dt}T_{n,j}=-n \alpha T_{n,j}, \quad \frac{d}{dt}\tau_n^{(l)}=-n\alpha\tau_n^{(l)},
\end{align}
from which $\frac{d}{dt}\gamma_{n, j}=0$ immediately follows. In a similar way, we also have $\frac{d}{dt}\beta_{n,j}=0$.

By imposing the inner products with $Q_{n+j}(z;t),\ j=0,1,\ldots,k-1$ on the both sides of recurrence relation ~\eqref{5-1-3}, we have
\begin{align}
b_{n,j}h_{n+j}=&\int^{+\infty}_0z^k(P_n(z;t)+a_nP_{n-1}(z;t))Q_{n+j}(z^{-k};t)w(z;t)dz,\nonumber\\
=&\frac{1}{\tau_n^{(0)}}\sum_{l=0}^{n+j}\beta_{n+j,l}M_{n+1,l}+\frac{a_n}{\tau_{n-1}^{(0)}}\sum_{l=0}^{n+j}\beta_{n+j,l}M_{n,l},\nonumber
\end{align}
where
\begin{align*}
M_{n+1,l}=
\det
 \begin{pmatrix}
c_{p-kq}\quad c_{p-k(l-1)}
\end{pmatrix}_{q=0,1,\ldots,n-1}^{p=0,1,\ldots,n}.
\end{align*}
With the help of ~\eqref{5-4-1}, we easily get
$$\frac{d}{dt}M_{n+1,l}=-(n+1) \alpha M_{n+1,l}, \quad \frac{d}{dt}h_n=-\alpha h_n,$$
resulting in 
$$\frac{d}{dt}b_{n,j}=0.$$
In addition, since we also get from ~\eqref{5-4-1}
$$\frac{d}{dt}\tau_n^{(l)}=-n\alpha\tau_n^{(l)},$$
we obtain
\begin{align}
\frac{d}{dt}a_n=\frac{d}{dt}\left(\frac{\tau_{n-1}^{(0)}\tau_{n+1}^{(k)}}{\tau_{n}^{(k)}\tau_n^{(0)}}\right)=0.\nonumber
\end{align}
\end{proof}

The above theorem implies that the nonisospectral equation ~\eqref{5-2-6} can indeed allow the stationary reduction so that a class of d-P-type equations \eqref{5-3-1} are obtained. Since the coefficients of the OPs are independent of the time $t$, it also shows that  $\frac{\partial }{\partial t}P_n(z;t)=0$, which means that the Lax pair of the nonisospectral generalized mixed rToda lattice can indeed admit stationary reduction so that a class of d-P-type equations are obtained as well as their Lax pairs.

\section{Nonisospectral deformation of Cauchy bi-OPs and d-P}\label{Cauchy}
Motivated by the result that the explicit expression of the multipeakon solutions of the Degasperis--Procesi equation are given by Cauchy bi-moment determinants \cite{lundmark2003multi,lundmark2005degasperis}, Cauchy bi-OPs were first proposed in \cite{bertola2010cauchy} and have been extensively investigated in the literature, including the inspired Cauchy two-matrix model,  Toda lattice of CKP-type (C-Toda), as well as their generalizations etc. (see e.g. \cite{bertola2009cauchy,bertola2009cubic,bertola2010cauchy,bertola2014cauchy,bertola2013strong,chang2018degasperis,chang2021two,li2019cauchy,forrester2021fox,lago2019mixed,bertola2014universality,chang2025isospectral,chang2018degasperis,gonzalez2022strong}). It is noted that the C-Toda lattice can arise from isospectral deformation of the Cauchy bi-OPs \cite{chang2018degasperis}.
However, to the best of our knowledge, there is no existing literature that addresses Painlev\'{e}-type equations in relation to the Cauchy bi-OPs. In this section, we shall consider nonisospectral deformation of the Cauchy bi-OPs and derive a nonisospectral C-Toda lattice, from which we implement stationary reduction so that an integrable difference system is obtained as well as its Lax pair. Additionally, we construct a concrete weight function to guarantee the feasibility of the stationary reduction process.

\subsection{Cauchy bi-OPs}\label{introductionCauchy}
\begin{definition}
Define an inner product $\left \langle \cdot ,\cdot\right\rangle$ of the form
\begin{align}
\left \langle f(z) ,g(z)\right\rangle=\iint \frac{f(x)g(y)}{x+y}d\mu_1(x)d\mu_2(y),\label{6-1-1}
\end{align}
where $d\mu_1,d\mu_2$ are two positive measures defined on the integral interval (here the integral interval is omitted).
The Cauchy bi-OPs denote a pair of polynomials $\{P_n(z)\}_{n=0}^{\infty}$ and $\{Q_n(z)\}_{n=0}^{\infty}$ that are orthogonal with respect to the inner product ~\eqref{6-1-1}, that is, they satisfy the biorthogonality condition
\begin{align}
\left \langle P_n(z) ,Q_m(z)\right\rangle=h_n\delta_{n,m}.\label{6-1-2}
\end{align}
\end{definition}
Suppose that each $P_n(z)$ or $Q_n(z)$ is a monic polynomial of degree $n$ in $z$ and $\mu_1,\ \mu_2$ are two positive measures on $\mathbb{R}_+$ such that all the bimoments
\begin{align*}
I_{i,j}=\left \langle z^i ,z^j\right\rangle=\iint \frac{x^iy^j}{x+y}d\mu_1(x)d\mu_2(y),\quad i,j=0,1,\cdots
\end{align*}
and the single moments
\begin{align*}
\alpha_i=\int x^id\mu_1(x),\quad \beta_i=\int y^i d\mu_2(y), \quad i=0,1,\cdots
\end{align*}
exist. 
In this case, it can be shown that $\tau_n=\det(I_{i,j})_{i,j=0}^{n-1}\neq 0$, as a result of which, the monic Cauchy bi-OPs are uniquely determined by the biorthogonality condition \cite{bertola2010cauchy}. In fact, based on the biorthogonality condition ~\eqref{6-1-2}, the explicit determinant representations of the monic Cauchy bi-OPs are given by
\begin{align*}
	&P_n(z)=\frac{1}{\tau_n}
	\det\left(
\begin{array}{cccc}
I_{i,j}\\
z^j \\
\end{array}
\right)_{\substack{i=0,\ldots,n-1\\ j=0,\ldots,n\,\,\,\,\,\,\,\,}},\quad
&Q_n(z)=\frac{1}{\tau_n}
\det
\begin{pmatrix}
I_{i,j}&z^i \\
\end{pmatrix}_{\substack{i=0,\ldots,n,\,\,\,\,\,\\ j=0,\ldots,n-1}},
\end{align*}
and $h_n$ can be expressed as
\begin{align*}
h_n=\frac{\tau_{n+1}}{\tau_n}.
\end{align*}
It is noted that there holds the following relation for the moments
\begin{align}
I_{i+1,j}+I_{i,j+1}=\alpha_i\beta_j.\label{6-1-5}
\end{align}
Furthermore, the Cauchy bi-OPs satisfy the four-term recurrence relation  and admit the generalized Christoffel–Darboux formulation, and their zeros exhibit an interlacing pattern \cite{bertola2010cauchy}.

In the following, we focus on the Cauchy bi-OPs with symmetric measures, that is, $d\mu_1(x)=d\mu_2(x)=w(x)dx$. In such case, $P_n(z)=Q_n (z)$ and the biorthogonality condition ~\eqref{6-1-2} and moments can be rewritten as
\begin{align}
&\left \langle P_n(z) ,P_m(z)\right\rangle=\iint \frac{P_n(x)P_m(y)}{x+y}w(x)w(y)dxdy=h_n\delta_{n,m},\label{6-1-6}\\
&I_{i,j}=\iint \frac{x^iy^j}{x+y}w(x)w(y)dxdy,\quad i,j=0,1,\ldots,\label{6-1-7}\\
&\alpha_i=\beta_i=\int x^iw(x)dx,\quad i=0,1,\ldots,\label{6-1-8}
\end{align}
where $w(x)$ is the weight function.

By introducing
\begin{align}
a_n=-\frac{\int P_{n}w(x)dx}{\int P_{n-1}w(x)dx},\label{6-1-9}
\end{align}
and using the biorthogonality conditions ~\eqref{6-1-6}, it can be shown that the symmetric Cauchy bi-OPs satisfy the four-term recurrence relation
\begin{align}
z(P_n(z)+a_nP_{n-1}(z))
=P_{n+1}(z)+b_nP_n(z)+c_nP_{n-1}(z)+d_nP_{n-2}(z),\label{6-1-10}
\end{align}
with $P_{-1}(z)=0, \ P_{0}(z)=1$, where the recurrence coefficients $\{a_n,b_n,c_n,d_n\}$ can be expressed in terms of the variables $\{u_n,v_n\}$ as follows (see \cite{chang2018degasperis})
$$a_n=-\sqrt{\frac{v_nu_n}{v_{n-1}}},\ b_n=\frac{1}{2}v_n-\sqrt{\frac{v_nu_n}{v_{n-1}}},\ c_n=-u_n+\frac{1}{2}\sqrt{v_nu_nv_{n-1}},\ d_n=u_{n-1}\sqrt{\frac{v_nu_n}{v_{n-1}}},$$
with
\begin{align*}
u_n=\frac{\tau_{n+1}\tau_{n-1}}{(\tau_n)^2},\quad v_n=\frac{(\sigma_{n+1})^2}{\tau_{n+1}\tau_n},
\end{align*}
\begin{align*}
\sigma_n=
\det
\begin{pmatrix}
\beta_i&I_{i,j} \\
\end{pmatrix}_{\substack{i=0,\ldots,n,\,\,\,\,\,\\ j=0,\ldots,n-1}},
\quad
\tau_n=
\det
\begin{pmatrix}
I_{i,j} \\
\end{pmatrix}_{i,j=0}^{n-1}.
\end{align*}

\subsection{Nonisospectral C-Toda lattice }\label{nct}
In this subsection, we derive a nonisospectral C-Toda lattice by considering nonisospectral deformation of the monic Cauchy bi-OPs. To this end, we introduce the time variable $t$ into the measure and suppose that the spectral parameters are also related to $t$. We then deduce a time evolution equation satisfied by $\{P_n(x;t)\}_{n=0}^{\infty}$. By investigating the compatibility condition between this evolution equation and the recurrence relation, we obtain the nonisospectral C-Toda lattice.

We note that the C-Toda lattice was obtained by imposing isospectral deformation on Cauchy bi-OPs \cite{chang2018degasperis}. Very recently, in \cite{krichever2022constrained},  Krichever and Zabrodin introduced the so-called \textit{constrained Toda (C-Toda) hierarchy} as a certain subhierarchy of the 2D Toda lattice. In fact, we can prove that two flows produced by isospectral deformation  of Cauchy bi-OPs coincide with the first two members in this hierarchy; see Remark \ref{remmark_sec6}.

Suppose that the integral variables $x,y$ and the spectral parameters $z$ satisfy the time evolution
\begin{align}
 \frac{d}{dt}x=\alpha x,\quad \frac{d}{dt}y=\alpha y ,\quad \frac{d}{dt}z=\alpha z,\label{6-2-1}
\end{align}
and we seek for an appropriate weight function to ensure that the moments admit the time evolution
\begin{align}
\frac{d}{dt}I_{i,j}(t)=&\alpha(i+j-1) I_{i,j}(t)+\alpha_1(I_{i+1,j}(t)+I_{i,j+1}(t))\nonumber\\
&+\alpha_2 (I_{i+2,j}(t)+I_{i,j+2}(t)).\label{6-2-2}
\end{align}
Since taking the derivation of the moments \eqref{6-1-7} gives
\begin{align*}
\frac{d}{dt}I_{i,j}(t)
=&\langle\alpha(i+j-1)x^i,y^j\rangle\nonumber\\
&+\iint \frac{x^iy^j}{x+y}\left(\frac{d}{dt}w(x;t)w(y;t)+ w(x;t)\frac{d}{dt}w(y;t)+2\alpha w(x;t)w(y,t)\right)dxdy\nonumber\\
=&\alpha(i+j-1) I_{i,j}(t)\nonumber\\
&+\iint \frac{x^iy^j}{x+y}\left(\frac{d}{dt}w(x;t)w(y;t)+ w(x;t)\frac{d}{dt}w(y;t)+2\alpha w(x;t)w(y;t)\right)dxdy,
\end{align*}
we have
\begin{align*}
&\iint \frac{x^iy^j}{x+y}\left(\frac{d}{dt}w(x;t)w(y;t)+ w(x;t)\frac{d}{dt}w(y;t)+2\alpha w(x;t)w(y;t)\right)dxdy\nonumber\\
=&\alpha_1(I_{i+1,j}(t)+I_{i,j+1}(t))+\alpha_2 (I_{i+2,j}(t)+I_{i,j+2}(t)).
\end{align*}
Therefore it is reasonable to consider the weight function satisfying
\begin{align}
&\frac{d}{dt}(w(x;t))w(y;t)+ w(x;t)\frac{d}{dt}w(y;t)+2\alpha w(x;t)w(y;t)\nonumber\\
=&(\alpha_1(x+y)+\alpha_2(x^2+y^2))w(x;t) w(y;t).\label{6-2-3}
\end{align}
Such an explicit  weight function exist; see \eqref{6-4-2-3} for the corresponding moments.
\begin{lemma}\label{th9}
Under the assumptions  \eqref{6-2-1} and \eqref{6-2-2}, the monic Cauchy bi-OPs $\{P_n(z;t)\}_{n\in \mathbb{N}}$ satisfy the evolution relation
\begin{align}
&\frac{d}{dt}P_n(z;t)\nonumber\\
=&n\alpha P_n-(\alpha_1-\alpha_2a_{n+1})(zP_n-P_{n+1}+(a_n-b_n)P_n)\nonumber\\
&+\left(\alpha_1\frac{d_{n+1}}{a_{n+1}}-\alpha_2(d_{n+1}+(b_n-a_n)c_n-\frac{d_{n+1}}{a_{n+1}}\left(\frac{c_{n+1}}{a_{n+1}}+\frac{c_n}{a_n}-\frac{d_{n+2}}{a_{n+1}a_{n+2}}-\frac{d_{n+1}}{a_na_{n+1}}\right)\right)P_{n-1}\nonumber\\
&+\alpha_2(a_n(b_n-a_n)-c_n+a_nb_{n-1}-a_{n-1}a_n)(zP_{n-1}-P_n)\nonumber\\
&-\alpha_2\left(\frac{d_nd_{n+1}}{a_na_{n+1}}+d_n(b_n-a_n)+x(d_n-a_nc_{n-1}+a_{n-1}a_nb_{n-2})\right)P_{n-2}\nonumber\\
&+\alpha_2 z a_n(z a_{n-2}a_{n-1}+d_{n-1}-a_{n-1}c_{n-2})P_{n-3}-\alpha_2 z a_{n-1}a_nd_{n-2}P_{n-4}.\label{6-2-4}
\end{align}
\end{lemma}
\begin{proof}
Please refer to Appendix \ref{cauchy6.1} for the detailed proof.
\end{proof}

Using the recurrence relation \eqref{6-1-10} and the evolution relation \eqref{6-2-4}, we can expand the derivative formulae for $P_n,P_{n-1},P_{n-2}$ in terms of the linear combinations of themselves. In fact, we have
\begin{align*}
\frac{d}{dt}P_n(z;t)
=&(n\alpha-\alpha_1a_n+\alpha_2\left(a_n(a_{n+1}+a_n+a_{n-1}-b_n-b_{n-1}-z)+c_n)\right)P_n\nonumber\\
&+\left((\alpha_1-\alpha_2a_{n+1})(za_n-c_n)+\alpha_2(za_n-c_n)(z+b_n-a_n)+\frac{d_{n+1}}{a_{n+1}}\alpha_1\right.
\nonumber\\
&-\alpha_2d_{n+1}-\frac{d_{n+1}}{a_{n+1}}\left(\alpha_2\left(\frac{c_{n+1}}{a_{n+1}}+\frac{c_n}{a_n}-\frac{d_{n+2}}{a_{n+1}a_{n+2}}-\frac{d_{n+1}}{a_na_{n+1}}\right)\right)P_{n-1}\nonumber\\
&-d_n\left(\alpha_1-\alpha_2(a_{n+1}+a_n-\frac{d_{n+1}}{a_na_{n+1}}-b_n-z)\right)P_{n-2},\\
\frac{d}{dt}P_{n-1}(z;t)
=&\left(\alpha_1-\alpha_2\left(a_n+a_{n-1}-\frac{d_n}{a_na_{n-1}}-b_{n-1}-z\right)\right)P_n\nonumber\\
&+\bigg((n-1)\alpha-(\alpha_1-\alpha_2a_n)(a_{n-1}-b_{n-1}+z)+\alpha_2\bigg(a_{n-1}(a_{n-1}+a_{n-2}\nonumber\\
&\qquad-2b_{n-1}-b_{n-2})+c_{n-1}-\frac{d_n}{a_na_{n-1}}z-z^2+\frac{b_{n-1}d_n}{a_na_{n-1}}+b_{n-1}^2\bigg)\bigg)P_{n-1}\nonumber\\
&+\left(\alpha_1\frac{d_n}{a_n}-\alpha_2d_n\left(1+\frac{c_n}{a_n^2}-\frac{d_{n+1}}{a_{n}^2a_{n+1}}-\frac{d_n}{a_{n}^2a_{n-1}}+\frac{z}{a_n}\right)\right)P_{n-2},\\
\frac{d}{dt}P_{n-2}(z;t)
=&\bigg(-\frac{\alpha_1}{a_{n-1}}+\alpha_2\bigg(1+\frac{c_{n-1}}{a_{n-1}^2}-\frac{d_n}{a_{n-1}^2a_n}-\frac{d_{n-1}}{a_{n-1}^2a_{n-2}}+\frac{z}{a_{n-1}}\bigg)\bigg)P_n\nonumber\\
&+\bigg(\alpha_1-\alpha_2\bigg(a_{n-1}+a_{n-2}-b_{n-1}-b_{n-2}-\frac{d_{n-1}}{a_{n-2}a_{n-1}}\bigg)\nonumber\\
&\quad+(z-b_{n-1})\bigg(\frac{\alpha_1}{a_{n-1}}-\frac{\alpha_2}{a_{n-1}}\bigg(\frac{c_{n-1}}{a_{n-1}}-\frac{d_n}{a_na_{n-1}}-\frac{d_{n-1}}{a_{n-2}a_{n-1}}+z\bigg)\bigg)\bigg)P_{n-1}\nonumber\\
&+\bigg((n-2)\alpha-(\alpha_1-\alpha_2a_{n-1})(a_{n-2}-b_{n-2})+\alpha_2\bigg(a_{n-2}(a_{n-2}+a_{n-3}-2b_{n-2}\nonumber\\
&\quad-b_{n-3})+c_{n-2}-2z^2+\frac{b_{n-2}d_{n-1}}{a_{n-1}a_{n-2}}+b_{n-2}^2-z a_{n-1}\bigg(\frac{c_{n-1}}{a_{n-1}^2}-\frac{d_n}{a_na_{n-1}^2}\bigg)\bigg)\nonumber\\
&\quad-c_{n-1}\bigg(\frac{\alpha_1}{a_{n-1}}-\alpha_2(1+\frac{c_{n-1}}{a_{n-1}^2}-\frac{d_n}{a_{n-1}^2a_n}-\frac{d_{n-1}}{a_{n-1}^2a_{n-2}}+\frac{z}{a_{n-1}}\bigg)\bigg)\bigg)P_{n-2}.
\end{align*}
Consequently, we immediately derive the following overdetermined system in matrix form
\begin{align}
\varphi_{n+1}=A_n\varphi_n,\quad \frac{d\varphi_n}{dt}=B_n\varphi_n,\label{6-2-16}
\end{align}
where $\varphi_n(z;t)=(P_{n-2}(z;t),P_{n-1}(z;t),P_n(z;t))^T$, and
\begin{align*}
A_n=\left(
\begin{array}{ccc}
 0 & 1 & 0 \\
 0 & 0 & 1 \\
 -d_n & z  a_n-c_n & z -b_n \\
\end{array}
\right),\quad
B_n=\left(
\begin{array}{ccc}
 e_{11} & e_{12} & e_{13} \\
 e_{21} & e_{22} & e_{23} \\
 e_{31} & e_{32} & e_{33} \\
\end{array}
\right),
\end{align*}
with 
\begin{align}
e_{11}=&\alpha _2 \Big(\frac{b_{n-2} d_{n-1}}{a_{n-1} a_{n-2}}-a_{n-2} (2 b_{n-2}-a_{n-2}+b_{n-3})-z  a_{n-1} (\frac{c_{n-1}}{a_{n-1}^2}-\frac{d_n}{a_{n-1}^2 a_n})\nonumber\\
&\qquad+a_{n-3} a_{n-2}+b_{n-2}^2+c_{n-2}-2 z ^2\Big)-(a_{n-2}-b_{n-2}) (\alpha _1-\alpha _2 a_{n-1})\nonumber\\
&-c_{n-1} \Big(-\alpha _2 (\frac{c_{n-1}}{a_{n-1}^2}-\frac{d_{n-1}}{a_{n-1}^2 a_{n-2}}-\frac{d_n}{a_{n-1}^2 a_n}+1)-\frac{\alpha _2 z }{a_{n-1}}+\frac{\alpha _1}{a_{n-1}}\Big)+\alpha  (n-2),\nonumber\\
e_{12}=&( z -b_{n-1}) \Big(-\alpha _2 (\frac{c_{n-1}}{a_{n-1}^2}-\frac{d_{n-1}}{a_{n-1}^2 a_{n-2}}-\frac{d_n}{a_{n-1}^2 a_n}+\frac{ z }{a_{n-1}})+\frac{\alpha _1}{a_{n-1}}\Big)\nonumber\\
&+\alpha _2 (\frac{d_{n-1}}{a_{n-1} a_{n-2}}-a_{n-2}+b_{n-2}+b_{n-1})+(\alpha _1-\alpha _2 a_{n-1}),\nonumber\\
e_{13}=&\alpha _2 (\frac{c_{n-1}}{a_{n-1}^2}-\frac{d_{n-1}}{a_{n-1}^2 a_{n-2}}-\frac{d_n}{a_{n-1}^2 a_n}+\frac{z }{a_{n-1}}+1)-\frac{\alpha _1}{a_{n-1}},\nonumber\\
e_{21}=&-\alpha _2 (\frac{c_n d_n}{a_n^2}-\frac{d_n^2}{a_n^2 a_{n-1}}-\frac{d_{n+1} d_n}{a_n^2 a_{n+1}}+\frac{z   d_n}{a_n}+d_n)+\frac{\alpha _1 d_n}{a_n},\nonumber\\
e_{22}=&\alpha _2 \Big(\frac{b_{n-1} d_n}{a_n a_{n-1}}-a_{n-1}(2 b_{n-1}-a_{n-1})-a_{n-1} b_{n-2}-\frac{z  d_n}{a_{n-1} a_n}+a_{n-2} a_{n-1}\nonumber\\
&\qquad+b_{n-1}^2+c_{n-1}-z ^2\Big)-(\alpha _1-\alpha _2 a_n) (a_{n-1}-b_{n-1}+z )+\alpha  (n-1),\nonumber\\
e_{23}=&\alpha _2 (-a_{n-1}+b_{n-1}+z )+\frac{\alpha _2 d_n}{a_{n-1} a_n}+(\alpha _1-\alpha _2 a_n),\nonumber\\
e_{31}=&-d_n (\alpha _1-\alpha _2 a_{n+1})-\alpha _2 (d_n (-a_n+b_n+z )+\frac{d_{n+1} d_n}{a_n a_{n+1}}),\nonumber\\
e_{32}=&-\alpha _2 (c_n (b_n-a_n)+\frac{c_{n+1} d_{n+1}}{a_{n+1}^2}+\frac{c_n d_{n+1}}{a_n a_{n+1}}-\frac{d_{n+1} d_{n+2}}{a_{n+1}^2 a_{n+2}}-\frac{d_{n+1}^2}{a_{n+1}^2 a_n}-z ^2 a_n+d_{n+1})\nonumber\\
&+\alpha _2 z  (a_n \left(b_n-a_n\right)-c_n)+(\alpha _1-\alpha _2 a_{n+1}) (z  a_n-c_n)+\frac{\alpha _1 d_{n+1}}{a_{n+1}},\nonumber\\
e_{33}=&-\alpha _2 (a_n (b_n-a_n)+a_n b_{n-1}+z  a_n-a_{n-1} a_n-c_n)-a_n (\alpha _1-\alpha _2 a_{n+1})+\alpha  n.\nonumber
\end{align}
 The compatibility condition in~\eqref{6-2-16} yields a matrix representation of the integrable lattice 
\begin{align}
\frac{dA_n}{dt}=B_{n+1}A_n-A_nB_n,
\end{align}
from which, after a lengthy calculation, we obtain the explicit expression for the nonisospectral lattice. In summary, we have the following theorem.

\begin{theorem}
Under the assumptions  \eqref{6-2-1} and \eqref{6-2-2}, the variables $\{u_n(t),v_n(t)\}$ in the four-term recurrence relation \eqref{6-1-10}
for the monic Cauchy bi-OPs $\{P_n(x;t)\}_{n\in \mathbb{N}}$ satisfy the following generalized nonisospectral C-Toda lattice
\begin{subequations}\label{6-2-17}
\begin{align}
\frac{d}{dt}u_n=&2\alpha u_n+\alpha _1u_n(v_n-v_{n-1})\nonumber\\
&-\frac{1}{2}\alpha _2u_n\bigg((v_{n-1}^2-v_n^2)+4(u_{n+1}-u_{n-1})\nonumber\\
&\qquad\qquad-4\left(\sqrt{v_nv_{n+1}u_{n+1}}-\sqrt{v_{n-2}v_{n-1}u_{n-1}}\right)\bigg),\\
\frac{d}{dt}v_n=&\alpha v_n+2\alpha _1\left(\sqrt{v_nv_{n+1}u_{n+1}}-\sqrt{v_{n-1}v_nu_n}\right)\nonumber\\
&\qquad-\alpha _2\bigg(\left(v_{n-1}\sqrt{v_{n-1}v_nu_n}-v_n\sqrt{v_nv_{n+1}u_{n+1}}\right)\nonumber\\
&\qquad\qquad-2\left(\sqrt{v_nv_{n+2}u_{n+1}u_{n+2}}-\sqrt{v_{n-2}v_nu_{n-1}u_n}\right)\nonumber\\
&\qquad\qquad-\left(v_{n+1}\sqrt{v_nv_{n+1}u_{n+1}}-v_n\sqrt{v_{n-1}v_nu_n}\right)\bigg),
\end{align}
\end{subequations}
with the Lax pair \eqref{6-2-16}.
\end{theorem}
\begin{remark}\label{remmark_sec6}
When $\alpha\neq0$, the lattice \eqref{6-2-17} is a nonisosepctral generalization of the generalized C-Toda latice that incorporates the first and second flows of the C-Toda hierarchy. In the case of $\alpha=\alpha_2=0$, it reduces to the first flow of the isospectral C-Toda hierarchy, while it gives the second flow in the case of  $\alpha=\alpha_1=0$. The first flow of nonlinear form has been studied in \cite{chang2018degasperis}, while it seems that the explicit second flow is introduced for the first time. It is noted that these two flows coincide with the results obtained by Krichever and Zabrodin in \cite[eq. (2.40)]{krichever2022constrained}, where a combination of the first two flows was presented.
\end{remark}

\subsection{An integrable difference system related to nonisospectral C-Toda}\label{ct to p}
In this subsection, by applying stationary reduction to the Lax pair for the nonisospectral C-Toda lattice obtained in the previous subsection, we derive an integrable difference system together with its Lax pair. 

First, we obtain the following stationary reduction from the Lax pair of the nonisospectral equations
\begin{align}\label{6-3-1}
\varphi_{n+1}=F_n\varphi_n,\quad \frac{\partial\varphi_n}{\partial z}=G_n\varphi_n,
\end{align}
where
$$F_n=A_n,\quad G_n=B_n/\left(\frac{dz}{dt}\right)=\frac{1}{\alpha z}B_n.$$
Then, by considering the compatibility condition in \eqref{6-3-1}, we have the following theorem.
\begin{theorem}
Under the assumptions \eqref{6-2-1} and \eqref{6-2-2} as well as the stationary reduction,
the variables $\{u_n,v_n\}_{n\in \mathbb{N}}$  in the four-term recurrence relation \eqref{6-1-10}
for the monic Cauchy bi-OPs satisfy the following integrable difference system with the Lax pair \eqref{6-3-1} 
\begin{subequations}\label{pain-cbop}
\begin{align}
&\alpha_1v_n+2n\alpha+\beta_0, \nonumber\\ 
&+\frac{1}{2}\alpha_2\bigg(v_{n}^2-4u_{n+1}-4u_{n}+4\sqrt{v_nv_{n+1}u_{n+1}}+4\sqrt{v_nv_{n-1}u_{n}}\bigg)=0,\label{pain-cbop1}\\
&\alpha \sqrt{v_n}+\alpha_1(\sqrt{v_{n+1}u_{n+1}}-\sqrt{v_{n-1}u_n})\nonumber\\
&+\alpha _2((v_n+\frac{1}{2}v_{n+1})\sqrt{v_{n+1}u_{n+1}}-(v_n+\frac{1}{2}v_{n-1})\sqrt{v_{n-1}u_n})\nonumber\\
&-\sqrt{\frac{u_{n+1}}{v_{n+1}}}(\alpha_2(-2u_{n+2}-2u_{n+1}+2\sqrt{v_nv_{n+1}u_{n+1}})+2(n+1)\alpha+\beta_0)\nonumber\\
&+\sqrt{\frac{u_{n}}{v_{n-1}}}(\alpha_2(-2u_{n}-2u_{n-1}+2\sqrt{v_nv_{n-1}u_{n}})+2(n-1)\alpha+\beta_0)=0,
\label{pain-cbop2}
\end{align}
\end{subequations}
where $\beta_0=-\alpha_1v_0-\frac{1}{2}\alpha_2\left(v_{0}^2-4u_{1}-4u_{0}+4\sqrt{v_0v_{1}u_{1}}\right)$.
\end{theorem}

\begin{proof}
The compatibility condition in \eqref{6-3-1} yields
$$F_{n,x}+F_nG_n-G_{n+1}F_n=0,$$ 
from which we derive the following difference system
\begin{subequations}\label{6-3-2}
\begin{align}
&2\alpha u_n+\alpha _1u_n(v_n-v_{n-1})\nonumber\\
&-\frac{1}{2}\alpha _2u_n\bigg((v_{n-1}^2-v_n^2)+4(u_{n+1}-u_{n-1})\nonumber\\
&\qquad\qquad-4\left(\sqrt{v_nv_{n+1}u_{n+1}}-\sqrt{v_{n-2}v_{n-1}u_{n-1}}\right)\bigg)=0,\label{nc1}\\
&\alpha v_n+2\alpha _1\left(\sqrt{v_nv_{n+1}u_{n+1}}-\sqrt{v_{n-1}v_nu_n}\right)\nonumber\\
&\qquad-\alpha _2\bigg(\left(v_{n-1}\sqrt{v_{n-1}v_nu_n}-v_n\sqrt{v_nv_{n+1}u_{n+1}}\right)\nonumber\\
&\qquad\qquad-2\left(\sqrt{v_nv_{n+2}u_{n+1}u_{n+2}}-\sqrt{v_{n-2}v_nu_{n-1}u_n}\right)\nonumber\\
&\qquad\qquad-\left(v_{n+1}\sqrt{v_nv_{n+1}u_{n+1}}-v_n\sqrt{v_{n-1}v_nu_n}\right)\bigg)=0.\label{nc2}
\end{align}
\end{subequations}
Eliminating $u_n$ in \eqref{nc1} and then summing for $n$, we readily obtain \eqref{pain-cbop1}. Furthermore, it is straightforward to obtain from \eqref{pain-cbop1} that
\begin{align*}
&2\alpha_2\sqrt{v_{n+1}v_{n+2}u_{n+2}}\\
=&-\alpha_1v_{n+1}-\frac{1}{2}\alpha_2(v_{n+1}^2-4(u_{n+2}+u_{n+1})+4\sqrt{v_nv_{n+1}u_{n+1}})-2(n+1)\alpha-\beta_0,\\
&2\alpha_2\sqrt{v_{n-2}v_{n-1}u_{n-1}}\\
=&-\alpha_1v_{n-1}-\frac{1}{2}\alpha_2(v_{n-1}^2-4(u_{n}+u_{n-1})+4\sqrt{v_{n-1}v_{n}u_{n}})-2(n-1)\alpha-\beta_0,
\end{align*}
using which, we can substitute the terms $\sqrt{v_nv_{n+2}u_{n+1}u_{n+2}}$ and $\sqrt{v_{n-2}v_nu_{n-1}u_n}$ in \eqref{nc2}, and thus obtain \eqref{pain-cbop2}.
\end{proof}

\begin{remark}
Obviously, \eqref{6-3-2} can be recognized as the stationary form of the nonisospectral C-Toda lattice equation \eqref{6-2-17}. Also we note that \eqref{pain-cbop} is integrable with a 3$\times3$ Lax pair and associated with the Toda hierarchy of CKP type.
It might be related to a d-P-type equation
of high order and a further simplification will be considered in the future.

\end{remark}

\subsection{Realization of stationary reduction}\label{ctasr}
Having derived the integrable difference system through the stationary reduction of the nonisospectral C-Toda lattice, we proceed to construct specific moments to guarantee the validity of the aforementioned stationary reduction.

First, based on the time evolution \eqref{6-2-1} resulting in
\begin{align*}
x=x(0)e^{\alpha t},\quad y=y(0)e^{\alpha t}, \quad z=z(0)e^{\alpha t},
\end{align*}
we observe that the moments can be rewritten as follows
\begin{align*}
I_{i,j}(t)=&\iint_{\mathbb{R}_{+}^{2}}\frac{x^i(t)y^j(t)}{x(t)+y(t)}w(x;t)w(y;t)dx(t)dy(t)\notag\\
=&\iint_{\mathbb{R}_{+}^{2}}\frac{x^i(0)y^j(0)e^{(i+j)\alpha t}}{(x(0)+y(0))e^{\alpha t}}f(x(0);t)f(y(0);t)dx(0)dy(0),
\end{align*}
where $f(\cdot;t)$, that is related to the weight function, needs to be determined, and the integral interval $(0,+\infty)$ is considered.

Then we differentiate the above relation to derive
\begin{align*}
\frac{d}{dt}I_{i,j}(t)=&(i+j-1)\alpha I_{i,j}(t)\notag\\
&+\iint_{\mathbb{R}_{+}^{2}}\frac{x^i(0)y^j(0)e^{(i+j)\alpha t}}{(x(0)+y(0))e^{\alpha t}}\frac{d(f(x(0);t)f(y(0);t))}{dt}dx(0)dy(0).
\end{align*}
Since the moments satisfy the time evolution~\eqref{6-2-2}, we require $f(x(0);t),\ f(y(0);t)$ to obey
\begin{align*}
\frac{d}{dt}f(x(0);t)&=(\alpha _1x(0)e^{\alpha t}+\alpha _2x^2(0)e^{2\alpha t})f(x(0);t),\\
\frac{d}{dt}f(y(0);t)&=(\alpha _1y(0)e^{\alpha t}+\alpha _2y^2(0)e^{2\alpha t})f(y(0);t).
\end{align*}
This means that it is reasonable to set
\begin{align*}
f(x(0);t)=e^{\frac{\alpha _1}{\alpha}x(0)e^{\alpha t}+\frac{\alpha _2}{2\alpha}x^2(0)e^{2\alpha t}},\qquad f(y(0);t)=e^{\frac{\alpha _1}{\alpha}y(0)e^{\alpha t}+\frac{\alpha _2}{2\alpha}y^2(0)e^{2\alpha t}},
\end{align*}
from which the following exact expressions of the moments are obtained
\begin{align*}
&I_{i,j}(t)
=\iint_{\mathbb{R}_{+}^{2}}\frac{x^i(0)y^j(0)e^{(i+j)\alpha t}}{(x(0)+y(0))e^{\alpha t}}e^{\frac{\alpha _1}{\alpha}(x(0)+y(0)) e^{\alpha t}+\frac{\alpha _2}{2\alpha}(x^2(0)+y^2(0))e^{2\alpha t}}dx(0)dy(0),\\ 
&\beta_i(t)
=\int_{\mathbb{R}_{+}}x^i(0)e^{i\alpha t}e^{\frac{\alpha _1}{\alpha}x(0) e^{\alpha t}+\frac{\alpha _2}{2\alpha}x^2(0)e^{2\alpha t}}dx(0).
\end{align*}
Furthermore, it is easy to show that single moments satisfy the time evolution
\begin{align}\label{evo_beta}
\frac{d}{dt}\beta_i(t)=i\alpha \beta_i(t)+\alpha_1\beta_{i+1}(t)+\alpha_2 \beta_{i+2}(t).
\end{align}
\begin{lemma}\label{th10}
Under the assumption  \eqref{6-2-1} together with $\frac{\alpha_2}{\alpha}<0$, define the moments as
\begin{subequations}\label{6-4-2-3}
\begin{align}
&I_{i,j}(t)
=\iint_{\mathbb{R}_{+}^{2}}\frac{x^i(0)y^j(0)e^{(i+j)\alpha t}}{(x(0)+y(0))e^{\alpha t}}e^{\frac{\alpha _1}{\alpha}(x(0)+y(0)) e^{\alpha t}+\frac{\alpha _2}{2\alpha}(x^2(0)+y^2(0))e^{2\alpha t}}dx(0)dy(0),\label{6-4-2}\\ 
&\beta_i(t)
=\int_{\mathbb{R}_{+}}x^i(0)e^{i\alpha t}e^{\frac{\alpha _1}{\alpha}x(0) e^{\alpha t}+\frac{\alpha _2}{2\alpha}x^2(0)e^{2\alpha t}}dx(0).\label{6-4-3}
\end{align}
\end{subequations}
Then the moments simultaneously satisfy the time evolution \eqref{6-2-2} and 
\begin{align}\label{6-4-4-4}
&\frac{d}{dt}I_{i,j}(t)=-2\alpha I_{i,j}(t),\qquad \frac{d}{dt}\beta_i(t)=-\alpha \beta_i(t).
\end{align}
\end{lemma}

\begin{proof} 
It is sufficient to confirm the relations in \eqref{6-4-4-4}.
In the case of $\frac{\alpha _2}{\alpha}<0$, we observe at the boundary
$$\lim_{x(0)\to0^+}x(0)f(x(0);t)=\lim_{x(0)\to +\infty}x(0)f(x(0);t)=0.$$
As for the exact expression ~\eqref{6-4-2} for bimoments, by using integration by parts with respect to $x(0)$, we have 
\begin{align*}
I_{i,j}(t)
=&\iint_{\mathbb{R}_{+}^{2}}\frac{x^i(0)y^j(0)e^{(i+j)\alpha t}}{(x(0)+y(0))e^{\alpha t}}f(x(0);t)f(y(0);t)dx(0)dy(0)\notag\\
=&-\iint_{\mathbb{R}_{+}^{2}}\frac{i x^i(0)y^j(0)e^{(i+j)\alpha t}}{(x(0)+y(0))e^{\alpha t}}f(x(0);t)f(y(0);t)dx(0)dy(0)\notag\\
&+\iint_{\mathbb{R}_{+}^{2}}\frac{ x^{i+1}(0)y^j(0)e^{(i+j+1)\alpha t}}{(x(0)+y(0))^2e^{2\alpha t}}f(x(0);t)f(y(0);t)dx(0)dy(0)\notag\\
&-\iint_{\mathbb{R}_{+}^{2}}\frac{x^{i+1}(0)y^j(0)e^{(i+j)\alpha t}}{(x(0)+y(0))e^{\alpha t}}\frac{\alpha _1}{\alpha} e^{\alpha t}f(x(0);t)f(y(0);t)dx(0)dy(0)\notag\\
&-\iint_{\mathbb{R}_{+}^{2}}\frac{x^{i+1}(0)y^j(0)e^{(i+j)\alpha t}}{(x(0)+y(0))e^{\alpha t}}\frac{\alpha _2}{\alpha}x(0)e^{2\alpha t}f(x(0);t)f(y(0);t)dx(0)dy(0)).
\end{align*}
Similarly,  by integrating the bimoments by parts with respect to $y(0)$, we have
\begin{align*}
I_{i,j}(t)
=&\iint_{\mathbb{R}_{+}^{2}}\frac{x^i(0)y^j(0)e^{(i+j)\alpha t}}{(x(0)+y(0))e^{\alpha t}}f(x(0);t)f(y(0);t)dx(0)dy(0)\notag\\
=&-\iint_{\mathbb{R}_{+}^{2}}\frac{j x^i(0)y^j(0)e^{(i+j)\alpha t}}{(x(0)+y(0))e^{\alpha t}}f(x(0);t)f(y(0);t)dx(0)dy(0)\notag\\
&+\iint_{\mathbb{R}_{+}^{2}}\frac{ x^i(0)y^{j+1}(0)e^{(i+j+1)\alpha t}}{(x(0)+y(0))^2e^{2\alpha t}}f(x(0);t)f(y(0);t)dx(0)dy(0)\notag\\
&-\iint_{\mathbb{R}_{+}^{2}}\frac{x^i(0)y^{j+1}(0)e^{(i+j)\alpha t}}{(x(0)+y(0))e^{\alpha t}}\frac{\alpha _1}{\alpha} e^{\alpha t}f(x(0);t)f(y(0);t)dx(0)dy(0)\notag\\
&-\iint_{\mathbb{R}_{+}^{2}}\frac{x^i(0)y^{j+1}(0)e^{(i+j)\alpha t}}{(x(0)+y(0))e^{\alpha t}}\frac{\alpha _2}{\alpha}y(0)e^{2\alpha t}f(x(0);t)f(y(0);t)dx(0)dy(0).
\end{align*}
Summing the above two formulae gives
\begin{align*}
(i+j+1)\alpha I_{i,j}(t)+\alpha_1(I_{i+1,j}(t)+I_{i,j+1}(t))+\alpha_2 (I_{i+2,j}(t)+I_{i,j+2}(t))=0,
\end{align*}
combining which with ~\eqref{6-2-2} immediately gives
\begin{align}
\frac{d}{dt}I_{i,j}(t)=-2\alpha I_{i,j}(t).\nonumber
\end{align}
Similarly, $\frac{d}{dt}\beta_i(t)=-\alpha \beta_i(t)$ can also be deduced.
\end{proof}

Based on the above lemma, it is not hard to conclude the following result.

\begin{theorem}
Given the moments in \eqref{6-4-2-3}, we have
\begin{align*}
\frac{d}{dt}u_n(t)=\frac{d}{dt}v_n(t)=0,\qquad \frac{d}{dt}\gamma_{n,j}(t)=0,
\end{align*}
where $\{u_n\}$ and $\{v_{n}\}$ are the variables  in the four-term recurrence relation \eqref{6-1-10} for the monic Cauchy bi-OPs and  $\{\gamma_{n,j}\}$ are the expansion coefficients of the monic Cauchy bi-OPs with $P_n(z)=\sum_{j=0}^n\gamma_{n,j}z^j$.
\end{theorem}
\begin{proof}
By use of ~\eqref{6-4-4-4},  it is easy to see that the determinants $ \tau_n$ and $ \sigma_n$ evolve according to
\begin{align*}
\frac{d}{dt}\tau_n=-2n\alpha \tau_n,\quad \frac{d}{dt}\sigma_n=-(2n-1)\alpha \sigma_n,
\end{align*}
from which, we conclude that $\{u_n,v_n\}_{n\in \mathbb{N}}$ are independent of time $t$ by recalling the formulae \begin{align*}
u_n=\frac{\tau_{n+1}\tau_{n-1}}{(\tau_n)^2},\quad v_n=\frac{(\sigma_{n+1})^2}{\tau_{n+1}\tau_n}.
\end{align*}

As for the coefficients $\gamma_{n,j}$ in the expansion of $P_n(z)$, it follows from the determinant expression of $\{P_n(z)\}_{n=0}^{\infty}$ that
$$\gamma_{n,j}=\frac{T_{n,j}}{\tau_n},$$
where 
\begin{align*}
T_{n,j}=(-1)^{n+j}
\det\begin{pmatrix}
I_{p,q}
\end{pmatrix}_{\substack{p=0,\ldots,n-1,\,\,\,\\\ q=0,\ldots,n,q,\neq j}}.
\end{align*}
Due to the evolution relation \eqref{6-4-4-4} satisfied by moments, we have
$$\frac{d}{dt}\tau_n=-2n\alpha \tau_n,\quad \frac{d}{dt}T_{n,j}=-2n\alpha T_{n,j}.$$
which immediately leads to $\frac{d}{dt}\gamma_{n,j}=0$.  

\end{proof}

In summary, we have demonstrated that the variables $\{u_n,v_n\}_{n\in \mathbb{N}}$ and the expansion coefficients $\gamma_j$ of the monic Cauchy bi-OPs are independent of time $t$. Therefore, under the definition of the moments~\eqref{6-4-2-3}, the nonisospectral C-Toda lattice can indeed admit stationary reduction, so that a difference system associated with the Toda hierarchy of CKP-type and Cauchy bi-OPs arises. Furthermore, the Lax pair of the nonisospectral C-Toda lattice can also indeed allow stationary reduction so that the Lax pair of the difference system is obtained.

\section{Nonisospectral deformation of partial-skew OPs and d-P}\label{Partial-skew}
Recently, the concept of \textit{partial-skew OPs} was introduced in \cite{chang2018partial}, with the motivation by the study of a random matrix model called the Bures random ensemble \cite{forrester2016relating} as well as the hints from the formulation of the Novikov peakon solution in terms of Pfaffians \cite{chang2018application}. It is shown that isospectral deformation of the partial-skew OPs are closely related to integrable lattices \cite{chang2022hermite,chang2025isospectral,chang2018partial} and they also solve certain mixed Hermite–Padé approximation problems \cite{chang2022hermite}. To the best of our knowledge, there is currently no existing literature that establishes a connection between the partial-skew OPs and  and Painlev\'{e}-type equations, which is the main objective of this section.

To achieve this objective, we first perform a nonisospectral deformation of the partial-skew OPs without specifying the weight function. It is noted that we employ the \textit{Pfaffian} as an algebraic tool (for an introduction on the Pfaffians, please refer to Appendix \ref{chap:chaapp}), and utilize the derivative rules and identities on Pfaffians to derive nonisospectral deformation of the partial-skew OPs and the nonisospectral B-Toda lattice. Then, by applying the stationary reduction to the Lax pair of the nonisospectral equation, we obtain an integrable difference system along with its Lax pair. And it is shown that the integrable difference system features a solution expressed in terms of Pfaffians.
Finally, we construct a specific weight function to demonstrate the viability of the aforementioned stationary reduction process.

\subsection{Partial-skew OPs}\label{introductionpartial}
In this subsection, let's briefly review the relevant basic knowledge about partial-skew OPs. For more details please refer to \cite{chang2018partial}.

\begin{definition}
    
Consider a skew symmetric inner product $\left \langle \cdot ,\cdot\right\rangle$ on the space of real coefficient polynomials, where $\left \langle \cdot ,\cdot\right\rangle$ is a bilinear 2-form $\mathbb{R}(z) \times \mathbb{R}(z) \rightarrow \mathbb{R}(z)$ satisfying the skew symmetric relation 
\begin{align*}
\left \langle f(z) ,g(z)\right\rangle=-\left \langle g(z) ,f(z)\right\rangle.
\end{align*}
Define the bimoment sequence $\{\mu_{i,j}\}_{i,j=0}^{\infty}$ as
\begin{align}
\mu_{i,j}=\left \langle z^i ,z^j\right\rangle=-\left \langle z^j ,z^i\right\rangle=-\mu_{j,i}.\label{7-1-2}
\end{align}
A family of monic polynomials $\{P_{n}(z)\}_{n=0}^{\infty}$ are called partial-skew OPs if they satisfy the orthogonality relation
\begin{align*}
\left \langle P_{2n}(z) ,z^m\right\rangle=\frac{\tau_{2n+2}}{\tau_{2n}}\delta_{2n+1,m},\quad 0\leq m \leq 2n+1,\nonumber\\
\left \langle P_{2n+1}(z) ,z^m\right\rangle=-\frac{\tau_{2n+2}}{\tau_{2n+1}}\beta_{m},\quad 0\leq m \leq 2n+1,\nonumber
\end{align*}
with
\begin{align}
\tau_{2n}=\Pf(0,1,\ldots,2n-1)\neq 0,\quad \tau_{2n+1}=\Pf(d_0,0,1,\ldots,2n)\neq 0,\nonumber
\end{align}
where $\beta_i$ are some constants and the Pfaffian entries are defined as 
\begin{subequations}\label{7-1-4}
\begin{align}
\Pf(i,j)=\mu_{i,j},\quad \Pf(i,z)=z^i,\\
\Pf(d_0,i)=\beta_i,\quad \Pf(d_0,z)=0.
\end{align}
\end{subequations}
\end{definition}

According to the orthogonality, it can be shown that the monic partial-skew OPs admit the following determinant expressions
    \begin{align*}
	&P_{2n}(z)=\frac{1}{\det(\mu_{i,j})_{0\leq i,j\leq 2n-1}}
	\det\begin{pmatrix}
\mu_{i,j}\\
z^j
\end{pmatrix}_{i=0,\ldots,2n-1\,\,\,}^{j=0,\ldots,2n},\\
&P_{2n+1}(z)=\frac{1}{\det(C_{2n+2,2n+2})}
\det\begin{pmatrix}
\mu_{i,j}&-\beta_{i}\\
z^j&0
\end{pmatrix}_{i=0,\ldots,2n+1\,\,\,}^{j=0,\ldots,2n+1},
\end{align*}
where
\begin{align*}
C_{k,l}=\left(\begin{array}{c|c}
 & \beta_0 \\
A_{k,l-1} &\vdots \\
 & \beta_{k-1} \\
\end{array}\right)
\end{align*}
$A_{k,l}=(\mu_{i,j})_{0\leq i\leq k-1,0\leq j \leq l-1}$.
It follows from \eqref{7-1-2} that the determinant $\det(\mu_{i,j})_{0\leq i,j\leq 2n-1}$ is a skew symmetric determinant of $2n$ order, implying that it can be written in terms of Pfaffians. 
In fact, in conjunction with \eqref{7-1-7} and~\eqref{7-1-8}, it is not hard to see that the Pfaffian representations of partial-skew OPs can be given by
\begin{subequations}\label{psop}
\begin{align}
&P_{2n}(z)=\frac{1}{\tau_{2n}}\Pf(0,1,\ldots,2n-1,2n,z),\label{7-1-9}\\
&P_{2n+1}(z)=\frac{1}{\tau_{2n+1}}\Pf(d_0,0,1,\ldots,2n,2n+1,z).\label{7-1-10}
\end{align}
\end{subequations}

In this work, we address the skew symmetric inner product of the following specific form
\begin{align}\label{skew_kernel}
\left \langle f(z) ,g(z)\right\rangle=\iint f(x)g(y)\frac{y-x}{x+y}w(x)w(y)dxdy,
\end{align}
in which $\frac{y-x}{x+y}$ is the skew symmetric integral kernel to guarantee the skew symmetric property and $w(x)$ is the weight function ensuring the bimoments 
\begin{align*}
\mu_{i,j}=\iint x^iy^j\frac{y-x}{x+y}w(x)w(y)dxdy=\Pf(i,j),
\end{align*}
and the single moments
\begin{align*}
&\beta_i=\int x^iw(x)dx=\Pf(d_0,i)
\end{align*}
exist.  In this case, it is not difficult to see that the Pfaffian entries satisfy
\begin{subequations}\label{7-2-7}
\begin{align}
&\Pf(i+1,j)+\Pf(i,j+1)=-\Pf(d_0,d_1,i,j),\\
&\Pf(d_0,i+1)=\Pf(d_1,i),
\end{align}
\end{subequations}
where $$\Pf(d_1,i)=\beta_{i+1},\qquad \Pf(d_0,d_1)=0.$$
It is worth noting that the specific partial-skew OPs satisfy a four-term recurrence relation
\begin{align}
z(P_n(z)-u_nP_{n-1}(z))=&P_{n+1}(z)+(b_{n+1}-b_n-u_n)P_{n}(z)\nonumber\\
&-u_n(b_{n}-b_{n+1}+u_{n+1})P_{n-1}(z)+u_n^2u_{n-1}P_{n-2}(z),\label{7-1-14}
\end{align}
where
\begin{align*}
u_n=\frac{\tau_{n+1}\tau_{n-1}}{\tau_{n}^2},\qquad b_{n}=\frac{\sigma_n}{\tau_n},
\end{align*}
with $\tau_n$ denoting the Pfaffians with the Pfaffian entries \eqref{7-1-4}
and $\sigma_n$ denoting the following Pfaffians
$$\sigma_{2m}=\Pf(0,\ldots,2m-2,2m),\quad \sigma_{2m+1}=\Pf(d_0,0,\ldots,2m-1,2m+1).
$$

\subsection{Nonisospectral B-Toda lattice }\label{nbt}
In this subsection we perform nonisospectral deformation on partial-skew OPs. To this end, we suppose that the integral variables $x,y$ and the spectral parameter $z$ satisfy the following time evolution
\begin{align}
\frac{d}{dt}x=\alpha x,\quad \frac{d}{dt}y=\alpha y, \quad \frac{d}{dt}z=\alpha z\label{7-2-1}
\end{align}
and we require an appropriate weight function to ensure that the moments undergo the time evolution
\begin{subequations}\label{7-2-2-3}
\begin{align}
\frac{d}{dt}\mu_{i,j}(t)=&\alpha(i+j) \mu_{i,j}(t)+\alpha_1(\mu_{i+1,j}(t)+\mu_{i,j+1}(t))+\alpha_2 (\mu_{i+2,j}(t)+\mu_{i,j+2}(t)),\label{7-2-2}\\
\frac{d}{dt}\beta_i(t)=&i\alpha \beta_i(t)+\alpha_1 \beta_{i+1}(t)+\alpha_2 \beta_{i+2}(t).\label{7-2-3}
\end{align}
\end{subequations}
According to the definitions \eqref{7-1-4} of Pfaffian entries, the above relations can be equivalently written as
\begin{subequations}\label{7-2-4}
\begin{align*}
\frac{d}{dt}\Pf(i,j)=&\alpha(i+j) \Pf(i,j)+\alpha_1(\Pf(i+1,j)+\Pf(i,j+1))\nonumber\\&
+\alpha_2 (\Pf(i+2,j)+\Pf(i,j+2)),\\
\frac{d}{dt}\Pf(d_0,i)=&i\alpha \Pf(d_0,i)+\alpha_1 \Pf(d_0,i+1)+\alpha_2 \Pf(d_0,i+2).
\end{align*}
\end{subequations}

Under the assumptions \eqref{7-2-1} and \eqref{7-2-2-3},  we can obtain some derivative formulae for the Pfaffians and the specific partial-skew OPs by employing ~\eqref{pf1}-\eqref{pf6} in Appendix \ref{DerPfaffian}. For facilitate the presentation, we first introduce the intermediate variable $t_1$ and list some derivative formulae with respect to $t_1$. 

Introduce the evolution on $t_1$ according to
\begin{align*}
\frac{\partial}{\partial t_1}\Pf(i,j)=&(\Pf(i+1,j)+\Pf(i,j+1)),\\
\frac{\partial}{\partial t_1}\Pf(d_0,i)=& \Pf(d_0,i+1),
\end{align*}
which can actually be regarded as a special evolution in the case of $\alpha=\alpha_2=0,\alpha_1=1$.
Then it follows from ~\eqref{pf1} and ~\eqref{pf2} that
\begin{subequations}\label{7-2-15-16}

\begin{align}
&\frac{\partial}{\partial t_1}\Pf(i_1,\ldots, i_{2N})=\sum_{k=i_1}^{i_{2N}}\Pf(i_1,i_2,\ldots,i_{k}+1, \ldots, i_{2N}),\label{7-2-15}\\
&\frac{\partial}{\partial t_1}\Pf(a_0,i_1,\ldots, i_{2N-1})=\sum_{k=i_1}^{i_{2N-1}}\Pf(a_0,i_1,i_2,\ldots,i_{k}+1, \ldots, i_{2N-1}).\label{7-2-16}
\end{align}
\end{subequations}

Based on the results in \cite[Lemma 3.15, Lemma 3.20]{chang2018partial}, we immediately obtain the derivative formulae for $\tau_{2n},\ \tau_{2n+1},\ \tau_{2n}P_{2n}(z;t),\ \tau_{2n+1}P_{2n+1}(z;t)$ with respect to $t_1$. 

\begin{lemma}\label{lem:t1_wron}
There hold the following differential relations
    \begin{subequations}\label{7-2-t1}
\begin{align}
&\frac{\partial}{\partial t_1}\tau_{2n}=\Pf(0,\ldots,2n-2,2n),\label{7-2-17}\\
&\frac{\partial}{\partial t_1}\tau_{2n+1}=\Pf(d_0,0,\ldots,2n-1,2n+1),\label{7-2-18}\\
&\frac{\partial}{\partial t_1}(\tau_{2n}P_{2n}(z;t))
=\Pf(0,\ldots,2n-1,2n+1,z)-z \Pf(0,1,\ldots,2n,z),\label{7-2-21}\\
&\frac{\partial}{\partial t_1}(\tau_{2n+1}P_{2n+1}(z;t))
=\Pf(d_0,0,\ldots,2n,2n+2,z)-z \Pf(d_0,0,\ldots,2n+1,z).\label{7-2-22}
\end{align}
\end{subequations}
\end{lemma}
By successively applying~\eqref{pf1} and ~\eqref{pf2}, we can also derive the following corollary on the second derivative with respect to $t_1$. 
\begin{coro}\label{co:7-2}
There hold
\begin{subequations}
\begin{align}
&\frac{\partial^2}{\partial t_1^2}\tau_{2n}
=\Pf(0,\ldots,2n-3,2n-1,2n)+\Pf(0,\ldots,2n-2,2n+1),\label{7-2-19}\\
&\frac{\partial^2}{\partial t_1^2}\tau_{2n+1}
=\Pf(d_0,0,\ldots,2n-2,2n,2n+1)+\Pf(d_0,0,\ldots,2n-1,2n+2),\label{7-2-20}\\
&\frac{\partial^2}{\partial t_1^2}(\tau_{2n}P_{2n}(z;t))\nonumber\\
=&-2z \Pf(0,\ldots,2n-1,2n+1,z)+\Pf(0,\ldots,2n-2,2n,2n+1,z)\nonumber\\
&+\Pf(0,\ldots,2n-1,2n+2,z)+z^2 \Pf(0,\ldots,2n,z),\label{7-2-23}\\
&\frac{\partial^2}{\partial t_1^2}(\tau_{2n+1}P_{2n+1}(z;t))\nonumber\\
=&-2z \Pf(d_0,0,\ldots,2n,2n+2,z)+\Pf(d_0,0,\ldots,2n,2n+3,z)\nonumber\\
&+\Pf(d_0,0,\ldots,2n-1,2n+1,2n+2,z)+z^2 \Pf(d_0,0,\ldots,2n+1,z).\label{7-2-24}
\end{align}
\end{subequations}
\end{coro}
\begin{proof}
See Appendix \ref{pco1}.
\end{proof}

Furthermore, observe that the derivative for the Pfaffian entries with respect to $t_1$ can be equivalently written as
\begin{align*}
\frac{\partial}{\partial t_1}\Pf(i,j)=&-\Pf(d_0,d_1,i,j),\\
\frac{\partial}{\partial t_1}\Pf(d_0,i)=& \Pf(d_1,i),
\end{align*}
by recalling \eqref{7-2-7}. By employing ~\eqref{pf3} and ~\eqref{pf4} with $\alpha=\alpha_2=0,\alpha_1=1$, similar to the results in \cite[Lemma 3.15, Corollary 3.16]{chang2018partial}, we also have the following evolution relations.
\begin{lemma}\label{lem:t1}
There hold the following differential relations
    \begin{subequations}\label{7-2-t1-2}
\begin{align}
&\frac{\partial}{\partial t_1}\tau_{2n}=-\Pf(d_0,d_1,0,\ldots,2n-1),\label{7-2-17-2}\\
&\frac{\partial}{\partial t_1}\tau_{2n+1}=\Pf(d_1,0,\ldots,2n),\label{7-2-18-2}\\
&\frac{\partial}{\partial t_1}(\tau_{2n}P_{2n}(z;t))
=-\Pf(d_0,d_1,0,\ldots,2n,z),\label{7-2-21-2}\\
&\frac{\partial}{\partial t_1}(\tau_{2n+1}P_{2n+1}(z;t))
=\Pf(d_1,0,\ldots,2n+1,z).\label{7-2-22-2}
\end{align}
\end{subequations}
\end{lemma}

After comparing the relations in Lemma \ref{lem:t1_wron} and Lemma \ref{lem:t1}, we immediately obtain the following identities.
\begin{lemma}\label{le1}
As for the specific Pfaffians, there hold the following relations:
\begin{subequations}
\begin{align}
&\Pf(0,1,\ldots,2n-2,2n)=-\Pf(d_0,d_1,0,\ldots,2n-1),\label{7-2-8}\\
&-\Pf(d_0,d_1,0,\ldots,2n,z)
=\Pf(0,\ldots,2n-1,2n+1,z)-z \Pf(0,\ldots,2n,z),\label{7-2-9}\\
&\Pf(d_0,0,\ldots,2n-1,2n+1)=\Pf(d_1,0,\ldots,2n),\label{7-2-10}\\
&\Pf(d_1,0,\ldots,2n+1,z)
=\Pf(d_0,0,\ldots,2n,2n+2,z)-z \Pf(d_0,0,\ldots,2n+1,z).\label{7-2-11}
\end{align}
  
\end{subequations}
\end{lemma}
In general, it is not hard to obtain the derivative expressions for $\tau_n$ with respect to $t$ by using the derivative formulae~\eqref{pf1} and~\eqref{pf2}.
\begin{lemma}
Under the assumptions \eqref{7-2-1} and \eqref{7-2-2-3}, we have
\begin{subequations}\label{7-2-12}
\begin{align}
\frac{d}{dt}\tau_{2n}=&n(2n-1)\alpha\tau_{2n}+\alpha_1 \Pf(0,1,\ldots,2n-2,2n)\nonumber\\
&+\alpha_2 (-\Pf(0,1,\ldots,2n-3,2n-1,2n)+\Pf(0,1,\ldots,2n-2,2n+1)),\\
\frac{d}{dt}\tau_{2n+1}=&n(2n+1)\alpha\tau_{2n+1}+\alpha_1 \Pf(d_0,0,1,\ldots,2n-1,2n+1)\nonumber\\
&+\alpha_2 (-\Pf(d_0,0,1,\ldots,2n-2,2n+1,2n)+\Pf(d_0,0,1,\ldots,2n-1,2n+2)).
\end{align}
\end{subequations}
\end{lemma}

The expressions for the derivative of $\tau_nP_{n}(z;t)$ with respect to the $t$ can also be derived.
\begin{lemma}\label{th11}
Under the assumptions \eqref{7-2-1} and \eqref{7-2-2-3}, we have
\begin{subequations}
\begin{align}
&\frac{d}{dt}(\tau_{2n}P_{2n}(z;t))\nonumber\\
=&\alpha n(2n+1)\Pf(0,1,\ldots,2n,z)\nonumber\\
&+\alpha_1\Big(\Pf(0,\ldots,2n-1,2n+1,z)-z \Pf(0,1,\ldots,2n,z)\Big)\nonumber\\
&-\alpha_2\Big(\Pf(0,1,\ldots,2n-2,2n,2n+1,z)\nonumber\\
&\qquad\quad-\Pf(0,1,\ldots,2n-1,2n+2,z)+z^2 \Pf(0,1,\ldots,2n,z)\Big)\label{7-2-13}\\
&\frac{d}{dt}(\tau_{2n+1}P_{2n+1}(z;t))\nonumber\\
=&\alpha (n+1)(2n+1)\Pf(d_0,0,1,\ldots,2n+1,z)\nonumber\\
&+\alpha_1\big(-z \Pf(d_0,0,\ldots,2n+1,z)+\Pf(d_0,0,\ldots,2n,2n+2,z)\big)\nonumber\\
&-\alpha_2\big(\Pf(d_0,0,\ldots,2n-1,2n+1,2n+2,z)\nonumber\\
&\qquad\quad-\Pf(d_0,0,\ldots,2n,2n+3,z)+z^2 \Pf(d_0,0,\ldots,2n+1,z)\big).\label{7-2-14}
\end{align}
\end{subequations}
\end{lemma}
\begin{proof}
See Appendix \ref{pth11} for the proof.
\end{proof}

Now we are ready to present the time evolution equation of the partial-skew OPs $\{P_n(z;t)\}$. By employing \eqref{7-2-t1}, \eqref{7-2-12},  Lemma \ref{th11} and Corollary \ref{co:7-2} as well as the Pfaffian identities, we can derive the following theorem.  

\begin{theorem}\label{th12}
Under the assumptions \eqref{7-2-1} and \eqref{7-2-2-3}, the monic partial-skew OPs $\{P_n(z;t)\}_{n\in \mathbb{N}}$ in \eqref{psop} with the inner product \eqref{skew_kernel}  admit the time evolution 
\begin{align}
&\frac{d}{dt}P_n(z;t)\nonumber\\
=&n \alpha P_n+\alpha_1(P_{n+1}+(-z+b_{n+1}-b_n)P_n-u_nu_{n+1}P_{n-1})\nonumber\\
&+\alpha_2\Big(P_{n+2}+(b_{n+2}-b_n)P_{n+1}-u_{n-1}u_n^2u_{n+1}P_{n-2}+u_nu_{n+1}(b_n-b_{n+2})P_{n-1}\nonumber\\
&\qquad\quad+\big(-z^2+(b_{n+1}-b_n)^2-u_{n+1}u_{n+2}-u_{n+1}(b_n-b_{n+2})\nonumber\\
&\qquad\qquad\qquad\qquad\qquad\qquad\qquad-u_nu_{n+1}-u_n(b_{n-1}-b_{n+1})\big)P_n\Big).\label{7-2-25}
\end{align} 
\begin{proof}
    See Appendix \ref{pth12}.
\end{proof}
\end{theorem}

By utilizing the four-term recurrence relation \eqref{7-1-14}, we can rewrite the derivative expressions for $P_n(z;t)$. As a consequence, the Lax pair, consisting of the recurrence relation and the time evolution, can be expressed in matrix form.

In fact, it follows from the four-term recurrence relation that
\begin{align*}
P_{n+2}=&z(P_{n+1}-u_{n+1}P_{n})-(b_{n+2}-b_{n+1}-u_{n+1})P_{n+1}\nonumber\\
&+u_{n+1}(-b_{n+2}+b_{n+1}+u_{n+2})P_{n}-u_nu_{n+1}^2P_{n-1},\\
P_{n+1}=&z(P_{n}-u_{n}P_{n-1})-(b_{n+1}-b_{n}-u_{n})P_{n}\nonumber\\
&+u_{n}(-b_{n+1}+b_{n}+u_{n+1})P_{n-1}-u_{n-1}u_{n}^2P_{n-2},\\
P_{n-2}=&-\frac{1}{u_{n-1}u_n^2}(P_{n+1}+(b_{n+1}-b_{n}-u_{n})P_{n}\nonumber\\
&-u_{n}(-b_{n+1}+b_{n}+u_{n+1})P_{n-1}-z(P_{n}-u_{n}P_{n-1})),\\
P_{n-3}=-&\frac{1}{u_{n-1}^2u_{n-2}}(P_{n}+(b_{n}-b_{n-1}-u_{n-1})P_{n-1}\nonumber\\
&-u_{n-1}(-b_{n}+b_{n-1}+u_{n})P_{n-2}-z(P_{n-1}-u_{n-1}P_{n-2})),
\end{align*}
employing  which, we can derive the following evolution relations 
\begin{align*}
&\frac{d}{dt}P_{n}(z;t)\nonumber\\
=&(n\alpha+\alpha_1u_n+\alpha_2u_n(-b_{n-1}+2b_{n+1}-b_n+z))P_{n}\nonumber\\
&+(-\alpha_1u_n(z+b_{n+1}-b_{n})\nonumber\\
&+\alpha_2u_n(-z^2+(b_{n+1}-b_{n})(-2z-b_{n+1}+b_{n})-u_{n+1}(b_{n+2}-b_{n})))P_{n-1}\\
&+u_{n-1}u_n^2(-\alpha_1-\alpha_2(z+b_{n+1}-b_n+2u_{n+1}))P_{n-2},\\
&\frac{d}{dt}P_{n-1}(z;t)\nonumber\\
=&(\alpha_1+\alpha_2(z+b_n-b_{n-1}+2u_{n}))P_{n}\nonumber\\
&+((n-1)\alpha+\alpha_1(b_n-b_{n-1}-z)\nonumber\\
&\,+\alpha_2(-z^2+(b_n-b_{n-1})^2-u_{n-1}(b_{n-2}-b_n)+2u_{n}(b_n-b_{n-1}-u_{n-1}-z)))P_{n-1}\nonumber\\
&+u_{n-1}u_{n}(-\alpha_1+\alpha_2(b_n-b_{n+1}-2u_{n}+z))P_{n-2},\\
&\frac{d}{dt}P_{n-2}(z;t)\nonumber\\
=&\frac{1}{u_{n-1}}(\alpha_1-\alpha_2(b_{n-1}-b_{n}-2u_{n-1}+z))P_{n}\nonumber\\
&+(\frac{1}{u_{n-1}}(\alpha_1+\alpha_2(-b_{n-1}+b_{n}-z))(b_{n}-b_{n-1}-z)+\alpha_2(b_n-b_{n-2}))P_{n-1}\nonumber\\
&+((n-2)\alpha+\alpha_1(b_n-b_{n-2}-u_{n})\nonumber\\
&\,+\alpha_2(b_n(b_n-2b_{n-1}+2u_{n-1}+u_{n})-2z^2+b_{n-1}(b_{n-1}+u_{n})+u_{n}(z-2u_{n-1})\nonumber\\
&\qquad+(b_{n-1}-b_{n-2})^2-u_{n-2}(b_{n-3}-b_{n-1})-2u_{n-1}(b_{n-2}+u_{n-2})))P_{n-2}.
\end{align*}
As a result, the recurrence relation \eqref{7-1-14} and the above time evolution equations of the partial-skew OPs $\{P_n(z;t)\}_{n\in \mathbb{N}}$ can yield the following matrix form
\begin{align}
\varphi_{n+1}=A_n\varphi_n,\quad \frac{d\varphi_n}{dt}=B_n\varphi_n,\label{7-2-29}
\end{align}
where $\varphi(z;t)=(P_{n-2}(z;t),P_{n-1}(z;t),P_{n}(z;t))^T$ and
\begin{align*}
&A_n=\left(
\begin{array}{ccc}
 0 & 1 & 0 \\
 0 & 0 & 1 \\
 -u_n^2 u_{n-1} & u_n \left(-b_{n+1}+b_n+u_{n+1}\right)-z  u_n & -b_{n+1}+b_n+z +u_n \\
\end{array}
\right),\\
&B_n=\left(
\begin{array}{ccc}
 v_{11} & v_{12} & v_{13} \\
 v_{21} & v_{22} & v_{23} \\
 v_{31} & v_{32} & v_{33} \\
\end{array}
\right),
\end{align*}
with
\begin{align}
v_{11}=&(n-2)\alpha+\alpha_1(b_n-b_{n-2}-u_{n})+\alpha_2(b_n(b_n-2b_{n-1}+2u_{n-1}-u_{n})\nonumber\\
&+z(u_{n}-2z)+b_{n-1}(b_{n-1}+u_{n})-2u_{n-1}u_n+(b_{n-1}-b_{n-2})^2\nonumber\\
&-u_{n-2}(b_{n-3}-b_{n-1})-2u_{n-1}(b_{n-2}+u_{n-2})),\nonumber\\
v_{12}=&\frac{1}{u_{n-1}}(\alpha_1+\alpha_2(-b_{n-1}+b_{n}-z))(b_{n}-b_{n-1}-z)+\alpha_2(b_n-b_{n-2}),\nonumber\\
v_{13}=&\frac{1}{u_{n-1}}(\alpha_1-\alpha_2(b_{n-1}-b_{n}-2u_{n-1}+z)),\nonumber\\
v_{21}=&u_{n-1}u_{n}(-\alpha_1+\alpha_2(b_n-b_{n+1}-2u_{n}+z)),\nonumber\\
v_{22}=&(n-1)\alpha+\alpha_1(b_n-b_{n-1}-z)+\alpha_2(-z^2+(b_n-b_{n-1})^2-u_{n-1}(b_{n-2}-b_n)\nonumber\\
&+2u_{n}(b_n-b_{n-1}-u_{n-1}-z)),\nonumber\\
v_{23}=&\alpha_1+\alpha_2(z+b_n-b_{n-1}+2u_{n}),\nonumber\\
v_{31}=&u_{n-1}u_n^2(-\alpha_1-\alpha_2(z+b_{n+1}-b_n+2u_{n+1})),\nonumber\\
v_{32}=&-\alpha_1u_n(z+b_{n+1}-b_{n})+\alpha_2u_n(-z^2+(b_{n+1}-b_{n})(-2z-b_{n+1}+b_{n})-u_{n+1}(b_{n+2}-b_{n})),\nonumber\\
v_{33}=&n\alpha+\alpha_1u_n+\alpha_2u_n(-b_{n-1}+2b_{n+1}-b_n+z).\nonumber
\end{align}
The compatibility condition of the overdetermined system \eqref{7-2-29} yields a matrix representation of the nonisospectral integrable lattice
\begin{align*}
\frac{dA_n}{dt}=B_{n+1}A_n-A_nB_n,
\end{align*}
from which we can write down the explicit differential equations for $\{u_n,b_n\}$ so that the following theorem is derived.
\begin{theorem}
Under the assumptions \eqref{7-2-1} and \eqref{7-2-2-3}, the variables $\{u_n(t),b_n(t)\}$ in the four-term recurrence relation \eqref{7-1-14}
for the monic partial-skew OPs $\{P_n(z;t)\}_{n\in \mathbb{N}}$ in \eqref{psop} with the inner product \eqref{skew_kernel} satisfy the following generalized nonisospectral B-Toda lattice
\begin{subequations}\label{7-2-30}
\begin{align}
\frac{d}{dt}u_n=&\alpha u_{n}+\alpha_1u_{n}(b_{n+1}-2b_{n}+b_{n-1})\nonumber\\
&+\alpha_2u_{n}[b_{n+1}^2-b_{n-1}^2-2b_{n}(b_{n+1}-b_{n-1})\nonumber\\
&+u_{n-1}(2u_{n}-b_{n}+b_{n-2})-u_{n+1}(2u_{n}-b_{n+2}+b_{n})],\label{7-2-31}\\
\frac{d}{dt}b_n=&\alpha b_{n}+\alpha_1u_{n}(b_{n+1}-b_{n-1})+\alpha_2u_{n}[(b_{n+1}-b_{n-1})^2\nonumber\\
&+u_{n+1}(b_{n+2}-b_{n})+u_{n-1}(b_{n}-b_{n-2})],\label{7-2-32}
\end{align}
\end{subequations}
with the Lax pair \eqref{7-2-29}.
\end{theorem}
\begin{remark}
When $\alpha\neq0$, the above equation represents a nonisospectral generalization of the lattice  that incorporates the first and second flows of the B-Toda hierarchy. In the case of $\alpha=\alpha_2=0$, it corresponds to the first flow of the isospectral B-Toda hierarchy, while it gives the second flow of the isospectral B-Toda hierarchy in the case of $\alpha=\alpha_1=0$. The first flow was investigated in \cite{chang2018partial} etc, while it seems that the explicit form of the second flow has not been reported elsewhere. It is noted that the first and second flows of the B-Toda hierarchy given here coincide with those obtained by Krichever and Zabrodin \cite[eq. (1.3)]{krichever2023toda}, where a combination was presented.
\end{remark}

\subsection{An integrable difference system related to nonisospectral B-Toda}\label{bt to p}
In the previous subsection, we provided the Lax pair for the nonisospectral B-Toda lattice. By implementing stationary reduction, we can obtain 
 \begin{align}\label{7-3-00}
\varphi_{n+1}=F_n\varphi_n,\quad \frac{\partial\varphi_n}{\partial z}=G_n\varphi_n,
\end{align}
where
$$F_n=A_n,\quad G_n=B_n/\left(\frac{dz}{dt}\right)=\frac{1}{\alpha z}B_n.$$
 The compatibility condition of the above overdetermined system results in an integrable difference system. In summary, we have the following theorem.
\begin{theorem}
Under the assumptions \eqref{7-2-1} and \eqref{7-2-2-3} as well as stationary reduction, the variables $\{u_n(t),b_n(t)\}$ in the four-term recurrence relation \eqref{7-1-14}
for the monic partial-skew OPs $\{P_n(z;t)\}_{n\in \mathbb{N}}$ in \eqref{psop} with the inner product \eqref{skew_kernel}  satisfy the following integrable difference system with the Lax pair \eqref{7-3-00}
\begin{subequations}\label{7-3-0}
\begin{align}
&(n-2)\alpha+\alpha_1(b_n-b_{n-1})+\alpha_2[(b_n-b_{n-1})^2-2u_{n-1}u_n\nonumber\\
&+u_{n-1}(b_n-b_{n-2})+u_{n}(b_{n+1}-b_{n-1})]-\beta_0=0, \label{pain-posp1}\\
&\alpha(b_n-(2n-3)u_n)+2\alpha_2u_n((b_n-b_{n-1})(b_{n+1}-b_n)\nonumber\\
&+u_n(u_{n-1}+u_{n+1}-b_{n+1}+b_{n-1}))+2\beta_0u_n=0,\label{pain-posp2}
\end{align}
\end{subequations}
where $\beta_0=\alpha_1(b_2-b_1)+\alpha_2[(b_2-b_1)^2-2u_1u_2+u_2(b_3-b_1)+u_1(b_2-b_0)]$.
\end{theorem}

\begin{proof}
By use of the compatibility condition
$$F_{n,z}+F_nG_n-G_{n+1}F_n=0,$$
we can obtain 
\begin{subequations}\label{7-3-1}
\begin{align}
&\alpha+\alpha_1(b_{n+1}-2b_{n}+b_{n-1})+\alpha_2[b_{n+1}^2-b_{n-1}^2-2b_{n}(b_{n+1}-b_{n-1})\nonumber\\
&+u_{n-1}(2u_{n}-b_{n}+b_{n-2})-u_{n+1}(2u_{n}-b_{n+2}+b_{n})]=0,\label{7-3-2}\\
&\alpha b_{n}+\alpha_1u_{n}(-b_{n-1}+b_{n+1})+\alpha_2u_{n}[(b_{n+1}-b_{n-1})^2\nonumber\\
&+u_{n+1}(b_{n+2}-b_{n})+u_{n-1}(b_{n}-b_{n-2})]=0.\label{7-3-3}
\end{align}
\end{subequations}
It is not hard to see that summing \eqref{7-3-2} for $n$ yields \eqref{pain-posp1}.
In addition, it is readily known from \eqref{pain-posp1} that
\begin{align*}
\alpha_2u_{n-1}(b_n-b_{n-2})=&-(n-2)\alpha-\alpha_1(b_n-b_{n-1})\\
&-\alpha_2[(b_n-b_{n-1})^2-2u_{n-1}u_n+u_{n}(b_{n+1}-b_{n-1})]+\beta_0,\nonumber\\
\alpha_2u_{n+1}(b_{n+2}-b_{n})=&-(n-1)\alpha-\alpha_1(b_{n+1}-b_{n})\\
&-\alpha_2[(b_{n+1}-b_{n})^2-2u_{n+1}u_n+u_{n}(b_{n+1}-b_{n-1})]+\beta_0,\nonumber
\end{align*}
which can be used to simplify \eqref{7-3-3} to arrive at \eqref{pain-posp2}.
\end{proof}
\begin{remark}
    It is evident that Eq.\eqref{7-3-1} corresponds to the stationary form of the nonisospectral B-Toda lattice \eqref{7-2-30}. Therefore, the difference system \eqref{7-3-0} 
is integrable with a 3×3 Lax pair and it is associated with the Toda hierarchy of BKP type.
Furthermore, the
argument also implies that it exhibits solutions with a Pfaffian structure. There are indications
that it might be related to a d-P-type equation of high order and a further simplification will
be investigated in the future work.
\end{remark}

\subsection{Realization of stationary reduction}\label{btasr}
In the previous subsections, we have derived the nonisospectral B-Toda lattice and obtained the difference system by applying stationary reduction. In this subsection, we shall construct concrete moments to demonstrate that both the nonisospectral equation and its Lax pair can indeed allow the stationary reduction.

From \eqref{7-2-1}, we obviously have 
\begin{align}
x=x(0)e^{\alpha t},\quad y=y(0)e^{\alpha t},\quad z=z(0)e^{\alpha t}.
\end{align}
If we let $x(0),\ y(0) \in \mathbb{R}_{+}$,
then the moments can be rewritten as 
\begin{align*}
\mu_{i,j}(t)=&\iint_{\mathbb{R}_{+}^{2}}x^i(t)y^j(t)\frac{y(t)-x(t)}{x(t)+y(t)}w(x;t)w(y;t)dx(t)dy(t)\notag\\
=&\iint_{\mathbb{R}_{+}^{2}}x^i(0)y^j(0)e^{(i+j)\alpha t}\frac{y(0)-x(0)}{x(0)+y(0)}f(x(0);t)f(y(0);t)dx(0)dy(0),
\end{align*}
where $f(\cdot\,;t)$ needs to be determined.
Differentiating the moments leads to
\begin{align*}
\frac{d}{dt}\mu_{i,j}(t)=&(i+j)\alpha \mu_{i,j}(t)\notag\\
&+\iint_{\mathbb{R}_{+}^{2}}x^i(0)y^j(0)e^{(i+j)\alpha t}\frac{y(0)-x(0)}{x(0)+y(0)}\frac{d(f(x(0);t)f(y(0);t))}{dt}dx(0)dy(0).
\end{align*}
If one wants the moments to satisfy the time evolution in \eqref{7-2-2}, it is required that $f(x(0);t)$ and $f(y(0);t)$ adhere to the time evolution
\begin{align*}
\frac{d}{dt}f(x(0);t)&=(\alpha _1x(0)e^{\alpha t}+\alpha _2x^2(0)e^{2\alpha t})f(x(0);t),\\
\frac{d}{dt}f(y(0);t)&=(\alpha _1y(0)e^{\alpha t}+\alpha _2y^2(0)e^{2\alpha t})f(y(0);t).
\end{align*}
As a result, it is reasonable to set
\begin{align*}
f(x(0);t)&=e^{\frac{\alpha _1}{\alpha}x(0)e^{\alpha t}+\frac{\alpha _2}{2\alpha}x^2(0)e^{2\alpha t}},\qquad f(y(0);t)=e^{\frac{\alpha _1}{\alpha}y(0)e^{\alpha t}+\frac{\alpha _2}{2\alpha}y^2(0)e^{2\alpha t}},
\end{align*}
which results in the exact expressions of the moments 
\begin{align*}
\mu_{i,j}(t)
&=\iint_{\mathbb{R}_{+}^{2}}x^i(0)y^j(0)e^{(i+j)\alpha t}\frac{y(0)-x(0)}{x(0)+y(0)}e^{\frac{\alpha _1}{\alpha}(x(0)+y(0)) e^{\alpha t}+\frac{\alpha _2}{2\alpha}(x^2(0)+y^2(0))e^{2\alpha t}}dx(0)dy(0),\\ 
\beta_i(t)
&=\int_{\mathbb{R}_{+}}x^i(0)e^{i\alpha t}e^{\frac{\alpha _1}{\alpha}x(0) e^{\alpha t}+\frac{\alpha _2}{2\alpha}x^2(0)e^{2\alpha t}}dx(0).
\end{align*}
It is obvious that the single moments $\beta_i$ given above satisfy the time evolution
\begin{align*}
\frac{d}{dt}\beta_i(t)=i\alpha \beta_i(t)+\alpha_1\beta_{i+1}(t)+\alpha_2 \beta_{i+2}(t).
\end{align*}
\begin{theorem}\label{th13}
Under the assumption \eqref{7-2-1} together with $\frac{\alpha_2}{\alpha}<0$, define the moments as
\begin{subequations}\label{psop:mom}
\begin{align}
\mu_{i,j}(t)
&=\iint_{\mathbb{R}_{+}^{2}}x^i(0)y^j(0)e^{(i+j)\alpha t}\frac{y(0)-x(0)}{x(0)+y(0)}e^{\frac{\alpha _1}{\alpha}(x(0)+y(0)) e^{\alpha t}+\frac{\alpha _2}{2\alpha}(x^2(0)+y^2(0))e^{2\alpha t}}dx(0)dy(0),\label{7-4-1}\\ 
\beta_i(t)
&=\int_{\mathbb{R}_{+}}x^i(0)e^{i\alpha t}e^{\frac{\alpha _1}{\alpha}x(0) e^{\alpha t}+\frac{\alpha _2}{2\alpha}x^2(0)e^{2\alpha t}}dx(0).\label{7-4-2}
\end{align}
\end{subequations}
Then the moments simultaneously satisfy the evolution relations \eqref{7-2-2-3} and 
\begin{align}\label{7-4-3-4}
&\frac{d}{dt}\mu_{i,j}(t)=-2\alpha \mu_{i,j}(t),
\qquad \frac{d}{dt}\beta_{i}(t)=-\alpha \beta_{i}(t).
\end{align}
\end{theorem}

\begin{proof}
It is sufficient to prove \eqref{7-4-3-4}. By performing integration by parts with respect to $x(0)$ and observing the fact at the boundary
$$\lim_{x(0)\to0}x(0)f(x(0);t)=\lim_{x(0)\to +\infty}x(0)f(x(0);t)=0,$$
 we obtain
\begin{align*}
\mu_{i,j}(t)
=&\iint_{\mathbb{R}_{+}^{2}}x^i(0)y^j(0)e^{(i+j)\alpha t}\frac{y(0)-x(0)}{x(0)+y(0)}f(x(0);t)f(y(0);t)dx(0)dy(0)\notag\\
=&-\iint_{\mathbb{R}_{+}^{2}}i x^i(0)y^j(0)e^{(i+j)\alpha t}\frac{y(0)-x(0)}{x(0)+y(0)}f(x(0);t)f(y(0);t)dx(0)dy(0)\notag\\
&+\iint_{\mathbb{R}_{+}^{2}} x^{i+1}(0)y^j(0)e^{(i+j)\alpha t}\frac{2y(0)}{(x(0)+y(0))^2}f(x(0);t)f(y(0);t)dx(0)dy(0)\notag\\
&-\iint_{\mathbb{R}_{+}^{2}}x^{i+1}(0)y^j(0)e^{(i+j)\alpha t}\frac{y(0)-x(0)}{x(0)+y(0)}\frac{\alpha _1e^{\alpha t}+\alpha _2x(0)e^{2\alpha t}}{\alpha}f(x(0);t)f(y(0);t)dx(0)dy(0)).
\end{align*}
On the other hand, by applying integration by parts to the bimoments with respect to $y(0)$, we also have
\begin{align*}
\mu_{i,j}(t)
=&\iint_{\mathbb{R}_{+}^{2}}x^i(0)y^j(0)e^{(i+j)\alpha t}\frac{y(0)-x(0)}{x(0)+y(0)}f(x(0);t)f(y(0);t)dx(0)dy(0)\notag\\
=&-\iint_{\mathbb{R}_{+}^{2}}j x^i(0)y^j(0)e^{(i+j)\alpha t}\frac{y(0)-x(0)}{x(0)+y(0)}f(x(0);t)f(y(0);t)dx(0)dy(0)\notag\\
&-\iint_{\mathbb{R}_{+}^{2}} x^{i}(0)y^{j+1}(0)e^{(i+j)\alpha t}\frac{2x(0)}{(x(0)+y(0))^2}f(x(0);t)f(y(0);t)dx(0)dy(0)\notag\\
&-\iint_{\mathbb{R}_{+}^{2}}x^{i}(0)y^{j+1}(0)e^{(i+j)\alpha t}\frac{y(0)-x(0)}{x(0)+y(0)}\frac{\alpha _1e^{\alpha t}+\alpha _2y(0)e^{2\alpha t}}{\alpha}f(x(0);t)f(y(0);t)dx(0)dy(0).
\end{align*}
By taking a summation on the above two equations, we see that the bimoments satisfy the recurrence relationship
\begin{align}
(i+j+2)\alpha \mu_{i,j}(t)+\frac{\alpha _1}{\alpha}(\mu_{i+1,j}(t)+\mu_{i,j+1}(t))+\frac{\alpha _2}{\alpha} (\mu_{i+2,j}(t)+\mu_{i,j+2}(t))=0.\label{7-4-7}
\end{align}
In a similar manner, we also get 
\begin{align*}
\beta_i(t)
=&\int_{\mathbb{R}_{+}}x^i(0)e^{i\alpha t}e^{\frac{\alpha _1}{\alpha}x(0) e^{\alpha t}+\frac{\alpha _2}{2\alpha}x^2(0)e^{2\alpha t}}dx(0)\notag\\
=&-\int_{\mathbb{R}_{+}}ix^i(0)e^{i\alpha t}e^{\frac{\alpha _1}{\alpha}x(0) e^{\alpha t}+\frac{\alpha _2}{2\alpha}x^2(0)e^{2\alpha t}}dx(0)\notag\\
&-\int_{\mathbb{R}_{+}}\frac{\alpha _1}{\alpha}x^{i+1}(0)e^{(i+1)\alpha t}e^{\frac{\alpha _1}{\alpha}x(0) e^{\alpha t}+\frac{\alpha _2}{2\alpha}x^2(0)e^{2\alpha t}}dx(0)\notag\\
&-\int_{\mathbb{R}_{+}}\frac{\alpha _2}{\alpha}x^{i+2}(0)e^{(i+2)\alpha t}e^{\frac{\alpha _1}{\alpha}x(0) e^{\alpha t}+\frac{\alpha _2}{2\alpha}x^2(0)e^{2\alpha t}}dx(0)\nonumber\\
=&-i\beta_i(t)-\frac{\alpha _1}{\alpha}\beta_{i+1}(t)-\frac{\alpha _2}{\alpha}\beta_{i+2}(t),
\end{align*}
which leads to the recurrence relationship
\begin{align}
(i+1)\beta_i(t)+\frac{\alpha _1}{\alpha}\beta_{i+1}(t)+\frac{\alpha _2}{\alpha}\beta_{i+2}(t)=0.\label{7-4-9}
\end{align}
 Finally, it is straightforward to derive the desired result by utilizing the time evolution of the moments \eqref{7-2-2-3} along with the recurrence relationships  \eqref{7-4-7} and \eqref{7-4-9}.
\end{proof}
Applying the time evolution \eqref{7-4-3-4} for the moments  and the derivative formulae for Pfaffians \eqref{pf5} and \eqref{pf6}, we can conclude the following result.
\begin{theorem}
    Given the moments in \eqref{psop:mom}, we have
    \begin{align*}
\frac{d}{dt}u_n(t)=\frac{d}{dt}b_n(t)=0
\end{align*}
and 
$$\frac{d}{dt}\gamma_{n,j}(t)=0,$$
where $\{u_n\}$ and $\{b_{n}\}$ are the variables  in the four-term recurrence relation 
\eqref{7-1-14}
for the monic partial-skew OPs $\{P_n(z;t)\}_{n\in \mathbb{N}}$ in \eqref{psop} with the inner product \eqref{skew_kernel} and  $\{\gamma_{n,j}\}$ are the coefficients of the monic partial-skew OPs with the expansion $P_n(z)=\sum_{j=0}^n\gamma_{n,j}z^j$.
\end{theorem}
\begin{proof}
 By use of ~\eqref{7-4-3-4} and the derivative formulae for Pfaffians \eqref{pf5} and \eqref{pf6},  it is easy to see that the Pfaffians $ \tau_n$  satisfy
\begin{align*}
\frac{d}{dt}\tau_{n}=-n\alpha\tau_{n}.
\end{align*}
Recall that we have
\begin{align*}
u_n=\frac{\tau_{n+1}\tau_{n-1}}{\tau_{n}^2},\qquad b_{n}=\frac{\sigma_n}{\tau_n}.
\end{align*}
By taking derivative with respect to $t$ for $u_n$ and $b_n$
and applying the above evolution relation,
we conclude that $\{u_n,b_n\}_{n\in \mathbb{N}}$ are independent of time $t$.

As for the coefficients $\gamma_{n,j}$ in the expansion of $P_n(z)$, it follows from the Pfaffian expression of $\{P_n(z)\}_{n=0}^{\infty}$ that
$$\gamma_{n,j}=\frac{T_{n,j}}{\tau_{n}},$$
where 
\begin{align*}  
&T_{2n,j}=(-1)^j\Pf(0,1,\ldots,\hat{j},\ldots,2n),\\
&T_{2n+1,j}=(-1)^{j+1}\Pf(d_0,0,1,\ldots,\hat{j},\ldots,2n+1).
\end{align*}
Employing the derivative formulae for Pfaffians \eqref{pf5} and \eqref{pf6}, one can see that
\begin{align*}
    \frac{d}{dt}T_{n,j}=-n\alpha T_{n,j},
\end{align*}
based on which it is not hard to see $\frac{d}{dt}\gamma_{n,j}=0$.

\end{proof}

The above theorem implies that the nonisospectral B-Toda lattice \eqref{7-2-30} can indeed allow stationary reduction so that an integrable difference system \eqref{7-3-0} is obtained. Furthermore, since the coefficients of the partial-skew OPs $\{P_n(z;t)\}_{n\in \mathbb{N}}$ are independent of $t$, it is valid that the Lax pair \eqref{7-2-29} of the nonisospectral B-Toda lattice  is transferred into a stationary form, resulting in the Lax pair \eqref{7-3-00} of the difference system. We also note that the integrable
difference system \eqref{7-3-0} admits a solution in terms of Pfaffians.

\section{Conclution and discussions}\label{cd}

We have developed a new approach called ``stationary reduction method based on nonisospectral deformation of OPs" for generating d-P-type equations and its effectiveness has been demonstrated by considering different classes of bi-OPs. As a result, we are enabled to obtain diverse families of d-P-type equations, along with their particular solutions and associated Lax pairs. It seems that some of the derived integrable difference systems exhibit several new features. In particular, the integrable difference systems related to partial-skew orthogonality admits a solution expressed in terms of Pfaffians. Some of the obtained integrable difference systems have been related to d-P-type equations, 
whereas the others, which we believe are also related, deserve further investigation.

\section{Acknowledgement}
We thank Professors D. Dai, A.N.W. Hone, N. Joshi, F. Nijhoff, W. Van Assche and Anton Dzhamay for their encouraging comments and valuable communications.
X.K. Chang's research was supported by the National Natural Science Foundation of China (Grant Nos. 12571270, 12222119, and 12288201),  Beijing Natural Science Foundation (Grant No.~JQ26004), and the Youth Innovation Promotion Association CAS. X.B Hu was supported in part by the National Natural Science Foundation of China (Grant Nos. 12671297, 11931017 and 12071447). 
X.L. Yue was partially supported by grants from the Research Grants Council
of the Hong Kong Special Administrative Region, China (Project No. CityU 11306723 and CityU 11301924)
and the National Natural Science Foundation of China (Grant Nos. 12571272 and 12171474).

\appendix
\section{On the Pfaffians}\label{chap:chaapp}
Consider a skew symmetric matrix $A$ of order $n$
\begin{align*}
A=(a_{i,j})_{i,j=1}^{n}.
\end{align*}
It is obvious that the determinant of $A$ is zero when $n$ is odd. We remark that, when $n$ is even, specifically in the case of $n=2m$, the determinant of $A$ is equal to the square of a Pfaffian of order $m$ (see e.g. \cite{cayley1849determinants,hirota2004direct})
\begin{align}\label{intro7}
\det(A)=(\Pf(1,2,\ldots,2m))^2,
\end{align}
where the Pfaffian entries $\Pf(i,j)=a_{i,j}$.
In general, following e.g. \cite{hirota2004direct}, a Pfaffian of order $m$ is defined according to the following expansion based on Pfaffian entries
\begin{align}\label{intro8}
\Pf(1,2,\ldots,2m)=\sum_P(-1)^P \Pf(i_1,i_2)\Pf(i_3,i_4)\Pf(i_5,i_6)\cdots \Pf(i_{2m-1},i_{2m}),
\end{align}
where $\sum_P$ denotes the summation over all pairs satisfying the conditions
$$
i_1<i_2,\ i_3<i_4,\ i_5<i_6,\ \ldots,i_{2m-1}<i_{2m},
$$
$$i_1<i_3<i_5<\cdots<i_{2m-1},
$$
with $i_j$ belongs to the set $\{1,2,\ldots,2m\}$, and $(-1)^P$ takes the value of $+1\ (-1)$ respectively if $\{i_1, i_2,\ldots, i_{2m}\}$ represents an even (odd) permutation of $\{1, 2,\ldots, 2m\}$.
It can also be shown that the Pfaffian admits the following expansion
\begin{align}\label{intro10}
&\Pf(1,2,\ldots,2m)\nonumber\\
=&\Pf(1,2)\Pf(3,4,\ldots,2m)-\Pf(1,3)\Pf(2,4,5,\ldots,2m)\nonumber\\
&+\Pf(1,4)\Pf(2,3,5,\ldots,2m)-\cdots+\Pf(1,2m)\Pf(2,3,\ldots,2m-1)\nonumber\\
=&\sum_{j=2}^{2m}(-1)^j\Pf(1,j)\Pf(2,3,\ldots,\hat{j},\ldots,2m),
\end{align}
where $\hat{j}$ represents the removal of the element $j$.

\subsection{Formulae related to determinants and Pfaffians}
There are a number of interesting connections between certain determinants and Pfaffians (see e.g. \cite[Appendix A.1]{chang2022hermite}).

(1). For a skew-symmetric matrix $A_{2n-1}$ of size $2n-1$ augmented with an arbitrary row and column, there holds
\begin{align}
\det\left( \begin{array}{ccc|c}
& & & x_1 \\
& A_{2n-1}& &\vdots \\
& &  & x_{2n-1} \\ \hline
-y_1&\cdots&-y_{2n-1}&z
\end{array}
\right)=\Pf(B_{2n})\Pf(C_{2n}),\label{7-1-7}
\end{align}
where
\begin{align}
B_{2n}=\left( \begin{array}{ccc|c}
& & & x_1 \\
& A_{2n-1}& &\vdots \\
& &  & x_{2n-1} \\ \hline
-x_1&\cdots&-x_{2n-1}&0
\end{array}
\right),\quad
C_{2n}=\left( \begin{array}{ccc|c}
& & & y_1 \\
& A_{2n-1}& &\vdots \\
& &  & y_{2n-1} \\ \hline
-y_1&\cdots&-y_{2n-1}&0
\end{array}
\right).\nonumber
\end{align}

(2). For a skew-symmetric matrix $A_{2n}$ of size $2n$ augmented with an arbitrary row and column, there holds
\begin{align}
\det\left( \begin{array}{ccc|c}
& & & x_1 \\
& A_{2n}& &\vdots \\
& &  & x_{2n} \\ \hline
-y_1&\cdots&-y_{2n}&z
\end{array}
\right)=\Pf(A_{2n})\Pf(B_{2n+2}),\label{7-1-8}
\end{align}
where
\begin{align}
B_{2n+2}=\left( \begin{array}{ccc|cc}
& & & x_1 & -y_1\\
& A_{2n}& &\vdots&\vdots \\
& &  & x_{2n} & -y_{2n}\\ \hline
-x_1&\cdots&-x_{2n}&0&z\\
y_1&\cdots&y_{2n}&-z&0
\end{array}
\right).\nonumber
\end{align}

\subsection{Derivative formulae for some special Pfaffians}\label{DerPfaffian}
For some special Pfaffians, there hold some derivative formulae (see e.g. \cite[Chapter 2]{hirota2004direct}).

(1). If the $t$-derivative of a Pfaffian entry $\Pf(i,j)$ satisfy 
\begin{align*}
\frac{d}{dt}\Pf(i,j)=&\alpha(i+j) \Pf(i,j)+\alpha_1(\Pf(i+1,j)+\Pf(i,j+1))\\
&\quad+\alpha_2 (\Pf(i+2,j)+\Pf(i,j+2)),
\end{align*} 
then
\begin{align}
&\frac{d}{dt}\Pf(i_1,\ldots, i_{2N})\nonumber\\
=&\alpha \sum_{k=i_1}^{i_{2N}}k \Pf(i_1,\ldots, i_{2N})+\alpha_1 \sum_{k=i_1}^{i_{2N}}\Pf(i_1,i_2,\ldots,i_{k}+1,\ldots, i_{2N})\nonumber\\
&\quad+\alpha_2  \sum_{k=i_1}^{i_{2N}}\Pf(i_1,i_2,\ldots,i_{k}+2,\ldots, i_{2N}).\label{pf1}
\end{align}

 (2). If 
\begin{align*}
\frac{d}{dt}\Pf(i,j)=&\alpha(i+j) \Pf(i,j)+\alpha_1(\Pf(i+1,j)+\Pf(i,j+1))\\
&\quad+\alpha_2 (\Pf(i+2,j)+\Pf(i,j+2))
\end{align*} 
 and 
 \begin{align*}
 \frac{d}{dt} \Pf(a_0,i)=\alpha i \Pf(a_0,i)+\alpha_1 \Pf(a_0,i+1)+\alpha_2 \Pf(a_0,i+2),
 \end{align*} 
  then
\begin{align}
&\frac{d}{dt}\Pf(a_0,i_1,\ldots, i_{2N-1})\nonumber\\
=&\alpha \sum_{k=i_1}^{i_{2N-1}}k \Pf(a_0,i_1,\ldots, i_{2N-1})+\alpha_1 \sum_{k=i_1}^{i_{2N-1}}\Pf(a_0,i_1,i_2,\ldots,i_{k}+1,\ldots, i_{2N-1})\nonumber\\
&\quad+\alpha_2 \sum_{k=i_1}^{i_{2N}}\Pf(a_0,i_1,i_2,\ldots,i_{k}+2,\ldots, i_{2N-1}).\label{pf2}
\end{align}

(3). If 
\begin{align*}
\frac{d}{dt}\Pf(i,j)=\alpha(i+j) \Pf(i,j)+\alpha_1 \Pf(a_0,b_0,i,j)+\alpha_2 (\Pf(i+2,j)+\Pf(i,j+2))
 \end{align*} 
 and $\Pf(a_0,b_0)=0$, then
\begin{align}
&\frac{d}{dt}\Pf(i_1,\ldots, i_{2N})\nonumber\\
=&\alpha \sum_{k=i_1}^{i_{2N}}k \Pf(i_1,\ldots, i_{2N})+\alpha_1 \Pf(a_0,b_0,i_1,\ldots, i_{2N})\nonumber\\
&+\alpha_2  \sum_{k=i_1}^{i_{2N}}\Pf(i_1,i_2,\ldots,i_{k}+2,\ldots, i_{2N}).\label{pf3}
\end{align}

 (4). If 
\begin{align*}
\frac{d}{dt}\Pf(i,j)=\alpha(i+j) \Pf(i,j)+\alpha_1 \Pf(a_0,b_0,i,j)+\alpha_2 (\Pf(i+2,j)+\Pf(i,j+2))
 \end{align*} 
 and 
 \begin{align*}
 \frac{d}{dt}\Pf(a_0,j)=\alpha i \Pf(a_0,j)+\alpha_1 \Pf(b_0,j)+\alpha_2 \Pf(a_0,i+2),
  \end{align*}  then
\begin{align}
&\frac{d}{dt}\Pf(a_0,i_1,\cdots i_{2N-1})\nonumber\\
=&\alpha \sum_{k=i_1}^{i_{2N-1}}k \Pf(a_0,i_1,\ldots, i_{2N-1})+\alpha_1 \Pf(b_0,i_1,i_2,\ldots, i_{2N-1})\nonumber\\
&\quad+\alpha_2 \sum_{k=i_1}^{i_{2N}}\Pf(a_0,i_1,i_2,\ldots,i_{k}+2,\ldots, i_{2N-1}).\label{pf4}
\end{align}

 (5). If 
  \begin{align*} 
 \frac{d}{dt}\Pf(i,j)=-2\alpha \Pf(i,j),
  \end{align*}  then
\begin{align}
\frac{d}{dt}\Pf(i_1,\cdots i_{2N})
=-2N\alpha \Pf(i_1,\cdots i_{2N}).\label{pf5}
\end{align}

(6). If 
\begin{align*}
\frac{d}{dt}\Pf(i,j)=-2\alpha \Pf(i,j), \quad \text{and}\quad \frac{d}{dt} \Pf(a_0,i)=-\alpha  \Pf(a_0,i),
 \end{align*} 
  then
\begin{align}
\frac{d}{dt}\Pf(a_0,i_1,\cdots i_{2N+1})
=-(2N+1)\alpha  \Pf(a_0,i_1,\cdots i_{2N+1}).\label{pf6}
\end{align}

\subsection{Bilinear identities on Pfaffians}\label{BilinearPfaffian}
Here we present two commonly used Pfaffian identities (see e.g. \cite[Chapter 2]{hirota2004direct}), that is, 
\begin{align}
&\Pf(a_1,a_2,a_3,a_4,1,2,\ldots,2N)\Pf(1,2,\ldots,2N)\nonumber\\
=&\sum_{j=2}^{4}(-1)^j \Pf(a_1,a_j,1,2,\ldots,2N)\Pf(a_2,\hat{a}_j,a_4,1,2,\ldots,2N),\label{iden1}\\
&\Pf(a_1,a_2,a_3,1,2,\ldots,2N-1)\Pf(1,2,\ldots,2N)\nonumber\\
=&\sum_{j=1}^{3}(-1)^{j-1}\Pf(a_j,1,2,\ldots,2N-1)\Pf(a_1,\hat{a}_j,a_3,1,2,\ldots,2N).\label{iden2}
\end{align}

\section{Proof of Lemma \ref{th9}}\label{cauchy6.1}
\begin{proof}
As for the biorthogonal relation
~\eqref{6-1-6} 
\begin{align*}
\langle P_n(x;t),P_m(y;t)\rangle=0,\quad\quad m=0,1,\ldots,n-1,
\end{align*}
by differentiating it
with respect to time $t$, we get
\begin{align}
0=&\iint \frac{\left(\frac{d}{dt}(P_n(x;t))P_m(y;t)+P_n(x;t)\frac{d}{dt}P_m(y;t)\right)(x+y)}{(x+y)^2}w(x;t) w(y;t)dxdy\nonumber\\
&-\iint \frac{\alpha (x+y)P_n(x;t)P_m(y;t)}{(x+y)^2}w(x;t) w(y;t)dxdy\nonumber\\
&+\iint \frac{P_n(x;t)P_m(y;t)}{x+y}\left(\frac{d}{dt}w(x;t)w(y;t)+ w(x;t)\frac{d}{dt}w(y;t)+2\alpha w(x;t)w(y;t)\right)dxdy\nonumber\\
=&\left\langle\frac{d}{dt}P_n(x;t),P_m(y;t)\right\rangle+\langle(\alpha_1(x+y)+\alpha_2(x^2+y^2))P_n(x;t),P_m(y;t)\rangle,\label{6-2-6}
\end{align}
where we used the orthogonality and \eqref{6-2-3} in the last step. In addition, it obviously follows from \eqref{6-1-9} that
\begin{align}
&\langle (x+y)(P_{n+1}(x;t)+a_{n+1}P_n(x;t)),P_m(y;t)\rangle\nonumber\\
=&\iint (P_{n+1}(x;t)+a_{n+1}P_n(x;t))P_m(y;t)w(x;t) w(y;t)dxdy\nonumber\\
=&\int (P_{n+1}(x;t)+a_{n+1}P_{n}(x;t))w(x;t)dx \int P_m(y;t)w(y;t)dy=0.\label{6-2-7}
\end{align}

Now we simplify ~\eqref{6-2-6} by use of the orthogonality ~\eqref{6-1-6} and ~\eqref{6-2-7}. First, we deal with the term $-\langle (x+y)P_n(x;t),P_m(y;t)\rangle$. Since $$\langle P_n(x;t),y P_k(y;t)\rangle=0, \quad k<n-1,$$ 
we have
\begin{align}
&-\langle (x+y)P_n(x;t),P_m(y;t)\rangle\nonumber\\
=&-\langle xP_n(x;t),P_m(y;t)\rangle-\frac{1}{a_{n+1}}\langle(x+y-x)(P_{n+1}(x;t)+a_{n+1}P_n(x;t)),P_m(y;t)\rangle\nonumber\\
=&-\langle xP_n(x;t),P_m(y;t)\rangle-\frac{1}{a_{n+1}}\iint (P_{n+1}(x;t)+a_{n+1}P_n(x;t))P_m(y;t)w(x;t) w(y;t)dxdy \nonumber\\ &+\frac{1}{a_{n+1}}\langle x(P_{n+1}(x;t)+a_{n+1}P_n(x;t)),P_m(y;t)\rangle.\nonumber
\end{align}
By employing ~\eqref{6-2-7}, we immediately get
\begin{align}
-\langle (x+y)P_n(x;t),P_m(y;t)\rangle
=&\frac{1}{a_{n+1}}\langle xP_{n+1}(x;t),P_m(y;t)\rangle\nonumber\\
=&\frac{1}{a_{n+1}}\langle (xP_{n+1}(x;t)-P_{n+2}(x;t)),P_m(y;t)\rangle.\nonumber
\end{align}
It is noted that, by rewriting the recurrence relation~\eqref{6-1-10} as
\begin{align}
xP_{n+1}(x;t)-P_{n+2}(x;t)=b_{n+1}P_{n+1}(x)+c_{n+1}P_{n}(x)+d_{n+1}P_{n-1}(x;t)-a_{n+1}xP_n(x;t),\nonumber
\end{align}
and using the orthogonality, we have
\begin{align}
&\frac{1}{a_{n+1}}\langle (xP_{n+1}(x;t)-P_{n+2}(x;t)),P_m(y;t)\rangle\nonumber\\
=&\frac{1}{a_{n+1}}\langle (d_{n+1}P_{n-1}(x;t)-a_{n+1}xP_n(x;t)),P_m(y;t)\rangle\nonumber\\
=&\left\langle \frac{d_{n+1}}{a_{n+1}}P_{n-1}(x;t)-xP_n(x;t),P_m(y;t)\right\rangle\nonumber\\
=&\left\langle \frac{d_{n+1}}{a_{n+1}}P_{n-1}(x;t)-xP_n(x;t)+P_{n+1}(x;t)-(a_n-b_n)P_n(x;t),P_m(y;t)\right\rangle.\nonumber
\end{align}
Therefore we conclude
\begin{align}
&-\langle (x+y)P_n(x;t),P_m(y;t)\rangle\nonumber\\
=&\left\langle \frac{d_{n+1}}{a_{n+1}}P_{n-1}(x;t)-xP_n(x;t)+P_{n+1}(x;t)-(a_n-b_n)P_n(x;t),P_m(y;t)\right\rangle.\label{6-2-8}
\end{align}

Next, we proceed to simplify the expression $-\alpha_2\langle (x^2+y^2)P_n(x;t),P_m(y;t)\rangle$. On one hand, we utilize the orthogonality condition and the recurrence relation~\eqref{6-1-10} as well as Eq.~\eqref{6-2-7} to get
\begin{align}
&-\alpha_2\langle P_n(x;t),y^2P_m(y;t)\rangle\nonumber\\
=&-\frac{\alpha_2}{a_{n+1}}\langle P_{n+1}(x;t)+a_{n+1}P_n(x;t),y^2P_m(y;t)\rangle+\frac{\alpha_2}{a_{n+1}}\langle P_{n+1}(x;t),y^2P_m(y;t)\rangle.\nonumber
\end{align}
One can easily observe that
\begin{align} 
&-\frac{\alpha_2}{a_{n+1}}\langle P_{n+1}(x;t)+a_{n+1}P_n(x;t),y^2P_m(y;t)\rangle\nonumber\\
=&-\frac{\alpha_2}{a_{n+1}}\langle (y(x+y)-xy)P_{n+1}(x;t)+a_{n+1}P_n(x;t),P_m(y;t)\rangle\nonumber\\
=&-\frac{\alpha_2}{a_{n+1}}\iint_{\mathbb{R}_{+}^{2}}(P_{n+1}(x;t)+a_{n+1}P_n(x;t))yP_m(y;t)w(x;t)w(y;t)dxdy\nonumber\\
&+\frac{\alpha_2}{a_{n+1}}\langle x(P_{n+1}(x;t)+a_{n+1}P_n(x;t)),yP_m(y;t)\rangle,\nonumber
\end{align}
based on which and using the recurrence relation~\eqref{6-1-10} and ~\eqref{6-2-7}, we can get
\begin{align}
&-\alpha_2\langle P_n(x;t),y^2P_m(y;t)\rangle\nonumber\\
=&\frac{\alpha_2}{a_{n+1}}\langle c_{n+1}P_n(x;t)+d_{n+1}P_{n-1}(x;t)),yP_m(y;t)\rangle+\frac{\alpha_2}{a_{n+1}}\langle P_{n+1}(x;t),y^2P_m(y;t)\nonumber\\
=&\frac{\alpha_2}{a_{n+1}}\langle c_{n+1}P_n(x;t)+d_{n+1}P_{n-1}(x;t)),yP_m(y;t)\rangle\nonumber\\
&+\frac{\alpha_2}{a_{n+1}a_{n+2}}\langle P_{n+2}(x;t)+a_{n+2}P_{n+1}(x;t),y^2P_m(y;t)\rangle\nonumber\\
=&\frac{\alpha_2}{a_{n+1}}\langle c_{n+1}P_n(x;t)+d_{n+1}P_{n-1}(x;t),yP_m(y;t)\rangle\nonumber\\
&-\frac{\alpha_2}{a_{n+1}a_{n+2}}\langle x(P_{n+2}(x;t)+a_{n+2}P_{n+1}(x;t)),yP_m(y;t)\rangle\nonumber\\
&+\frac{\alpha_2}{a_{n+1}a_{n+2}}\iint_{\mathbb{R}_{+}^{2}}(P_{n+2}(x;t)+a_{n+2}P_{n+1}(x;t))yP_m(y;t)w(x;t)w(y;t)dxdy \nonumber\\
=&\quad \cdots\nonumber\\
=&-\frac{\alpha_2d_nd_{n+1}}{a_na_{n+1}}\langle P_{n-2}(x;t),P_m(y;t)\rangle\nonumber\\
&-\frac{\alpha_2d_{n+1}}{a_{n+1}}\left(\frac{c_{n+1}}{a_{n+1}}+\frac{c_n}{a_n}-\frac{d_{n+2}}{a_{n+1}a_{n+2}}-\frac{d_{n+1}}{a_na_{n+1}}\right)\langle P_{n-1}(x;t),P_m(y;t)\rangle.\label{6-2-9}
\end{align}
On the other hand, we can simplify $-\alpha_2\langle x^2P_n(x;t),P_m(y;t)\rangle$ in the same way. In fact, we have
\begin{align}
&-\alpha_2\langle x^2P_n(x;t),P_m(y;t)\rangle\nonumber\\
=&-\alpha_2\langle x^2(P_n(x;t)+a_nP_{n-1}(x;t)),P_m(y;t)\rangle+\alpha_2a_n\langle x^2P_{n-1}(x;t),P_m(y;t)\rangle\nonumber\\
=&-\alpha_2\langle x(P_{n+1}(x;t)+b_nP_n(x;t)+c_nP_{n-1}(x;t)+d_nP_{n-2}(x;t)),P_m(y;t)\rangle\nonumber\\&+\alpha_2a_n\langle x^2(P_{n-1}(x;t)+a_{n-1}P_{n-2}(x;t)),P_m(y;t)\rangle-\alpha_2a_{n-1}a_n\langle x^2P_{n-2}(x;t),P_m(y;t)\rangle\nonumber\\
=&-\alpha_2\langle x(P_{n+1}(x;t)+b_nP_n(x;t)+c_nP_{n-1}(x;t)+d_nP_{n-2}(x;t)),P_m(y;t)\rangle\nonumber\\&+\alpha_2a_n\langle x(P_n(x;t)+b_{n-1}P_{n-1}(x;t)+c_{n-1}P_{n-2}(x;t)+d_{n-1}P_{n-3}(x;t)),P_m(y;t)\rangle\nonumber\\&-\alpha_2a_{n-1}a_n\langle x^2P_{n-2}(x;t),P_m(y;t)\rangle,\nonumber
\end{align}
where we used the recurrence relation~\eqref{6-1-10} in the last two steps. By means of orthogonality and the recurrence relation, we can derive
\begin{align}
&-\alpha_2\langle x(P_{n+1}(x;t)+b_nP_n(x;t)+c_nP_{n-1}(x;t)+d_nP_{n-2}(x;t)),P_m(y;t)\rangle\nonumber\\&+\alpha_2a_n\langle x(P_n(x;t)+b_{n-1}P_{n-1}(x;t)+c_{n-1}P_{n-2}(x;t)+d_{n-1}P_{n-3}(x;t)),P_m(y;t)\rangle\nonumber\\&-\alpha_2a_{n-1}a_n\langle x^2P_{n-2}(x;t),P_m(y;t)\rangle\nonumber\\
=&-\alpha_2\langle xP_{n+1}(x;t)-P_{n+2}(x;t),P_m(y;t)\rangle-\alpha_2(b_n-a_n)\langle  xP_n(x;t)-P_{n+1}(x;t),P_m(y;t)\rangle\nonumber\\&-\alpha_2\langle x(c_nP_{n-1}(x;t)+d_nP_{n-2}(x;t)),P_m(y;t)\rangle-\alpha_2a_{n-1}a_n\langle x^2P_{n-2}(x;t),P_m(y;t)\rangle\nonumber\\
&+\alpha_2a_n\langle x(b_{n-1}P_{n-1}(x;t)+c_{n-1}P_{n-2}(x;t)+d_{n-1}P_{n-3}(x;t)),P_m(y;t)\rangle\nonumber\\
=&\quad \cdots\nonumber\\
=&\alpha_2\langle a_{n+1}(x P_n(x;t)-P_{n+1}(x;t)+(a_n-b_n)P_n(x;t)),P_m(y;t)\rangle-\alpha_2\langle d_{n+1}P_{n-1}(x;t),P_m(y;t)\rangle\nonumber\\
&-\alpha_2(b_n-a_n)\langle c_nP_{n-1}(x;t)+d_nP_{n-2}(x;t),P_m(y;t)\rangle\nonumber\\
&+\alpha_2a_n(b_n-a_n)\langle xP_{n-1}(x;t)-P_n(x;t),P_m(y;t)\rangle\nonumber\\
&-\alpha_2\langle x(c_nP_{n-1}(x;t)+d_nP_{n-2}(x;t))-c_nP_n(x;t),P_m(y;t)\rangle\nonumber\\
&+\alpha_2a_n\langle x(b_{n-1}P_{n-1}(x;t)+c_{n-1}P_{n-2}(x;t)+d_{n-1}P_{n-3}(x;t))-b_{n-1}P_n(x;t),P_m(y;t)\rangle\nonumber\\
&-\alpha_2a_{n-1}a_n\langle x(P_{n-1}(x;t)+b_{n-2}P_{n-2}(x;t)+c_{n-2}P_{n-3}(x;t)+d_{n-2}P_{n-4}(x;t))-P_n(x;t),P_m(y;t)\rangle\nonumber\\
&+\alpha_2a_{n-2}a_{n-1}a_n\langle x^2P_{n-3}(x;t),P_m(y;t)\rangle.\nonumber
\end{align}
Consequently, we have
\begin{align}
&-\alpha_2\langle x^2P_n(x;t),P_m(y;t)\rangle\nonumber\\
=&\alpha_2\langle a_{n+1}(x P_n(x;t)-P_{n+1}(x;t)+(a_n-b_n)P_n(x;t)),P_m(y;t)\rangle\nonumber\\
&-\alpha_2\langle d_{n+1}P_{n-1}(x;t),P_m(y;t)\rangle-\alpha_2(b_n-a_n)\langle c_nP_{n-1}(x;t)+d_nP_{n-2}(x;t),P_m(y;t)\rangle\nonumber\\
&+\alpha_2a_n(b_n-a_n)\langle xP_{n-1}(x;t)-P_n(x;t),P_m(y;t)\rangle\nonumber\\
&-\alpha_2\langle x(c_nP_{n-1}(x;t)+d_nP_{n-2}(x;t))-c_nP_n(x;t),P_m(y;t)\rangle\nonumber\\
&+\alpha_2a_n\langle x(b_{n-1}P_{n-1}(x;t)+c_{n-1}P_{n-2}(x;t)+d_{n-1}P_{n-3}(x;t))-b_{n-1}P_n(x;t),P_m(y;t)\rangle\nonumber\\
&-\alpha_2a_{n-1}a_n\langle x(P_{n-1}(x;t)+b_{n-2}P_{n-2}(x;t)+c_{n-2}P_{n-3}(x;t)+d_{n-2}P_{n-4}(x;t))-P_n(x;t),P_m(y;t)\rangle\nonumber\\
&+\alpha_2a_{n-2}a_{n-1}a_n\langle x^2P_{n-3}(x;t),P_m(y;t)\rangle.\label{6-2-10}
\end{align}

Finally, by inserting Eq.~\eqref{6-2-8},~\eqref{6-2-9} and ~\eqref{6-2-10} into ~\eqref{6-2-6}, we eventually obtain
\begin{align*}
&\left\langle\frac{d}{dt}P_n(x;t)-n\alpha P_n{(x;t)},P_m(y;t)\right\rangle\nonumber\\
=&\langle-(\alpha_1-\alpha_2a_{n+1})(x P_n(x;t)-P_{n+1}(x;t)+(a_n-b_n)P_n(x;t))+\bigg(\alpha_1\frac{d_{n+1}}{a_{n+1}}-\alpha_2d_{n+1}\nonumber\\
&-\alpha_2(b_n-a_n)c_n-\frac{\alpha_2d_{n+1}}{a_{n+1}}\bigg(\frac{c_{n+1}}{a_{n+1}}+\frac{c_n}{a_n}-\frac{d_{n+2}}{a_{n+1}a_{n+2}}-\frac{d_{n+1}}{a_na_{n+1}}\bigg)\bigg)P_{n-1}(x;t)\nonumber\\
&+\alpha_2(a_n(a_n-b_n)-c_n+a_nb_{n-1}-a_{n-1}a_n)(xP_{n-1}(x)-P_n(x;t))\nonumber\\
&+\alpha_2\bigg(-\frac{d_nd_{n+1}}{a_na_{n+1}}+d_n(a_n-b_n)-x(d_n-a_nc_{n-1}+a_{n-1}a_nb_{n-2})\bigg)P_{n-2}(x;t)\nonumber\\
&+\alpha_2a_nx(xa_{n-2}a_{n-1}+d_{n-1}-a_{n-1}c_{n-2})P_{n-3}(x;t)-\alpha_2a_{n-1}a_nd_{n-2}xP_{n-4}(x;t),P_m(y;t)\rangle,
\end{align*}
where we used the fact that $\langle n\alpha P_n{(x;t)},P_m(y;t)\rangle=0$ for $m=0,1,\ldots,n-1$.
Observe that the 
$\frac{d}{dt}P_n(x;t)-n\alpha P_n(x;t),$
$x P_n(x;t)-P_{n+1}(x;t)+(a_n-b_n)P_n(x;t)
$
and $xP_{n-1}(x;t)-P_n(x;t)$
are all polynomials of degree $n-1$ in $x$, the expression \eqref{6-2-4} for $\frac{d}{dt}P_n(x;t)$  follows from the orthogonality.
Therefore, we complete the proof.
\end{proof}

\section{Proof of Corollary \ref{co:7-2}}\label{pco1}
\begin{proof}
Since \eqref{7-2-19} and \eqref{7-2-20} are straightforward consequence of applying the derivative formulae \eqref{7-2-15-16} to \eqref{7-2-17} and \eqref{7-2-18}, we are remaining to prove \eqref{7-2-23} and \eqref{7-2-24}.
Based on \eqref{7-2-21}, by use of the derivative formulae \eqref{7-2-15-16}, we get 
\begin{align}
&\frac{\partial^2}{\partial t_1^2}(\tau_{2n}P_{2n}(z;t))\nonumber\\
=&\frac{\partial}{\partial t_1}(\Pf(0,\ldots,2n-1,2n+1,z)-z \Pf(0,1,\ldots,2n,z))\nonumber\\
=&\frac{\partial}{\partial t_1}\bigg(\sum_{j=0}^{2n-1}(-z)^j\Pf(0,\ldots,\hat{j},\ldots,2n-1,2n+1)\nonumber\\
&\qquad\qquad+\sum_{j=0}^{2n-1}(-1)^{j+1}z^{j+1}\Pf(0,1,\ldots,\hat{j},\ldots,2n)\bigg)\nonumber\\
=&\sum_{j=0}^{2n-2}(-z)^j\Pf(0,\ldots,\hat{j},\ldots,2n-2,2n,2n+1)\nonumber\\
&+\sum_{j=0}^{2n-1}(-z)^j\Pf(0,\ldots,\hat{j},\ldots,2n-1,2n+2)\nonumber\\
&-2z\sum_{j=0}^{2n-2}(-z)^j\Pf(0,\ldots,\hat{j},\ldots,2n-1,2n+1)\nonumber\\
&+z^2\sum_{j=0}^{2n-2}(-z)^j\Pf(0,\ldots,\hat{j},\ldots,2n)+(-1)^{2n}z^{2n} \Pf(0,\ldots,2n-2,2n+1),\label{pf11}
\end{align}
Observe that the expansion of the Pfaffian gives
\begin{align}
&\Pf(0,\ldots,2n-1,2n+1,z)
=\sum_{j=0}^{2n-2}(-z)^j\Pf(0,\ldots,\hat{j},\ldots,2n-1,2n+1)\nonumber\\
&\quad\quad\quad\quad\quad\quad\quad\quad\quad\quad\quad\quad-z^{2n-1}\Pf(0,\ldots,2n-2,2n+1)+z^{2n+1}\Pf(0,\ldots,2n-1),\nonumber\\
&\Pf(0,\ldots,2n-2,2n,2n+1,z)
=\sum_{j=0}^{2n-2}(-z)^j\Pf(0,\ldots,\hat{j},\ldots,2n-2,2n,2n+1)\nonumber\\
&\quad\quad\quad\quad\quad\quad\quad\quad\quad\quad\quad\quad\quad\quad-z^{2n}\Pf(0,\ldots,2n-2,2n+1)+z^{2n+1}\Pf(0,\ldots,2n-2,2n),\nonumber\\
&\Pf(0,\ldots,2n-1,2n+2,z)
=\sum_{j=0}^{2n-1}(-z)^j\Pf(0,\ldots,\hat{j},\ldots,2n-1,2n+2)\nonumber\\
&\quad\quad\quad\quad\quad\quad\quad\quad\quad\quad\quad\quad+z^{2n+2}\Pf(0,\ldots,2n-1),\nonumber\\
&\Pf(0,\ldots,2n,z)
=\sum_{j=0}^{2n-2}(-z)^j\Pf(0,\ldots,\hat{j},\ldots,2n)-z^{2n-1}\Pf(0,\ldots,2n-2,2n)\nonumber\\
&\quad\quad\quad\quad\quad\quad\quad\quad+z^{2n}\Pf(0,\ldots,2n-1),\nonumber
\end{align}
from which it is not hard to see that ~\eqref{pf11} can be equivalently written as
\begin{align*}
&\frac{\partial^2}{\partial t_1^2}(\tau_{2n}P_{2n}(z;t))\nonumber\\
=&-2z \Pf(0,\ldots,2n-1,2n+1,z)+\Pf(0,\ldots,2n-2,2n,2n+1,z)\nonumber\\
&+\Pf(0,\ldots,2n-1,2n+2,z)+z^2 \Pf(0,\ldots,2n,z).
\end{align*}
This confirms the validity of~\eqref{7-2-23}.

Similarly, we can demonstrate the formula ~\eqref{7-2-24}. In fact, based on the  expression \eqref{7-2-22} and the derivative formula \eqref{7-2-15-16}, we have
\begin{align*}
&\frac{\partial^2}{\partial t_1^2}(\tau_{2n+1}P_{2n+1}(z;t))\nonumber\\
=&\frac{\partial}{\partial t_1}(\Pf(d_0,0,\ldots,2n,2n+2,z)-z \Pf(d_0,0,\ldots,2n+1,z))\nonumber\\
=&\frac{\partial}{\partial t_1}\bigg(\sum_{j=0}^{2n}(-1)^{j+1}z^j \Pf(d_0,0,\ldots,\hat{j},\ldots,2n,2n+2)\nonumber\\
&\qquad\quad-z\sum_{j=0}^{2n}(-1)^{j+1}z^j \Pf(d_0,0,\ldots,\hat{j},\ldots,2n+1)\bigg)\nonumber\\
=&\sum_{j=0}^{2n}(-1)^{j+1}z^j \Pf(d_0,0,\ldots,\hat{j},\ldots,2n,2n+3)\nonumber\\
&+\sum_{j=0}^{2n-1}(-1)^{j+1}z^j \Pf(d_0,0,\ldots,\hat{j},\ldots,2n-1,2n+1,2n+2)\nonumber\\
&+2z\sum_{j=0}^{2n-1}(-z)^j \Pf(d_0,0,\ldots,\hat{j},\ldots,2n,2n+2)\nonumber\\
&-z^2\sum_{j=0}^{2n-1}(-z)^j \Pf(d_0,0,\ldots,\hat{j},\ldots,2n+1)+z^{2n+1}\Pf(d_0,0,\ldots,2n-1,2n+2).
\end{align*}
Recall that Pfaffian's expansion gives
\begin{align}
&\Pf(d_0,0,\ldots,2n,2n+2,z)
=\sum_{j=0}^{2n-1}(-1)^{j+1}z^j \Pf(d_0,0,\ldots,\hat{j},\ldots,2n,2n+2)\nonumber\\
&\quad\quad\quad\quad\quad\quad\quad\quad\quad\quad\quad\quad-z^{2n}\Pf(d_0,0,\ldots,2n-1,2n+2)+z^{2n+2}\Pf(d_0,0,\ldots,2n),\nonumber\\
&\Pf(d_0,0,\ldots,2n,2n+3,z)
=\sum_{j=0}^{2n}(-1)^{j+1}z^j \Pf(d_0,0,\ldots,\hat{j},\ldots,2n,2n+3)\nonumber\\
&\quad\quad\quad\quad\quad\quad\quad\quad\quad\quad\quad\quad+z^{2n+3}\Pf(d_0,0,\ldots,2n),\nonumber\\
&\Pf(d_0,0,\ldots,2n-1,2n+1,2n+2,z)\nonumber\\
&\quad\quad\quad\quad=\sum_{j=0}^{2n-1}(-1)^{j+1}z^j \Pf(d_0,0,\ldots,\hat{j},\ldots,2n-1,2n+1,2n+2)\nonumber\\
&\quad\quad\quad\quad\quad\quad\quad-z^{2n+1}\Pf(d_0,0,\ldots,2n-1,2n+2)+z^{2n+2}\Pf(d_0,0,\ldots,2n-1,2n+1),\nonumber\\
&\Pf(d_0,0,\ldots,2n+1,z)
=\sum_{j=0}^{2n-1}(-1)^{j+1}z^j \Pf(d_0,0,\ldots,\hat{j},\ldots,2n+1)\nonumber\\
&\quad\quad\quad\quad\quad\quad\quad\quad\quad\quad\quad-z^{2n}\Pf(d_0,0,\ldots,2n-1,2n+1)+z^{2n+1}\Pf(d_0,0,\ldots,2n),\nonumber
\end{align}
based on which, we are led to \eqref{7-2-24} from \eqref{7-2-24}.
\end{proof}

\section{Proof of Lemma \ref{th11}}\label{pth11}
\begin{proof}
We first proceed the proof of \eqref{7-2-13}.
Based on the expansion formulae of the Pfaffian, we have 
\begin{align}
&\frac{d}{dt}(\tau_{2n}P_{2n}(z;t))\nonumber\\
=&\frac{d}{dt}\left(\sum_{j=0}^{2n}(-1)^j \Pf(0,1,\ldots,\hat{j},\ldots,2n)z^j\right)\nonumber\\
=&\sum_{j=0}^{2n}\left((-z)^j\frac{d}{dt} \Pf(0,1,\ldots,\hat{j},\ldots,2n)\right)\nonumber\\
&\qquad+\sum_{j=0}^{2n}\left((-1)^j\alpha j \Pf(0,1,\ldots,\hat{j},\ldots,2n)z^j\right).\label{pf7}
\end{align}
As for the first expression at the right hand side above, applying the derivative formula~\eqref{pf1}, we have 
\begin{align}
&\sum_{j=0}^{2n}\left((-z)^j\frac{d}{dt} \Pf(0,1,\ldots,\hat{j},\ldots,2n)\right)\nonumber\\
=&\alpha\sum_{j=0}^{2n}(-z)^j(n(2n+1)-j)\Pf(0,1,\ldots,\hat{j},\ldots,2n)\nonumber\\
&+\alpha_1\bigg(\Pf(1,\ldots,2n-1,2n+1)+z^{2n}\Pf(0,1,\ldots,2n-2,2n)\nonumber\\
&\qquad+\sum_{j=1}^{2n-1}(-z)^j\left(\Pf(0,1,\ldots,\widehat{j-1},\ldots,2n)+\Pf(0,1,\ldots,\hat{j},\ldots,2n-1,2n+1)\right)\bigg)\nonumber\\
&+\alpha_2\bigg(\sum_{j=0}^{2n-2}(-z)^j \Pf(0,\ldots,\hat{j},\ldots,2n-2,2n+1,2n)\nonumber\\
&\qquad+\sum_{j=0}^{2n-1}(-z)^j \Pf(0,\ldots,\hat{j},\ldots,2n-1,2n+2)\nonumber\\
&\qquad+\sum_{j=3}^{2n-1}(-z)^j \Pf(0,\ldots,j-3,j,j-1,j+1,\ldots,2n)+z^2 \Pf(2,1,3,\ldots,2n)\nonumber\\
&\qquad+z^{2n}\big(\Pf(0,\ldots,2n-3,2n,2n-1)+\Pf(0,\ldots,2n-2,2n+1)\big)\bigg)\nonumber\\
=&\alpha\sum_{j=0}^{2n}(-z)^j(n(2n+1)-j)\Pf(0,1,\ldots,\hat{j},\ldots,2n)\nonumber\\
&+\alpha_1\bigg(\sum_{j=0}^{2n-1}(-z)^j \Pf(0,\ldots,\hat{j},\ldots,2n-1,2n+1)\nonumber\\
&\qquad+\sum_{j=1}^{2n}(-z)^j\Pf(0,1,\ldots,\widehat{j-1},\ldots,2n)\bigg)\nonumber\\
&+\alpha_2\bigg(\sum_{j=0}^{2n-2}(-z)^j \Pf(0,\ldots,\hat{j},\ldots,2n-2,2n+1,2n)\nonumber\\
&\quad+\sum_{j=0}^{2n-1}(-z)^j \Pf(0,\ldots,\hat{j},\ldots,2n-1,2n+2)-\sum_{j=3}^{2n-1}(-z)^j \Pf(0,\ldots,\widehat{j-2},\ldots,2n)\nonumber\\
&\quad-z^2 \Pf(1,\ldots,2n)
-z^{2n}(\Pf(0,\ldots,\widehat{2n-2},2n-1,2n)-\Pf(0,\ldots,2n-2,2n+1))\bigg).\nonumber
\end{align}
By use of the Pfaffian expansion, we easily see
\begin{align}
&\sum_{j=0}^{2n-1}(-z)^j \Pf(0,\ldots,\hat{j},\ldots,2n-1,2n+1)\nonumber\\
=&\Pf(0,\ldots,2n-1,2n+1,z)-z^{2n+1}\Pf(0,\ldots,2n-1),\nonumber\\
&\sum_{j=1}^{2n}(-z)^j\Pf(0,1,\ldots,\widehat{j-1},\ldots,2n)\nonumber\\
=&-z \Pf(0,1,\ldots,2n,z)+z^{2n+1}\Pf(0,\ldots,2n-1),\nonumber\\
&\sum_{j=2}^{2n}(-z)^j \Pf(0,\ldots,\widehat{j-2},\ldots,2n)\nonumber\\
=&\sum_{j=3}^{2n-1}(-z)^j \Pf(0,\ldots,\widehat{j-2},\ldots,2n)+z^2 \Pf(1,\ldots,2n)+z^{2n}\Pf(0,\ldots,\widehat{2n-2},2n-1,2n).\nonumber
\end{align}
Consequently, we obtain
\begin{align}
&\sum_{j=0}^{2n}(-z)^j\frac{d}{dt} \Pf(0,1,\ldots,\hat{j},\ldots,2n)\nonumber\\
=&\alpha\sum_{j=0}^{2n}(-z)^j(n(2n+1)-j)\Pf(0,1,\ldots,\hat{j},\ldots,2n)\nonumber\\
&+\alpha_1\Big(\Pf(0,\ldots,2n-1,2n+1,z)-z \Pf(0,1,\ldots,2n,z)\Big)\nonumber\\
&+\alpha_2\bigg(-\sum_{j=0}^{2n-2}(-z)^j \Pf(0,\ldots,\hat{j},\ldots,2n-2,2n,2n+1)\nonumber\\
&\qquad+\sum_{j=0}^{2n-1}(-z)^j \Pf(0,\ldots,\hat{j},\ldots,2n-1,2n+2)\nonumber\\
&\qquad-\sum_{j=2}^{2n}(-z)^j \Pf(0,\ldots,\widehat{j-2},\ldots,2n)+z^{2n}\Pf(0,\ldots,2n-2,2n+1)\bigg).\nonumber
\end{align}
Furthermore, it is noted that
\begin{align}
&\Pf(0,1,\ldots,2n-2,2n,2n+1,z)
=\sum_{j=0}^{2n-2}(-z)^j \Pf(0,\ldots,\hat{j},\ldots,2n-2,2n,2n+1)\nonumber\\
&\quad\quad\quad\quad\quad\quad\quad\quad\quad\quad\qquad\qquad\qquad\quad-z^{2n}\Pf(0,\ldots,2n-2,2n+1)\nonumber\\
&\quad\quad\quad\quad\quad\quad\quad\quad\quad\quad\qquad\qquad\qquad\quad+z^{2n+1}\Pf(0,\ldots,2n-2,2n),\nonumber\\
&\Pf(0,1,\ldots,2n-1,2n+2,z)
=\sum_{j=0}^{2n-1}(-z)^j \Pf(0,\ldots,\hat{j},\ldots,2n-1,2n+2)\nonumber\\
&\quad\quad\quad\quad\quad\quad\quad\quad\quad\quad\quad\quad\quad\quad\quad+z^{2n+2}\Pf(0,\ldots,2n-1),\nonumber\\
&z^2 \Pf(0,1,\ldots,2n,z)
=\sum_{j=2}^{2n}(-z)^j \Pf(0,\ldots,\widehat{j-2},\ldots,2n)\nonumber\\
&\quad\quad\quad\quad\quad\quad\quad\quad\quad\qquad-z^{2n+1}\Pf(0,\ldots,2n-2,2n)+z^{2n+2}\Pf(0,\ldots,2n-1),\nonumber
\end{align}
from which we derive for the second expression in \eqref{pf7} that
\begin{align}
&\sum_{j=0}^{2n}(-z)^j\frac{d}{dt} \Pf(0,1,\ldots,\hat{j},\ldots,2n)\nonumber\\
=&\alpha\sum_{j=0}^{2n}(-z)^j(n(2n+1)-j)\Pf(0,1,\ldots,\hat{j},\ldots,2n)\nonumber\\
&+\alpha_1\Big(\Pf(0,\ldots,2n-1,2n+1,z)-z \Pf(0,1,\ldots,2n,z)\Big)\nonumber\\
&+\alpha_2\Big(-\Pf(0,1,\ldots,2n-2,2n,2n+1,z)\nonumber\\
&\qquad+\Pf(0,1,\ldots,2n-1,2n+2,z)-z^2 \Pf(0,1,\ldots,2n,z)\Big).\label{pf8}
\end{align}
Inserting ~\eqref{pf8} into  ~\eqref{pf7} immediately leads to \eqref{7-2-13}.

Similarly, we can establish the validity of ~\eqref{7-2-14}. In fact, we have
\begin{align}
&\frac{d}{dt}(\tau_{2n+1}P_{2n+1}(z;t))\nonumber\\
=&\frac{d}{dt}\bigg(\sum_{j=0}^{2n+1}(-1)^{j+1} \Pf(d_0,0,1,\ldots,\hat{j},\ldots,2n+1)z^j\bigg)\nonumber\\
=&\sum_{j=0}^{2n+1}(-1)^{j+1}z^j\frac{d}{dt} \Pf(d_0,0,\ldots,\hat{j},\ldots,2n+1)\nonumber\\
&\quad+\sum_{j=0}^{2n+1}(-1)^{j+1}\alpha j \Pf(d_0,0,\ldots,\hat{j},\ldots,2n+1)z^j.\label{pf9}
\end{align}
By employing the derivative formula ~\eqref{pf2}, we obtain
\begin{align}
&\sum_{j=0}^{2n+1}(-1)^{j+1}z^j\frac{d}{dt} \Pf(d_0,0,\ldots,\hat{j},\ldots,2n+1)\nonumber\\
=&\alpha\sum_{j=0}^{2n+1}(-1)^{j+1}((n+1)(2n+1)-j)z^j \Pf(d_0,0,\ldots,\hat{j},\ldots,2n+1)\nonumber\\
&+\alpha_1\bigg(\sum_{j=1}^{2n+1}(-1)^{j+1}z^j \Pf(d_0,0,\ldots,\widehat{j-1},\ldots,2n+1)\nonumber\\
&\qquad\quad+\sum_{j=0}^{2n}(-1)^{j+1}z^j \Pf(d_0,0,\ldots,\hat{j},\ldots,2n,2n+2)\bigg)\nonumber\\
&+\alpha_2\bigg(\sum_{j=0}^{2n-1}(-1)^{j+1}z^j \Pf(d_0,0,\ldots,\hat{j},\ldots,2n-1,2n+2,2n+1)\nonumber\\
&\qquad\quad+\sum_{j=0}^{2n}(-1)^{j+1}z^j \Pf(d_0,0,\ldots,\hat{j},\ldots,2n,2n+3)-z^2 \Pf(d_0,2,1,3,\ldots,2n+1)\nonumber\\
&\qquad\quad+\sum_{j=3}^{2n}(-1)^{j+1}z^j \Pf(d_0,0,\ldots,j-3,j,j-1,j+1,\ldots,2n+1)\nonumber\\
&\qquad\quad+z^{2n+1}(\Pf(d_0,0,\ldots,2n-2,2n+1,2n)+\Pf(d_0,0,\ldots,2n-1,2n+2))\bigg).\nonumber
\end{align}
According to the expansion formulae for Pfaffians, we have
\begin{align}
&z \Pf(d_0,0,\ldots,2n+1,z),\nonumber\\
=&-\sum_{j=1}^{2n+1}(-1)^{j+1}z^j \Pf(d_0,0,\ldots,\widehat{j-1},\ldots,2n+1)+z^{2n+2}\Pf(d_0,0,\ldots,2n),\nonumber\\
&\Pf(d_0,0,\ldots,2n,2n+2,z),\nonumber\\
=&\sum_{j=0}^{2n}(-1)^{j+1}z^j \Pf(d_0,0,\ldots,\hat{j},\ldots,2n,2n+2)+z^{2n+2}\Pf(d_0,0,\ldots,2n),\nonumber
\end{align}
using which, we obtain
\begin{align}
&\sum_{j=0}^{2n+1}(-1)^{j+1}z^j\frac{d}{dt} \Pf(d_0,0,\ldots,\hat{j},\ldots,2n+1)\nonumber\\
=&\alpha\sum_{j=0}^{2n+1}(-1)^{j+1}((n+1)(2n+1)-j)z^j \Pf(d_0,0,\ldots,\hat{j},\ldots,2n+1)\nonumber\\
&+\alpha_1\big(-z \Pf(d_0,0,\ldots,2n+1,z)+\Pf(d_0,0,\ldots,2n,2n+2,z)\big)\nonumber\\
&+\alpha_2\bigg(\sum_{j=0}^{2n-1}(-z)^j \Pf(d_0,0,\ldots,\hat{j},\ldots,2n-1,2n+1,2n+2)\nonumber\\
&\quad\quad-\sum_{j=0}^{2n}(-z)^j \Pf(d_0,0,\ldots,\hat{j},\ldots,2n,2n+3)\nonumber\\
&\quad\quad+\sum_{j=2}^{2n+1}(-z)^j \Pf(d_0,0,\ldots,\widehat{j-2},\ldots,2n+1)+z^{2n+1}\Pf(d_0,0,\ldots,2n-1,2n+2)\bigg).\nonumber
\end{align}
Recall that we also have the following expansions 
\begin{align*}
&\Pf(d_0,0,\ldots,2n-1,2n+1,2n+2,z)\nonumber\\
=&-\sum_{j=0}^{2n-1}(-z)^j \Pf(d_0,0,\ldots,\hat{j},\ldots,2n-1,2n+1,2n+2)\nonumber\\
&-z^{2n+1}\Pf(d_0,1,\ldots,2n-1,2n+2)+z^{2n+2}\Pf(d_0,0,\ldots,2n-1,2n+1),\\
&\Pf(d_0,0,\ldots,2n,2n+3,z)\nonumber\\
=&-\sum_{j=0}^{2n}(-z)^j \Pf(d_0,0,\ldots,\hat{j},\ldots,2n,2n+3) +z^{2n+3}\Pf(d_0,0,\ldots,2n),\\
&z^2\Pf(d_0,0,\ldots,2n-1,2n,2n+1,z)\nonumber\\
=&-\sum_{j=2}^{2n+1}(-z)^j \Pf(d_0,0,\ldots,\widehat{j-2},\ldots,2n+1)+z^{2n+3}\Pf(d_0,0,\ldots,2n)\nonumber\\
&-z^{2n+2}\Pf(d_0,0,\ldots,2n-1,2n+1).
\end{align*}
Consequently, we are led to
\begin{align}
&\sum_{j=0}^{2n+1}(-1)^{j+1}z^j\frac{d}{dt} \Pf(d_0,0,\ldots,\hat{j},\ldots,2n+1)\nonumber\\
=&\alpha\sum_{j=0}^{2n+1}(-1)^{j+1}((n+1)(2n+1)-j)z^j \Pf(d_0,0,\ldots,\hat{j},\ldots,2n+1)\nonumber\\
&+\alpha_1\big(-z \Pf(d_0,0,\ldots,2n+1,z)+\Pf(d_0,0,\ldots,2n,2n+2,z)\big)\nonumber\\
&+\alpha_2\big(-\Pf(d_0,0,\ldots,2n-1,2n+1,2n+2,z)\nonumber\\
&\qquad\quad+\Pf(d_0,0,\ldots,2n,2n+3,z)-z^2 \Pf(d_0,0,\ldots,2n+1,z)\big).\label{pf10}
\end{align}
Plugging ~\eqref{pf10} into Eqs.~\eqref{pf9}, we immediately get the desired formula \eqref{7-2-14}.
Therefore, we complete the proof.
\end{proof}

\section{Proof of Theorem \ref{th12}}\label{pth12}
\begin{proof}
We proceed the proof based on the parity.

(I). Let's first consider the even case. It follows from~\eqref{7-2-12} and~\eqref{7-2-13} that
\begin{align}
&\tau_{2n}^2\frac{d}{dt}P_{2n}(z;t)\nonumber\\
=&-\Pf(0,\ldots,2n,z)\Big(n(2n-1)\alpha\tau_{2n}+\alpha_1 \Pf(0,\ldots,2n-2,2n)\nonumber\\
&\qquad\quad+\alpha_2\big(\Pf(0,\ldots,2n-3,2n,2n-1)+\Pf(0,\ldots,2n-2,2n+1)\big)\Big)\nonumber\\
&+\tau_{2n}\Big(n(2n+1)\alpha \Pf(0,\ldots,2n,z)+\alpha_1(\Pf(0,\ldots,2n-1,2n+1,z)\nonumber\\
&\qquad\quad-z \Pf(0,\ldots,2n,z))+\alpha_2\big(-\Pf(0,\ldots,2n-2,2n,2n+1,z)\nonumber\\
&\qquad\quad+\Pf(0,\ldots,2n-1,2n+2,z)-z^2 \Pf(0,\ldots,2n,z)\big)\Big).\label{der_p2n}
\end{align}
Regarding the terms with $\alpha_1$, we have
\begin{align}
&-\Pf(0,\ldots,2n,z)\Pf(0,\ldots,2n-2,2n)\nonumber\\
&+\tau_{2n}\big(\Pf(0,\ldots,2n-1,2n+1,z)-z \Pf(0,\ldots,2n,z)\big)\nonumber\\
=&-\tau_{2n}P_{2n}\frac{\partial}{\partial t_1}\tau_{2n}+\tau_{2n}\frac{\partial}{\partial t_1}(\tau_{2n}P_{2n})=\tau_{2n}^2\frac{\partial}{\partial t_1}P_{2n}
\label{pf13}
\end{align}
from ~\eqref{7-2-17} and~\eqref{7-2-21}. In the following, we deal with the terms with $\alpha_2$. We begin with collecting some useful formulae.
By employing the Pfaffian identities in \eqref{iden1} and \eqref{iden2}, it is easy to see that the following equalities hold
\begin{align}
&\Pf(d_0,d_1,0,\ldots,2n,z)\tau_{2n}\nonumber\\
=&\Pf(d_0,d_1,0,\ldots,2n-1)\Pf(0,\ldots,2n,z)\nonumber\\
&-\Pf(d_0,0,\ldots,2n)\Pf(d_1,0,\ldots,2n-1,z)\nonumber\\
&+\Pf(d_0,0,\ldots,2n-1,z)\Pf(d_1,0,\ldots,2n),\label{pf14}\\
&-\Pf(0,\ldots,2n-2,2n+1)\Pf(0,\ldots,2n,z)\nonumber\\
=&\Pf(2n-1,2n,z,0,\ldots,2n-2)\Pf(0,\ldots,2n-2,2n+1)\nonumber\\
&-\Pf(0,\ldots,2n-2,2n)\Pf(0,\ldots,2n-2,2n-1,2n+1,z)\nonumber\\
&-\Pf(0,\ldots,2n-2,z)\Pf(0,\ldots,2n-2,2n-1,2n,2n+1),\label{pf15}\\
&-\Pf(0,\ldots,2n-2,2n,2n+1,z)\Pf(0,\ldots,2n-1)\nonumber\\
=&\Pf(0,\ldots,2n-2,2n-1)\Pf(0,\ldots,2n-2,2n,2n+1,z)\nonumber\\
&-\Pf(0,\ldots,2n-2,2n)\Pf(0,\ldots,2n-2,2n-1,2n+1,z)\nonumber\\
&-\Pf(0,\ldots,2n-2,z)\Pf(0,\ldots,2n-2,2n-1,2n,2n+1).\label{pf16}
\end{align}
With the help of \eqref{pf15}, \eqref{pf16}, \eqref{7-2-9} and \eqref{7-2-19}, 
we obtain
\begin{align}
&-\Pf(0,\ldots,2n,z)(\Pf(0,\ldots,2n-3,2n,2n-1)+\Pf(0,\ldots,2n-2,2n+1))\nonumber\\
&+\tau_{2n}(-\Pf(0,\ldots,2n-2,2n,2n+1,z)+\Pf(0,\ldots,2n-1,2n+2,z)-z^2 \Pf(0,\ldots,2n,z))\nonumber\\
=&\Pf(0,\ldots,2n,z)\frac{\partial^2}{\partial t_1^2}\tau_{2n}\nonumber\\
&+\tau_{2n}(\Pf(0,\ldots,2n-1,2n+2,z)-z^2 \Pf(0,\ldots,2n,z))\nonumber\\
&+2 \Pf(0,\ldots,2n-1)\Pf(0,\ldots,2n-2,2n,2n+1,z)\nonumber\\
&+\Pf(0,\ldots,2n-2,2n+1)\Pf(0,\ldots,2n,z)\nonumber\\
&-3\Pf(0,\ldots,2n-2,2n)\Pf(0,\ldots,2n-1,2n+1,z)\nonumber\\
&-3 \Pf(0,\ldots,2n-2,z)\Pf(0,\ldots,2n+1).\label{pf18}
\end{align}
We can also get from \eqref{7-2-9} and the Pfaffian identity \eqref{pf14}
\begin{align*}
&-z^2 \Pf(0,\ldots,2n,z)\tau_{2n}\nonumber\\
=&z(-\Pf(d_0,d_1,0,\ldots,2n,z)-\Pf(0,\ldots,2n-1,2n+1,z))\tau_{2n}\nonumber\\
=&-z \Pf(0,\ldots,2n-1,2n+1,z)\tau_{2n}-z(\Pf(d_0,d_1,0,\ldots,2n-1)\Pf(0,\ldots,2n,z)\nonumber\\
&-\Pf(d_0,0,\ldots,2n)\Pf(d_1,0,\ldots,2n-1,z)
+\Pf(d_0,0,\ldots,2n-1,z)\Pf(d_1,0,\ldots,2n)),
\end{align*}
by substituting which into \eqref{pf18}, we have
\begin{align*}
&-\Pf(0,\ldots,2n,z)(\Pf(0,\ldots,2n-3,2n,2n-1)+\Pf(0,\ldots,2n-2,2n+1))\nonumber\\
&+\tau_{2n}(-\Pf(0,\ldots,2n-2,2n,2n+1,z)+\Pf(0,\ldots,2n-1,2n+2,z)-z^2 \Pf(0,\ldots,2n,z))\nonumber\\
=&\Pf(0,\ldots,2n,z)\frac{\partial^2}{\partial t_1^2}\tau_{2n}\nonumber\\
&+\tau_{2n}(\Pf(0,\ldots,2n-1,2n+2,z)-z \Pf(0,\ldots,2n-1,2n+1,z))\nonumber\\
&+2 \Pf(0,\ldots,2n-1)\Pf(0,\ldots,2n-2,2n,2n+1,z)\nonumber\\
&+\Pf(0,\ldots,2n-2,2n+1)\Pf(0,\ldots,2n,z)\nonumber\\
&-3\Pf(0,\ldots,2n-2,2n)\Pf(0,\ldots,2n-1,2n+1,z)\nonumber\\
&-3 \Pf(0,\ldots,2n-2,z)\Pf(0,\ldots,2n+1)\nonumber\\
&-z(\Pf(d_0,d_1,0,\ldots,2n-1)\Pf(0,\ldots,2n,z)\nonumber\\
&+\Pf(d_0,0,\ldots,2n)\Pf(d_1,0,\ldots,2n-1,z)
-\Pf(d_0,0,\ldots,2n-1,z)\Pf(d_1,0,\ldots,2n))\nonumber\\
=&\Pf(0,\ldots,2n,z)\frac{\partial^2}{\partial t_1^2}\tau_{2n}\nonumber\\
&+(\Pf(0,\ldots,2n-1,2n+2,z)-z \Pf(0,\ldots,2n-1,2n+1,z))\Pf(0,\ldots,2n-1)\nonumber\\
&+2 \Pf(0,\ldots,2n-1)\Pf(0,\ldots,2n-2,2n,2n+1,z)\nonumber\\
&+\Pf(0,\ldots,2n-2,2n+1)\Pf(0,\ldots,2n,z)\nonumber\\
&-2 \Pf(0,\ldots,2n-2,2n)\Pf(0,\ldots,2n-1,2n+1,z)\nonumber\\
&-3 \Pf(0,\ldots,2n-2,z)\Pf(0,\ldots,2n+1)+z \Pf(d_0,0,\ldots,2n)\Pf(d_1,0,\ldots,2n-1,z)\nonumber\\
&-\Pf(d_0,d_1,0,\ldots,2n-1)\Pf(d_0,d_1,0,\ldots,2n,z)\nonumber\\
&-\Pf(d_1,0,\ldots,2n)(\Pf(d_0,0,\ldots,2n-2,2n,z)-\Pf(d_1,0,\ldots,2n-1,z)),
\end{align*}
where \eqref{7-2-9} and \eqref{7-2-11} are used. If we make use of the Pfaffian identity
\begin{align*}
&\Pf(0,\ldots,2n-2,2n)\Pf(0,\ldots,2n-1,2n+1,z)\nonumber\\
=&\Pf(0,\ldots,2n-1)\Pf(0,\ldots,2n-2,2n,2n+1,z)\\
&+\Pf(0,\ldots,2n-2,2n+1)\Pf(0,\ldots,2n,z)\nonumber\\
&-\Pf(0,\ldots,2n-2,z)\Pf(0,\ldots,2n+1),
\end{align*}
then we can arrive at
\begin{align}
&-\Pf(0,\ldots,2n,z)(\Pf(0,\ldots,2n-3,2n,2n-1)+\Pf(0,\ldots,2n-2,2n+1))\nonumber\\
&+\tau_{2n}(-\Pf(0,\ldots,2n-2,2n,2n+1,z)+\Pf(0,\ldots,2n-1,2n+2,z)-z^2 \Pf(0,\ldots,2n,z))\nonumber\\
=&\Pf(0,\ldots,2n,z)\frac{\partial^2}{\partial t_1^2}\tau_{2n}\nonumber\\
&+(\Pf(0,\ldots,2n-1,2n+2,z)-z \Pf(0,\ldots,2n-1,2n+1,z))\Pf(0,\ldots,2n-1)\nonumber\\
&-\Pf(0,\ldots,2n-2,2n+1)\Pf(0,\ldots,2n,z)-\Pf(0,\ldots,2n-2,z)\Pf(0,\ldots,2n+1)\nonumber\\
&+z \Pf(d_0,0,\ldots,2n)\Pf(d_1,0,\ldots,2n-1,z)\nonumber\\
&-\Pf(d_0,d_1,0,\ldots,2n-1)\Pf(d_0,d_1,0,\ldots,2n,z)\nonumber\\
&-\Pf(d_1,0,\ldots,2n)(\Pf(d_0,0,\ldots,2n-2,2n,z)-\Pf(d_1,0,\ldots,2n-1,z))\nonumber\\
=&\tau_{2n}P_{2n}\frac{\partial^2}{\partial t_1^2}\tau_{2n}+\Pf(0,\ldots,2n-1,2n+2,z)\Pf(0,\ldots,2n-1)\nonumber\\
&-z\tau_{2n}(\frac{\partial}{\partial t_1}(\tau_{2n}P_{2n})+z\tau_{2n}P_{2n})
-\Pf(0,\ldots,2n-2,2n+1)\Pf(0,\ldots,2n,z)\nonumber\\
&-\tau_{2n-2}P_{2n-2}\tau_{2n+2}+z\tau_{2n+1}\frac{\partial}{\partial t_1}(\tau_{2n-1}P_{2n-1})\nonumber\\
&+\frac{\partial}{\partial t_1}\tau_{2n}\frac{\partial}{\partial t_1}(\tau_{2n}P_{2n})-z\tau_{2n-1}P_{2n-1}\frac{\partial}{\partial t_1}\tau_{2n+1}.\label{pf22}
\end{align}

Recall that \eqref{7-2-23} gives
\begin{align}
&\left(\Pf(0,\ldots,2n-1,2n+2,z)+\Pf(0,\ldots,2n-2,2n,2n+1,z)\right)\tau_{2n}\nonumber\\
=&\bigg(\frac{\partial^2}{\partial t_1^2}(\tau_{2n}P_{2n})+2z \frac{\partial}{\partial t_1}(\tau_{2n}P_{2n})+z^2\tau_{2n}P_{2n}\bigg)\tau_{2n}.\label{pf23}
\end{align}
If we plug the Pfaffian identities
\begin{align}
&-\Pf(0,\ldots,2n-2,2n+1)\Pf(0,\ldots,2n,z)\nonumber\\
=&\Pf(0,\ldots,2n-1)\Pf(0,\ldots,2n-2,2n,2n+1,z)\nonumber\\
&-\Pf(0,\ldots,2n-2,2n)\Pf(0,\ldots,2n-1,2n+1,z)\nonumber\\
&-\Pf(0,\ldots,2n-2,z)\Pf(0,\ldots,2n+1)\nonumber\\
=&\tau_{2n}\Pf(0,\ldots,2n-2,2n,2n+1,z)-\frac{\partial}{\partial t_1}\tau_{2n}(\frac{\partial}{\partial t_1}(\tau_{2n}P_{2n})\nonumber\\
&+z\tau_{2n}P_{2n})-\tau_{2n-2}P_{2n-2}\tau_{2n+2}\nonumber
\end{align}
and \eqref{pf23} to the right hand side of \eqref{pf22}, we get
\begin{align}
&-\Pf(0,\ldots,2n,z)(\Pf(0,\ldots,2n-3,2n,2n-1)+\Pf(0,\ldots,2n-2,2n+1))\nonumber\\
&+\tau_{2n}(-\Pf(0,\ldots,2n-2,2n,2n+1,z)+\Pf(0,\ldots,2n-1,2n+2,z)-z^2 \Pf(0,\ldots,2n,z))\nonumber\\
=&\tau_{2n}P_{2n}\frac{\partial^2}{\partial t_1^2}\tau_{2n}+z\tau_{2n}\frac{\partial}{\partial t_1}(\tau_{2n}P_{2n})-2\tau_{2n-2}\tau_{2n+2}P_{2n-2}\nonumber\\
&+z\tau_{2n+1}\frac{\partial}{\partial t_1}(\tau_{2n-1}P_{2n-1})-z\tau_{2n-1}P_{2n-1}\frac{\partial}{\partial t_1}\tau_{2n+1}-2\frac{\partial}{\partial t_1}\tau_{2n}\frac{\partial}{\partial t_1}(\tau_{2n}P_{2n})\nonumber\\
&-z\tau_{2n}P_{2n}\frac{\partial}{\partial t_1}\tau_{2n}+\tau_{2n}\frac{\partial^2}{\partial t_1^2}(\tau_{2n}P_{2n})\nonumber\\
=&2\tau_{2n}P_{2n}\frac{\partial^2}{\partial t_1^2}\tau_{2n}+z\tau^2_{2n}\frac{\partial}{\partial t_1}P_{2n}-2\tau_{2n-2}\tau_{2n+2}P_{2n-2}+z\tau_{2n-1}\tau_{2n+1}\frac{\partial}{\partial t_1}P_{2n-1}\nonumber\\
&+z\tau_{2n+1}P_{2n-1}\frac{\partial}{\partial t_1}\tau_{2n-1}-z\tau_{2n-1}P_{2n-1}\frac{\partial}{\partial t_1}\tau_{2n+1}-2P_{2n}(\frac{\partial}{\partial t_1}\tau_{2n})^2+\tau^2_{2n}\frac{\partial^2}{\partial t_1^2}P_{2n}.\label{pf24}
\end{align}
According to the formula in \cite[eq.(3.38)]{chang2018partial}, it can be inferred that
\begin{align*}
\tau^2_{2n}\frac{\partial}{\partial t_1}P_{2n}=-\left(\tau_{2n-1}\frac{\partial}{\partial t_1}\tau_{2n+1}-\tau_{2n+1}\frac{\partial}{\partial t_1}\tau_{2n-1}\right)P_{2n-1}+\tau_{2n-1}\tau_{2n+1}\frac{\partial}{\partial t_1}P_{2n-1},
\end{align*}
by plugging which into \eqref{pf24}, we can arrive at
\begin{align}
&-\Pf(0,\ldots,2n,z)(\Pf(0,\ldots,2n-3,2n,2n-1)+\Pf(0,\ldots,2n-2,2n+1))\nonumber\\
&+\tau_{2n}(-\Pf(0,\ldots,2n-2,2n,2n+1,z)+\Pf(0,\ldots,2n-1,2n+2,z)-z^2 \Pf(0,\ldots,2n,z))\nonumber\\
=&2\tau_{2n}P_{2n}\frac{\partial^2}{\partial t_1^2}\tau_{2n}-2\tau_{2n-2}\tau_{2n+2}P_{2n-2}+2z\tau_{2n-1}\tau_{2n+1}\frac{\partial}{\partial t_1}P_{2n-1}+2z\tau_{2n+1}P_{2n-1}\frac{\partial}{\partial t_1}\tau_{2n-1}\nonumber\\
&-2z\tau_{2n-1}P_{2n-1}\frac{\partial}{\partial t_1}\tau_{2n+1}-2P_{2n}(\frac{\partial}{\partial t_1}\tau_{2n})^2+\tau^2_{2n}\frac{\partial^2}{\partial t_1^2}P_{2n}.\label{pf26}
\end{align}
Finally, based on \eqref{der_p2n}, \eqref{pf13} and \eqref{pf26}, we deduce that
\begin{align}
&\tau_{2n}^2\frac{d}{dt}P_{2n}(z;t)\nonumber\\
=&2n\alpha\tau_{2n}^2P_{2n}+\alpha_1\tau_{2n}^2\frac{\partial}{\partial t_1}P_{2n}+\alpha_2\bigg(2\tau_{2n}P_{2n}\frac{\partial^2}{\partial t_1^2}\tau_{2n}-2\tau_{2n-2}\tau_{2n+2}P_{2n-2}\nonumber\\
&\qquad\qquad+2z\tau_{2n-1}\tau_{2n+1}\frac{\partial}{\partial t_1}P_{2n-1}+2z\tau_{2n+1}P_{2n-1}\frac{\partial}{\partial t_1}\tau_{2n-1}\nonumber\\
&\qquad\qquad-2z\tau_{2n-1}P_{2n-1}\frac{\partial}{\partial t_1}\tau_{2n+1}-2P_{2n}\left(\frac{\partial}{\partial t_1}\tau_{2n}\right)^2+\tau^2_{2n}\frac{\partial^2}{\partial t_1^2}P_{2n}\bigg).\label{pf27}
\end{align}

(II). Now we focus on some formulae on $\tau_{2n+1}^2\frac{d}{dt}P_{2n+1}(z;t)$.
We first collect some useful formulae.
Combining equations \eqref{7-2-12} and \eqref{7-2-14}, we obtain
\begin{align}
&\tau_{2n+1}^2\frac{d}{dt}P_{2n+1}(z;t)\nonumber\\
=&(2n+1)\alpha\tau_{2n+1}^2P_{2n+1}+\alpha_1\Big(-\Pf(d_0,0,\ldots,2n-1,2n+1)\Pf(d_0,0,\ldots,2n+1,z)\nonumber\\
&\qquad\qquad+\big(-z \Pf(d_0,0,\ldots,2n+1,z)+\Pf(d_0,0,\ldots,2n,2n+2,z)\big)\tau_{2n+1}\Big)\nonumber\\
&+\alpha_2\Big(\big(\Pf(d_0,0,\ldots,2n-2,2n,2n+1)-\Pf(d_0,0,\ldots,2n-1,2n+2)\big)\Pf(d_0,0,\ldots,2n+1,z)\nonumber\\
&\qquad+\big(-\Pf(d_0,0,\ldots,2n-1,2n+1,2n+2,z)+\Pf(d_0,0,\ldots,2n,2n+3,z)\nonumber\\
&\qquad\qquad-z^2\Pf(d_0,0,\ldots,2n+1,z)\big)\Pf(d_0,0,\ldots,2n)\Big). \label{der_p2n+1}
\end{align}
By employing equations \eqref{7-2-18} and \eqref{7-2-22}, we have
\begin{align}
&-\Pf(d_0,0,\ldots,2n-1,2n+1)\Pf(d_0,0,\ldots,2n+1,z)\nonumber\\
&+(-z \Pf(d_0,0,\ldots,2n+1,z)+\Pf(d_0,0,\ldots,2n,2n+2,z))\tau_{2n+1}\nonumber\\
=&-\tau_{2n+1}P_{2n+1}\frac{\partial}{\partial t_1}\tau_{2n+1}+\tau_{2n+1}\frac{\partial}{\partial t_1}(\tau_{2n+1}P_{2n+1})\nonumber\\
=&\tau_{2n+1}^2\frac{\partial}{\partial t_1}P_{2n+1}.\label{pf29}
\end{align}
Furthermore, it easily follows from the Pfaffian identities \eqref{iden1} and \eqref{iden2} that
\begin{align}
&-\Pf(d_0,0,\ldots,2n-1,2n+2)\Pf(d_0,0,\ldots,2n+1,z)\nonumber\\
=&\Pf(d_0,0,\ldots,2n)\Pf(d_0,0,\ldots,2n-1,2n+1,2n+2,z)\nonumber\\
&-\Pf(d_0,0,\ldots,2n-1,2n+1)\Pf(d_0,0,\ldots,2n,2n+2,z)\nonumber\\
&-\Pf(d_0,0,\ldots,2n-1,z)\Pf(d_0,0,\ldots,2n+2),\label{pf30}
\end{align}
\begin{align}
&-\Pf(d_0,0,\ldots,2n-1,2n+1,2n+2,z)\Pf(d_0,0,\ldots,2n)\nonumber\\
=&\Pf(d_0,0,\ldots,2n+1,z)\Pf(d_0,0,\ldots,2n-1,2n+2)\nonumber\\
&-\Pf(d_0,0,\ldots,2n+1,2n+1)\Pf(d_0,0,\ldots,2n,2n+2,z)\nonumber\\
&-\Pf(d_0,0,\ldots,2n-1,z)\Pf(d_0,0,\ldots,2n+2),\label{pf31}
\end{align}
and
\begin{align}
&\Pf(d_1,0,\ldots,2n+1,z)\Pf(d_0,0,\ldots,2n)\nonumber\\
=&\Pf(d_1,0,\ldots,2n)\Pf(d_0,0,\ldots,2n+1,z)\nonumber\\
&-\Pf(0,\ldots,2n+1)\Pf(d_0,d_1,0,\ldots,2n,z)\nonumber\\
&+\Pf(0,\ldots,2n,z)\Pf(d_0,d_1,0,\ldots,2n+1).\label{pf32}
\end{align}

By virtue of \eqref{pf30} and \eqref{pf31}, we can derive
\begin{align}
&(\Pf(d_0,0,\ldots,2n-2,2n,2n+1)-\Pf(d_0,0,\ldots,2n-1,2n+2))\Pf(d_0,0,\ldots,2n+1,z)\nonumber\\
&+(-\Pf(d_0,0,\ldots,2n-1,2n+1,2n+2,z)+\Pf(d_0,0,\ldots,2n,2n+3,z)\nonumber\\
&-z^2\Pf(d_0,0,\ldots,2n+1,z))\Pf(d_0,0,\ldots,2n)\nonumber\\
=&\Pf(d_0,0,\ldots,2n+1,z)\frac{\partial^2}{\partial t_1^2}\tau_{2n+1}+2\Pf(d_0,0,\ldots,2n)\Pf(d_0,0,\ldots,2n-1,2n+1,2n+2,z)\nonumber\\
&+\Pf(d_0,0,\ldots,2n+1,z)\Pf(d_0,0,\ldots,2n-1,2n+2)\nonumber\\
&-3\Pf(d_0,0,\ldots,2n-1,2n+1)\Pf(d_0,0,\ldots,2n,2n+2,z)\nonumber\\
&-3\Pf(d_0,0,\ldots,2n-1,z)\Pf(d_0,0,\ldots,2n+2)\nonumber\\
&+\Pf(d_0,0,\ldots,2n,2n+3,z)\Pf(d_0,0,\ldots,2n)\nonumber\\
&+z(\Pf(d_1,0,\ldots,2n+1,z)-\Pf(d_0,0,\ldots,2n,2n+2,z))\Pf(d_0,0,\ldots,2n).\label{pf33}
\end{align}
where \eqref{7-2-11} is also employed. Substituing \eqref{pf32} into the right-hand side of \eqref{pf33}, we obtain
\begin{align*}
&(\Pf(d_0,0,\ldots,2n-2,2n,2n+1)-\Pf(d_0,0,\ldots,2n-1,2n+2))\Pf(d_0,0,\ldots,2n+1,z)\nonumber\\
&+(-\Pf(d_0,0,\ldots,2n-1,2n+1,2n+2,z)+\Pf(d_0,0,\ldots,2n,2n+3,z)\nonumber\\
&-z^2\Pf(d_0,0,\ldots,2n+1,z))\Pf(d_0,0,\ldots,2n)\nonumber\\
=&\Pf(d_0,0,\ldots,2n+1,z)\frac{\partial^2}{\partial t_1^2}\tau_{2n+1}\nonumber\\
&+2\Pf(d_0,0,\ldots,2n)\Pf(d_0,0,\ldots,2n-1,2n+1,2n+2,z)\nonumber\\
&+\Pf(d_0,0,\ldots,2n+1,z)\Pf(d_0,0,\ldots,2n-1,2n+2)\nonumber\\
&-2\Pf(d_0,0,\ldots,2n-1,2n+1)\Pf(d_0,0,\ldots,2n,2n+2,z)\nonumber\\
&-3\Pf(d_0,0,\ldots,2n-1,z)\Pf(d_0,0,\ldots,2n+2)\nonumber\\
&+\Pf(d_0,0,\ldots,2n,2n+3,z)\Pf(d_0,0,\ldots,2n)\nonumber\\
&-z \Pf(d_0,0,\ldots,2n,2n+2,z)\Pf(d_0,0,\ldots,2n)\nonumber\\
&-z \Pf(0,\ldots,2n+1)\Pf(d_0,d_1,0,\ldots,2n,z)\nonumber\\
&-\Pf(d_1,0,\ldots,2n)\Pf(d_1,0,\ldots,2n+1,z)\nonumber\\
&+\Pf(d_0,d_1,0,\ldots,2n+1)(\Pf(0,\ldots,2n-1,2n+1,z)+\Pf(d_0,d_1,0,\ldots,2n,z)),
\end{align*}
from which we have
\begin{align*}
&(\Pf(d_0,0,\ldots,2n-2,2n,2n+1)-\Pf(d_0,0,\ldots,2n-1,2n+2))\Pf(d_0,0,\ldots,2n+1,z)\nonumber\\
&+(-\Pf(d_0,0,\ldots,2n-1,2n+1,2n+2,z)+\Pf(d_0,0,\ldots,2n,2n+3,z)\nonumber\\
&-z^2\Pf(d_0,0,\ldots,2n+1,z))\Pf(d_0,0,\ldots,2n)\nonumber\\
=&\Pf(d_0,0,\ldots,2n+1,z)\frac{\partial^2}{\partial t_1^2}\tau_{2n+1}\nonumber\\
&-\Pf(d_0,0,\ldots,2n+1,z)\Pf(d_0,0,\ldots,2n-1,2n+2)\nonumber\\
&-\Pf(d_0,0,\ldots,2n-1,z)\Pf(d_0,0,\ldots,2n+2)\nonumber\\
&+\Pf(d_0,0,\ldots,2n,2n+3,z)\Pf(d_0,0,\ldots,2n)\nonumber\\
&-z \Pf(d_0,0,\ldots,2n,2n+2,z)\Pf(d_0,0,\ldots,2n)\nonumber\\
&-z \Pf(0,\ldots,2n+1)\Pf(d_0,d_1,0,\ldots,2n,z)\nonumber\\
&-\Pf(d_1,0,\ldots,2n)\Pf(d_1,0,\ldots,2n+1,z)\nonumber\\
&+\Pf(d_0,d_1,0,\ldots,2n+1)(\Pf(0,\ldots,2n-1,2n+1,z)\nonumber\\
&+\Pf(d_0,d_1,0,\ldots,2n,z))
\end{align*}
with the help of the Pfaffian identity
\begin{align*}
&\Pf(d_0,0,\ldots,2n-1,2n+1)\Pf(d_0,0,\ldots,2n,2n+2,z)\nonumber\\
=&\Pf(d_0,0,\ldots,2n)\Pf(d_0,0,\ldots,2n-1,2n+1,2n+2,z)\nonumber\\
&+\Pf(d_0,0,\ldots,2n+1,z)\Pf(d_0,0,\ldots,2n-1,2n+2)\nonumber\\
&-\Pf(d_0,0,\ldots,2n-1,z)\Pf(d_0,0,\ldots,2n+2).
\end{align*}
Moreover, by observing that
\begin{align*}
&-\Pf(d_0,0,\ldots,2n+1,z)\Pf(d_0,0,\ldots,2n-1,2n+2)\nonumber\\
=&\Pf(d_0,0,\ldots,2n)\Pf(d_0,0,\ldots,2n-1,2n+1,2n+2,z)\nonumber\\
&-\Pf(d_0,0,\ldots,2n-1,2n+1)\Pf(d_0,0,\ldots,2n,2n+2,z)\nonumber\\
&-\Pf(d_0,0,\ldots,2n-1,z)\Pf(d_0,0,\ldots,2n+2)\nonumber\\
=&-\tau_{2n-1}P_{2n-1}\tau_{2n+3}-\left(\frac{\partial}{\partial t_1}(\tau_{2n+1}P_{2n+1})+z\tau_{2n+1}P_{2n+1}\right)\frac{\partial}{\partial t_1}\tau_{2n+1}\nonumber\\
&+\Pf(d_0,0,\ldots,2n-1,2n+1,2n+2,z)\tau_{2n+1},
\end{align*}
and
\begin{align*}
&(\Pf(d_0,0,\ldots,2n,2n+3,z)+\Pf(d_0,0,\ldots,2n-1,2n+1,2n+2,z))\tau_{2n+1}\nonumber\\
=&\left(\frac{\partial^2}{\partial t_1^2}(\tau_{2n+1}P_{2n+1})+2z\frac{\partial}{\partial t_1}(\tau_{2n+1}P_{2n+1})+z^2\tau_{2n+1}P_{2n+1}\right)\tau_{2n+1},
\end{align*}
we can arrive at
\begin{align*}
&(\Pf(d_0,0,\ldots,2n-2,2n,2n+1)-\Pf(d_0,0,\ldots,2n-1,2n+2))\Pf(d_0,0,\ldots,2n+1,z)\nonumber\\
&+(-\Pf(d_0,0,\ldots,2n-1,2n+1,2n+2,z)+\Pf(d_0,0,\ldots,2n,2n+3,z)\nonumber\\
&-z^2\Pf(d_0,0,\ldots,2n+1,z))\Pf(d_0,0,\ldots,2n)\nonumber\\
=&\tau_{2n+1}P_{2n+1}\frac{\partial^2}{\partial t_1^2}\tau_{2n+1}-2\tau_{2n-1}\tau_{2n+3}P_{2n-1}+z\tau_{2n+2}\frac{\partial}{\partial t_1}(\tau_{2n}P_{2n})\nonumber\\
&-2\frac{\partial}{\partial t_1}\tau_{2n+1}\frac{\partial}{\partial t_1}(\tau_{2n+1}P_{2n+1})-z\tau_{2n}P_{2n}\frac{\partial}{\partial t_1}\tau_{2n+2}-z\tau_{2n+1}P_{2n+1}\frac{\partial}{\partial t_1}\tau_{2n+1}\nonumber\\
&+\tau_{2n+1}\left(\frac{\partial^2}{\partial t_1^2}(\tau_{2n+1}P_{2n+1})+z\frac{\partial}{\partial t_1}(\tau_{2n+1}P_{2n+1})\right)\nonumber\\
=&2\tau_{2n+1}P_{2n+1}\frac{\partial^2}{\partial t_1^2}\tau_{2n+1}-2\tau_{2n-1}\tau_{2n+3}P_{2n-1}+z\tau_{2n+2}\frac{\partial}{\partial t_1}(\tau_{2n}P_{2n})\nonumber\\
&-2P_{2n+1}(\frac{\partial}{\partial t_1}\tau_{2n+1})^2-z\tau_{2n}P_{2n}\frac{\partial}{\partial t_1}\tau_{2n+2}+\tau_{2n+1}^2\frac{\partial^2}{\partial t_1^2}P_{2n+1}+z\tau_{2n+1}^2\frac{\partial}{\partial t_1}P_{2n+1}.
\end{align*}
Consequently, we have
\begin{align}
&(\Pf(d_0,0,\ldots,2n-2,2n,2n+1)-\Pf(d_0,0,\ldots,2n-1,2n+2))\Pf(d_0,0,\ldots,2n+1,z)\nonumber\\
&+(-\Pf(d_0,0,\ldots,2n-1,2n+1,2n+2,z)+\Pf(d_0,0,\ldots,2n,2n+3,z)\nonumber\\
&-z^2\Pf(d_0,0,\ldots,2n+1,z))\Pf(d_0,0,\ldots,2n)\nonumber\\
=&2\tau_{2n+1}P_{2n+1}\frac{\partial^2}{\partial t_1^2}\tau_{2n+1}-2\tau_{2n-1}\tau_{2n+3}P_{2n-1}+2z\tau_{2n+2}\frac{\partial}{\partial t_1}(\tau_{2n}P_{2n})\nonumber\\
&-2P_{2n+1}(\frac{\partial}{\partial t_1}\tau_{2n+1})^2-2z\tau_{2n}P_{2n}\frac{\partial}{\partial t_1}\tau_{2n+2}+\tau_{2n+1}^2\frac{\partial^2}{\partial t_1^2}P_{2n+1},\label{pf41}
\end{align}
by applying the similar relationship  in \cite[eq.(3.38)]{chang2018partial}
\begin{align*}
\tau^2_{2n+1}\frac{\partial}{\partial t_1}P_{2n+1}=-\left(\tau_{2n}\frac{\partial}{\partial t_1}\tau_{2n+2}-\tau_{2n+2}\frac{\partial}{\partial t_1}\tau_{2n}\right)P_{2n}+\tau_{2n}\tau_{2n+2}\frac{\partial}{\partial t_1}P_{2n}.
\end{align*}
Finally, we combine \eqref{der_p2n+1}, \eqref{pf29} and \eqref{pf41} to get
\begin{align}
&\tau_{2n+1}^2\frac{d}{dt}P_{2n+1}(z;t)\nonumber\\
=&(2n+1)\alpha\tau_{2n+1}^2P_{2n+1}+\alpha_1\tau_{2n+1}^2\frac{\partial}{\partial t_1}P_{2n+1}+\alpha_2\bigg(2\tau_{2n+1}P_{2n+1}\frac{\partial^2}{\partial t_1^2}\tau_{2n+1}\nonumber\\
&-2\tau_{2n-1}\tau_{2n+3}P_{2n-1}+2z\tau_{2n+2}\frac{\partial}{\partial t_1}(\tau_{2n}P_{2n})-2P_{2n+1}(\frac{\partial}{\partial t_1}\tau_{2n+1})^2\nonumber\\
&-2z\tau_{2n}P_{2n}\frac{\partial}{\partial t_1}\tau_{2n+2}+\tau_{2n+1}^2\frac{\partial^2}{\partial t_1^2}P_{2n+1}\bigg).\label{pf42}
\end{align}

(III). The odd case \eqref{pf42} and  the even case \eqref{pf27} can be written together as
\begin{align*}
&\tau_{n}^2\frac{d}{dt}P_{n}(z;t)\nonumber\\
=&n\alpha\tau_{n}^2P_{n}+\alpha_1\tau_{n}^2\frac{\partial}{\partial t_1}P_{n}+\alpha_2\left(2\tau_{n}P_{n}\frac{\partial^2}{\partial t_1^2}\tau_{n}-2\tau_{n-2}\tau_{n+2}P_{n-2}\right.\nonumber\\
&\qquad \left.+2z\tau_{n+1}\frac{\partial}{\partial t_1}(\tau_{n-1}P_{n-1})-2P_{n}(\frac{\partial}{\partial t_1}\tau_{n})^2-2z\tau_{n-1}P_{n-1}\frac{\partial}{\partial t_1}\tau_{n+1}+\tau_{n}^2\frac{\partial^2}{\partial t_1^2}P_{n}\right).
\end{align*}
Consequently, we have
\begin{align}
&\frac{d}{dt}P_{n}(z;t)\nonumber\\
=&n\alpha P_{n}+\alpha_1\frac{\partial}{\partial t_1}P_{n}+\alpha_2\left(2P_{n}\frac{\frac{\partial^2}{\partial t_1^2}\tau_{n}}{\tau_{n}}-2\frac{\tau_{n-2}\tau_{n+2}}{\tau_{n}^2}P_{n-2}+2z\frac{\tau_{n+1}\tau_{n-1}}{\tau_{n}^2}\frac{\partial}{\partial t_1}P_{n-1}\right.\nonumber\\
&\left.+2z\frac{\tau_{n+1}\tau_{n-1}}{\tau_{n}^2}P_{n-1}\frac{\frac{\partial}{\partial t_1}\tau_{n-1}}{\tau_{n-1}}-2P_{n}\left(\frac{\frac{\partial}{\partial t_1}\tau_{n}}{\tau_{n}}\right)^2-2z\frac{\tau_{n+1}\tau_{n-1}}{\tau_{n}^2}P_{n-1}\frac{\frac{\partial}{\partial t_1}\tau_{n+1}}{\tau_{n+1}}
+\frac{\partial^2}{\partial t_1^2}P_{n}\right)\nonumber\\
=&n\alpha P_{n}+\alpha_1\frac{\partial}{\partial t_1}P_{n}+\alpha_2\left(2\frac{\partial}{\partial t_1}b_{n}P_{n}-2u_{n-1}u_n^2u_{n+1}P_{n-2}+2zu_n\frac{\partial}{\partial t_1}P_{n-1}\right.\nonumber\\
&\left.+2zu_nb_{n-1}P_{n-1}-2zu_nb_{n+1}P_{n-1}+\frac{\partial^2}{\partial t_1^2}P_{n}
\right),\label{pf44}
\end{align}
where 
$$u_n=\frac{\tau_{n+1}\tau_{n-1}}{\tau_{n}^2},\quad b_n=\frac{\partial}{\partial t_1}\log \tau_{n},$$ and 
$$\frac{\partial}{\partial t_1}b_n=u_n(b_{n+1}-b_{n-1}).$$
Furthermore, observe that $P_n$ satisfies the time evolution (following \cite[eq.(3.42)]{chang2018partial})
\begin{align}
\frac{\partial}{\partial t_1}P_{n}=-zP_{n}+P_{n+1}+(b_{n+1}-b_n)P_{n}-u_nu_{n+1}P_{n-1},\label{pf45}
\end{align}
from which we can get
\begin{align}
\frac{\partial^2}{\partial t_1^2}P_{n}
=&-z\frac{\partial}{\partial t_1}P_{n}+\frac{\partial}{\partial t_1}P_{n+1}+(\frac{\partial}{\partial t_1}b_{n+1}-\frac{\partial}{\partial t_1}b_n)P_{n}\nonumber\\
&+(b_{n+1}-b_n)\frac{\partial}{\partial t_1}P_{n}-\frac{\partial}{\partial t_1}(u_nu_{n+1}P_{n-1})\nonumber\\
=&P_{n+2}+(-2z-b_n+b_{n+2})P_{n+1}+((-z+b_{n+1}-b_n)^2-u_{n+1}u_{n+2}\nonumber\\
&-u_{n+1}(b_n-b_{n+2})+u_n(b_{n-1}-b_{n+1})-u_nu_{n+1})P_{n}\nonumber\\
&+u_nu_{n+1}(2z+b_n-b_{n+2})P_{n-1}+u_{n-1}u_n^2u_{n+1}P_{n-2}.\label{pf46}
\end{align}

The formula \eqref{7-2-25} is eventually confirmed by substituting \eqref{pf45} and \eqref{pf46} into \eqref{pf44} and therefore we complete the proof.
\end{proof}

\bibliographystyle{abbrv}

\end{document}